\documentclass[11pt]{article}
\usepackage[margin=1in]{geometry}
\usepackage[utf8]{inputenc} 
\usepackage[T1]{fontenc}    
\RequirePackage[colorlinks,citecolor=blue,urlcolor=blue]{hyperref}
\usepackage{hyperref}       
\usepackage{url}            
\usepackage{booktabs}       
\usepackage{amsfonts}       
\usepackage{nicefrac}       
\usepackage{microtype}      
\usepackage[numbers]{natbib}
\usepackage{times}
\usepackage{bbm}
\usepackage{amsmath}
\usepackage{amsthm}
\usepackage{amssymb}
\usepackage{xcolor}
\usepackage{graphicx}
\usepackage{subcaption}
\usepackage[linesnumbered,ruled,vlined]{algorithm2e}
\SetKwInput{KwInput}{Input}                
\SetKwInput{KwOutput}{Output}              

\SetCommentSty{mycommfont}

\usepackage{authblk}

\usepackage[capitalise,noabbrev]{cleveref}
\usepackage{makecell}
\usepackage{accents}
\usepackage[titletoc,title]{appendix}
\usepackage{xpatch}
\makeatletter
\AtBeginDocument{\xpatchcmd{\@thm}{\thm@headpunct{.}}{\thm@headpunct{}}{}{}}
\makeatother
\usepackage{titlesec}
\titlelabel{\thetitle.\quad}

\let\hat\widehat
\let\tilde\widetilde

\newtheorem{theorem}{Theorem}[section]
\newtheorem{lemma}[theorem]{Lemma}

\newtheorem{corollary}[theorem]{Corollary}

\newtheorem{proposition}[theorem]{Proposition}
\newtheorem{remark}[theorem]{Remark}
\newtheorem{fact}[theorem]{Fact}

\newtheorem{definition}[theorem]{Definition}

\usepackage{amsmath}
\DeclareMathOperator*{\argmax}{arg\,max}
\DeclareMathOperator*{\argmin}{arg\,min}
\DeclareMathOperator*{\esssup}{esssup}
\newcommand{\indep}{\rotatebox[origin=c]{90}{$\models$}}

\newcommand{\mep}{m}

\newcommand{\GLR}{\mathrm{GL}}
\newcommand{\KL}{\mathrm{KL}}

\newcommand{\SR}{\mathrm{SR}}
\newcommand{\ASR}{\mathrm{aSR}}

\newcommand{\MSR}{\mathrm{mSR}}

\newcommand{\CS}{\mathrm{CU}}
\newcommand{\ACS}{\mathrm{aCU}}

\newcommand{\MCS}{\mathrm{mCU}}

\newcommand{\e}{\mathrm{e}}

\newcommand{\op}{\mathrm{op}}

\newcommand{\n}{n}

\newcommand{\st}{N}

\title{E-detectors: a nonparametric framework for\\ sequential change detection}

\author[1]{Jaehyeok Shin\thanks{shinjaehyeok@google.com}}
\author[2]{Aaditya Ramdas\thanks{aramdas@cmu.edu}}
\author[2]{Alessandro Rinaldo\thanks{arinaldo@cmu.edu}}
\affil[1]{Google}
\affil[2]{Carnegie Mellon University }

\begin{document}
	\maketitle

 \begin{abstract}
  Sequential change detection is a classical problem with a variety of applications. However, the majority of prior work has been parametric, for example, focusing on exponential families. We develop a fundamentally new and general framework for  sequential change detection when the pre- and post-change distributions are nonparametrically specified (and thus  composite). Our procedures come with clean, nonasymptotic bounds on the average run length (frequency of false alarms). In certain nonparametric cases (like sub-Gaussian or sub-exponential), we also provide near-optimal bounds on the detection delay following a changepoint. The primary technical tool that we introduce is called an \emph{e-detector}, which is composed of sums of e-processes---a fundamental generalization of nonnegative supermartingales---that are started at consecutive times. We first introduce simple Shiryaev-Roberts and CUSUM-style e-detectors, and then show how to design their mixtures in order to achieve both statistical and computational efficiency. Our e-detector framework can be instantiated to recover classical likelihood-based procedures for parametric problems, as well as yielding the first change detection method for many nonparametric problems. As a running example, we tackle the problem of detecting changes in the mean of a bounded random variable without i.i.d.\ assumptions, with an application to tracking the performance of a basketball team over multiple seasons.
\end{abstract}

\section{Introduction}

Suppose we observe sequentially a stream of random variables $X_1,X_2, \dots$, whose marginal distributions may change at some unknown time, or {\it changepoint}, $\nu$. To take one concrete example that we generalize later, denote the data stream by $X_1, \dots, X_v \sim P_{\mu_0}$ and $X_{v+1}, X_{v+2}, \dots \sim P_{\mu}$ where $P_{\mu_0}$ and $P_{\mu}$ are some probability distributions with parameters $\mu_0, \mu \in \mathbb{R}$, respectively.  Let $\mathbb{P}_{\nu}$ and $\mathbb{E}_{\nu}$ denote probability and expectation, respectively, with respect to the distribution of the entire infinite data stream, when the change occurs at time $\nu$. If there is no change, we think of $\nu$ as being equal to $\infty$, and we let $\mathbb{P}_\infty$ and $\mathbb{E}_\infty$ refer to the corresponding probability and expectation.

 We are concerned with designing sequential changepoint detection procedures to determine, at each time, whether a changepoint has occurred in the near past that (i) provide non-asymptotic false alarm guarantees, (ii) allow for non-parametric classes of pre- and post-change distributions and (iii) are computationally efficient. Formally, a sequential changepoint detection algorithm consists of a data-dependent stopping rule $\st^* \geq 1$.  If not specified explicitly, the underlying filtration with respect to which $\st^*$ is defined is assumed  to be the natural filtration generated by the data stream $X_1, X_2, \cdots$, but in some cases it is beneficial to coarsen the filtration. If the stopping time $\st^*$ is finite, then we declare that a changepoint has been detected, in the sense that sufficient evidence has accumulated to support the hypothesis that the data generating distribution has changed.  If the algorithm never stops, we set $N^*=\infty$ and no changepoint is proclaimed. To evaluate the performances of detection algorithms, we can check how quickly the algorithm can detect a change in the distribution while controlling the frequency of false alarms. 

 To control false alarms, we adopt the average run length (ARL)  metric~\citep{page1954continuous}, defined by 
 \begin{equation}\label{eq:ARL}
 \text{ARL} := \mathbb{E}_\infty \st^*.
 \end{equation}
 We say that the ``ARL is controlled at level $\alpha$'' if
$\mathbb{E}_\infty\st^* \geq 1/\alpha$. 
An equivalent error metric is the  False Alarm Rate (FAR), which is the reciprocal of the ARL, and we would like to ensure that the FAR is at most $\alpha$, and we call a  sequential change detection procedure as ``valid'' if it satisfies the above constraint.

 A widely-used measure of the speed of detection after a changepoint is the worst average delays conditioned on the least favorable observations before the change \citep{lorden1971procedures}, or conditioned on the event that the algorithm stops after the changepoint \citep{pollak1975approximations}. These are defined by 
 \begin{align}
     \mathcal{J}_L(\st^*)&:= \sup_{\nu \geq 0} \esssup  \mathbb{E}_{\nu} \left[[\st^* - \nu]_+ \mid \mathcal{F}_\nu \right], \label{eq::worst_delay}\\
    \mathcal{J}_P(\st^*)&:= \sup_{\nu \geq 0} \mathbb{E}_{\nu} \left[\st^* - \nu \mid \st^* > \nu\right], \label{eq::ave_delay}
 \end{align}
 where the subscripts indicate the authors, 
  and it is known that $\mathcal{J}_P(\st^*) \leq \mathcal{J}_L(\st^*)$ \citep{moustakides2008sequential}.

 Our implicit objective is to (approximately) minimize $\mathcal{J}_L(\st^*)$ or $\mathcal{J}_P(\st^*)$ while guaranteeing that the ARL is controlled at a prespecified level $\alpha\in(0,1)$. We differ from other work in that our focus is on nonparametric and composite pre- and post-change distributions, as well as on deriving nonasymptotic guarantees. In several such settings, it is apriori unclear how to define \emph{any} valid sequential change detection algorithm, let alone an optimal one. Accordingly, we first address the design problem, and then move to  questions involving (approximate) optimality for detection delays.

 \subsection{Prior work and our contributions}
 
 If both pre- and post-change parameters $\mu_0, \mu$ are known and the distributions have densities $p_{\mu_0}$ and $p_\mu$ with respect to some common reference measure, then the CUSUM procedure \citep{page1954continuous} has been known to achieve the optimal worst average delay (exactly for $\mathcal{J}_L(\st^*)$ and asymptotically for $\mathcal{J}_P(\st^*)$ as $\alpha \to 0$) among all procedures controlling ARL at the same level \citep{lorden1971procedures, moustakides1986optimal, ritov1990decision, lai1998information}. Recall that the CUSUM procedure is defined by the stopping time
\(\st_{\CS}^* := \inf\left\{\n \geq 1: M^\CS_\n  \geq c^{\CS}_\alpha\right\},
\)
where $c^{\CS}_\alpha > 0$ is a constant chosen so that $\mathbb{E}_\infty(\st_{\CS}^*) = 1/\alpha$, and the test statistic $M^\CS_n$ is defined by the recursive formula
\begin{equation} \label{eq::CUSUM_formula}
    M^\CS_\n = \frac{p_{\mu}(X_\n)}{p_{\mu_0}(X_\n)} \cdot \max\left\{M^\CS_{\n-1}, 1\right\},~~M^\CS_0 := 0.
\end{equation}
The Shiryaev–Roberts (SR) procedure \citep{shiryaev1963optimum, roberts1966comparison} is defined by the stopping time
\( \st_{\SR}^* := \inf\left\{\n \geq 1: M^\SR_\n  \geq c^{\SR}_\alpha\right\},
\)
where $c_\alpha^\SR > 0$ is a constant chosen so that $\mathbb{E}_\infty(\st_{\SR}^*) = 1/\alpha$,  and the test statistic $M^\SR_n$ is obtained recursively as
\begin{equation} \label{eq::SR_formula}
    M^\SR_\n = \frac{p_{\mu}(X_\n)}{p_{\mu_0}(X_\n)} \cdot \left[M^\SR_{\n-1}+ 1\right],~~M^\SR_0 := 0.
\end{equation}
Unlike  CUSUM, the SR procedure  does not achieve exact minimax optimality for $\mathcal{J}_L(\st^*)$. However, the SR procedure and its generalized versions enjoy strong \emph{asymptotic} optimality guarantees \citep{pollak1985optimal, polunchenko2010optimality,tartakovsky2012third, tartakovsky2014sequential}.

The literature sometimes assumes that the pre-change distribution is known or can be approximated with a high precision by using the previous history. However, the post-change distribution is typically unknown and is assumed to belong to a family of distributions $\mathcal{P} := \left\{p_\mu : \mu \in \Theta\right\}$. In this case, one natural approach would be to replace the unknown parameter $\mu$ with an estimator $\hat{\mu}$. If we use the maximum likelihood estimator (MLE), then we obtain the CUSUM procedure based on the generalized likelihood ratio (GLR) rule \citep[e.g., see][]{barnard1959control,willsky1976generalized,siegmund1995using}.  For those not familiar with sequential change detection, \cite{lai2001sequential} provides a good overview.

\paragraph*{Limitations of prior work.} Though the usage of GLR statistic for the sequential change detection problem  has a long history and often yields good empirical performance, the current literature has two main limitations. 

First, most existing methodology relies on parametric assumptions on the family of distributions (eg: exponential families), both pre-change and post-change. There have been attempts to move away from this setting, and we will discuss these later. However, a general framework for deriving sequential change detection procedures in general nonparametric or composite settings has not been previously presented in the generality that we do here.

Second, the study of statistical properties has typically focused on the asymptotic regime of $\alpha \to 0$, unless the GLR statistic is defined on a well-separated post-change parameter space. In this paper, we guarantee nonasymptotic control on the ARL at a prespecified level (such as $\alpha=0.001$). In fact, in many settings considered, we do not know of any existing method to guarantee (even asymptotic) ARL control.

(We think that in theory and practice, the first is a bigger issue than the second. Luckily, our solution for the first automatically handles the second. Indeed, in the composite and nonparametric settings we consider, it is quite unclear how to control the ARL in any asymptotic sense.)

Finally, a direct online implementation of the GLR rule is infeasible since the memory and computation time both increase at least linearly as $\n \to \infty$. One natural approach to tackle this online implementation issue is to use window-limited versions of the GLR rule \citep{willsky1976generalized, lai1995sequential}. For instance, a simple form of the window-limited GLR rule can be defined by, at each time $n$, computing $\hat{\mu}$ over only times $n-W$ to $n$ for a properly chosen window size $W >0$. However, the optimal choice of window size $W$ has been studied only in the asymptotic setting ($\alpha \to 0$). For a fixed $\alpha$, the optimal window size depends on the difference between the pre- and post-change distributions, which is unknown.

\paragraph*{Our contributions.} We present a general framework for sequential change detection, focusing on (parametric and nonparametric) settings with composite pre- and post-change distributions, and nonasymptotic guarantees on the ARL:
\begin{enumerate}
    \item We introduce the concept of an \emph{e-detector} that underlies our construction of sequential change detection procedures.  The e-detector utilizes a generalization of the underlying martingale structure of likelihood ratios in classical sequential change detection procedures. This e-detector framework is applicable in nonparametric settings including sub-Gaussian, sub-exponential, and bounded random variables, among many others. In such settings, there is no common reference measure and likelihood ratios cannot be directly defined, thus composite nonnegative supermartingales, or more generally ``e-processes'', must be employed in their place. 
    \item Despite handling composite pre- and post-change distributions, even without an iid (independent and identically distributed) assumption on the data, our e-CUSUM and e-SR sequential change detection procedures based on e-detectors can always nonasymptotically control the ARL at level $\alpha$. 
    
    \item Nonasymptotic bounds on the worst average delay are derived in special cases for nonparametric distributions with exponential tail decay (such as sub-Gaussian or sub-exponential), and they match the rate of known lower bounds for exponential families as $\alpha \to 0$.
    \item Computationally feasible algorithms are presented, in order to run our procedures in an online fashion without windowing. Practical strategies to choose hyperparameters are discussed. These are based on an adaptive mixture method, with the number of mixture components growing slowly over time.
\end{enumerate}
 
Our procedures have natural gambling interpretations, and our work can be viewed as setting the foundations for a game-theoretic approach to sequential change detection. 
Before discussing the general framework in detail, in the following subsection, we present a motivating real-world example involving bounded random variables to illustrate how our nonparametric framework can be easily used in settings in which it is nontrivial to apply other common methods.

\subsection{Example: A changepoint in Cleveland Cavaliers 2011 - 2018} \label{subSec::intro_Cav_example}
The Cleveland Cavaliers are an American professional basketball team. We use the Cavaliers' game point records over 2010-11 to 2017-18 NBA seasons to illustrate how our proposed nonparametric sequential change detection algorithm can be applied to detect an interesting changepoint in the Cavaliers' recent history. \footnote{The R code to reproduce all the plots and simulation results of the paper is available at \url{https://github.com/shinjaehyeok/e_detector_paper}.}

    The left plot in \cref{fig::cavs_score} shows the difference between the scores of the Cavaliers and those of their opposing teams (also known as Plus-Minus) in all the games from the 2010-11 to the 2017-18 regular seasons.  Each red line refers to the yearly average difference score in each season. Roughly, this  value shows how well the Cavaliers performed against their opponent in terms of scoring. Typically, if a seasonal average is positive (or negative) then we may say that the Cavaliers showed a strong (or poor) performance in the corresponding season. The right plot shows a changepoint detected in early 2015; 
    NBA fans may recall one major cause of the sharp improvement --- LeBron James returned to the Cavaliers in 2014. However, how can we detect such a change on the fly by only tracking the Plus-Minus for each game?    
 
This type of question fits well into the sequential change detection framework. Let $X_1, X_2,\dots$ be the sequence of Plus-Minus stats we observe sequentially. After observing a poor performance of the Cavaliers in 2010-11 season, we define the Cavaliers' pre-change distribution on the Plus-Minus stats as follows: the average Plus-Minus of the team is less than or equal to $\mu_{0} := -1$. Now, we are interested in detecting a meaningful performance improvement on the fly by defining the post-change distribution as follows: the average Plus-Minus of the team is greater than $\mu_{1}:= 1$. Here, the gap $|\mu_1 - \mu_0|$ between averages of Plus-Minus in pre- and post-changes refers to the degree of improvement we consider as a significant one.

\begin{figure*}
    \begin{center}
    \includegraphics[width=0.5\textwidth]{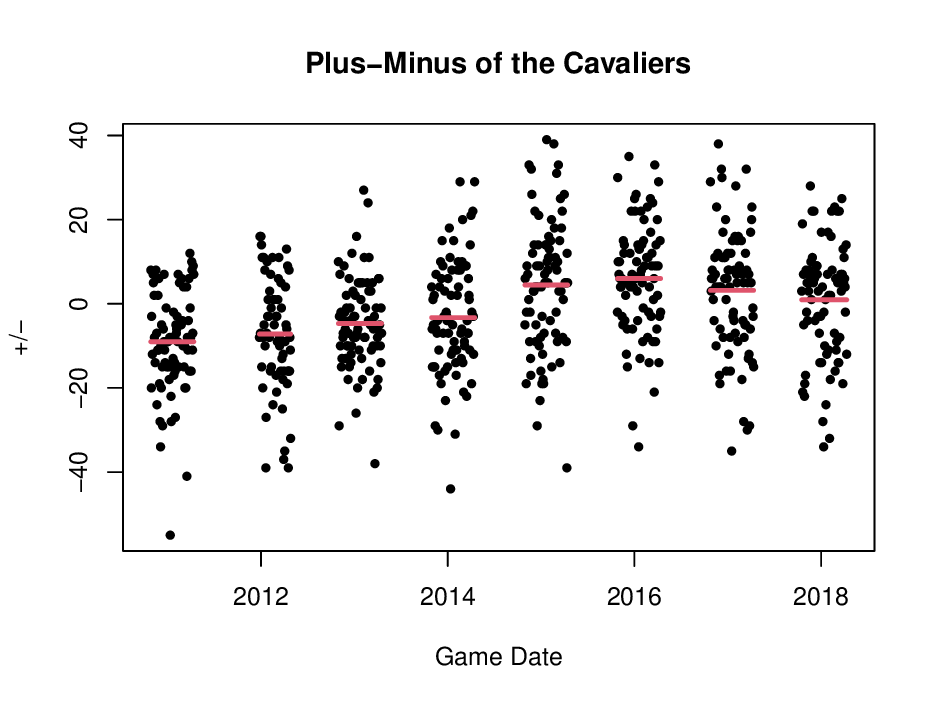}%
    \includegraphics[width=0.5\textwidth]{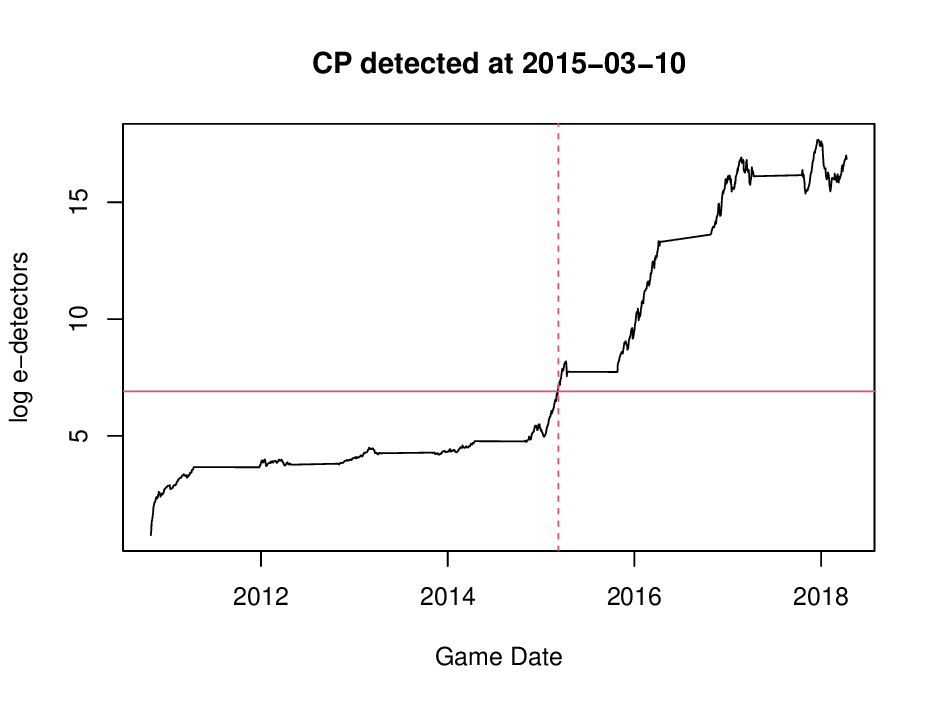}
    \end{center}
    \caption{\em Left: Plus-Minus of the Cavaliers from 2010-11 to 2017-18 seasons. Each horizontal red line corresponds to the seasonal average. Right: The sample path of (the logarithm of) one of our e-detectors. The horizontal red line is the threshold equal to $\log(1/\alpha)$ controlling the ARL by $1/\alpha = 1000$. In this example, the procedure detects the changepoint in the Plus-Minus of the Cavaliers at the end of 2014-15 season.}
    \label{fig::cavs_score}
\end{figure*}

Although the formulation of the problem is simple as described above, it is still nontrivial to fit this problem into commonly used sequential change detection procedures for the following reasons. First, it is not easy to choose a proper parametric model to fit observed Plus-Minus stats since they are integer-valued samples with varying mean and variance over seasons as \cref{fig::cavs_score} illustrates. Second, even if we can choose a proper model, it is difficult to find a threshold to detect the changepoint since we are interested in detecting any changes larger than $|\mu_1 - \mu_0|$ instead of a fixed post-change. Many common methods have been relying on a high-quality simulator or large enough sample history for pre-change observations to compute a valid threshold. For this example, however, it is hard to access such tools since the Cavaliers' overall Plus-Minus stats are difficult to model directly, and it is tricky to justify using existing records to get a valid threshold as the team's overall performance varies a lot around the 2010-11 season. 

Based on the introduced framework of the sequential change detection procedure using e-detectors, we detour difficulties in the commonly used methods illustrated above as follows. First, due to the nonparametric nature of the new framework, we do not need to choose any parametric model to fit the data. Instead, we simply assume the absolute value of each Plus-Minus stat is bounded by a large number --- in this example, we set $80$ as the boundary. Though we set a conservatively large boundary, our detection procedure is variance-adaptive so that we can detect the changepoint efficiently without specifying the variance of observations. Second, the nonasymptotic analysis of the new framework makes it possible to choose an explicit detection boundary, which is equal to $\log(1/\alpha)$, to build a sequential change detection procedure controlling the ARL by $1/\alpha$ for any given $\alpha \in (0,1)$. In this example, we choose $\alpha = 10^{-3}$ to make ARL is larger than at least 10 regular seasons of games. 
The right plot in \cref{fig::cavs_score} shows the log of e-detectors on which we build the sequential change detection procedure. The red horizontal line corresponds to the detection boundary given by $\log(1/\alpha)$. We can check the procedure detects the changepoint of the Plus-Minus stat of the Cavaliers in the middle of the 2014-15 season. See \cref{subSec::bounded_rv} for the detailed explanation about how we construct the sequential change detection procedure based on the general methodology we introduce in the paper.

\paragraph*{Paper outline.} The rest of the paper is organized as follows. In \cref{sec::general}, we introduce a general framework about how to build composite, nonparametric, and nonasymptotic sequential change detection procedures using e-detectors. \cref{sec::mixture_general} extends the previous framework to the case where we have a set of e-detectors and explain how to use a mixture method to combine multiple e-detectors effectively. In \cref{sec::mixture_simple_exponential}, we introduce an exponential structure of e-detectors that makes it possible to design a near-optimal detection procedure with an explicit upper bound on worst average delays. Based on the proposed framework, \cref{sec::examples} presents two canonical examples of Bernoulli (parametric) and bounded random variables (nonparametric) cases with real data applications of the Cavaliers 2011-2018 statistics. We conclude with a discussion, and defer proofs to the supplement.

\section{Nonparametric sequential change detection using e-detectors} \label{sec::general}

\subsection{Problem Setup}
    Let  $\mathcal{P}$ denote the set of possible pre-change distributions, which could in general be a nonparametric class.  We do not assume the observations in the sequence to be independent or identically distributed: each $P \in \mathcal P$ is a distribution over an infinite sequence of random variables.

    We will assume throughout that up to the unknown changepoint $\nu$, the observations $X_1, \dots, X_\nu$  follow a distribution $P \in \mathcal P$. The remaining  observations $X_{\nu+1}, X_{\nu+2}, \dots$ are drawn from a distribution $Q$ in a class of post-change distributions $\mathcal{Q}$. 
   In this case, we let $\mathbb{P}_{P,\nu,Q}, \mathbb{E}_{P,\nu,Q}$ and $\mathbb{V}_{P,\nu,Q}$ denote probability, expectation and variance operators over the entire data stream.

   If there never is a change, we will use the notation $\mathbb{P}_{P,\infty}, \mathbb{E}_{P,\infty}$ and $\mathbb{V}_{P, \infty}$. Also, if a change occurs at the beginning ($\nu =0)$ then we use $\mathbb{P}_{0, Q}, \mathbb{E}_{0, Q}$ and $\mathbb{V}_{0, Q}$. Note that technically $\mathbb{P}_{0, Q} = \mathbb{P}_{Q,\infty}$, but we use the former to denote that $Q$ is a post-change distribution and a changepoint has occurred at the very start, and the latter to denote that $Q$ is a pre-change distribution and a changepoint never occurs.
  
Let  $\mathcal{F}:= \{\mathcal{F}_n\}_{n \geq 0}$ be a filtration where we let $\mathcal{F}_0 := \{\emptyset, \Omega\}$ for simplicity. Let $M := \{M_n\}_{n \geq 0}$ be a nonnegative adapted process  with respect to the filtration $\mathcal{F}$. If required, we define $M_\infty := \limsup_{n \to \infty} M_n$ and $\mathcal F_\infty:=\sigma(\bigcup_{n \geq 0} \mathcal F_n)$ \citep[see, for example,][]{durrett2019probability}. It is common to consider the natural filtration $\mathcal{F}_n := \sigma(X_1,\dots, X_n)$, but there are situations where restricting the filtration could be advantageous (for example, when there are nuisance parameters). There are yet other situations when enlarging the filtration with external randomness can be useful.

Let $\mathcal{T}$ denote the set of all stopping times with respect to $\mathcal{F}$, but we will later see that it typically suffices to consider finite stopping times, or those with finite expectation. 

  \begin{remark} 
  In our paper, the changepoint $\nu$ and the post-change observations also do not need to be independent of the pre-change observations, and they do not need to be identically distributed. In other words, $\nu$ could be a stopping time, and $Q$ could itself depend on the pre-change data. It may be helpful to imagine an adversary who adaptively decides at each step whether or not to ``stop'' the pre-change data. If they choose to stop at time $\nu$, there are two options for post-change points: (a) then can pick a distribution $Q$, draw a sequence $Y_1,Y_2,\dots,$ from $Q$, and reveal $X_{\nu+i}=Y_i$ sequentially, or (b) the pick a distribution $Q$ and draw $X_{\nu+1},X_{\nu+2},\dots$ from $Q \mid \mathcal F_\nu$. In situations without a common reference measure, setting (b) may be tricky to formally define (especially if the pre-change data have probability zero under $Q$), so one may think of setting (a).
  All of our results on ARL do not require any assumptions on the changepoint $\nu$ or post-change distribution of the data. When analyzing detection delay, we will typically assume that $X_{\nu+1},X_{\nu+2},\dots$ are independent of the pre-change data, and are drawn from $Q$ as if the time $\nu$ was reset to zero; we believe this can be relaxed in future work.
  But for much of the paper, it is also okay to assume that the post-change data is drawn a data-dependent $Q$ (or drawn from $Q \mid \mathcal F_{\nu}$).
  \end{remark}
  
  With the appropriate definitions and setup in place, we can now define our central concept, an e-detector.

\subsection{What is an e-detector?}

  \begin{definition}[$\mathcal{P}$-$\e$-detector]
  The process $M$ is called an \emph{e-detector} with respect to the class of pre-change  distributions $\mathcal{P}$ if it satisfies the property
\begin{equation} \label{eq::e_detector_cond}
 \mathbb{E}_{P, \infty}\left[M_{\tau} \right]\leq \mathbb{E}_{P, \infty}\left[\tau\right] ,~~\forall \tau \in \mathcal{T}, ~~\forall P \in \mathcal{P}.
\end{equation}
  \end{definition}
  For brevity, we refer to $M$ as ``an e-detector for $\mathcal P$'', or if $\mathcal P$ is understood from context, then simply ``an e-detector''.
   If a stopping time $\tau$ has a nonzero probability of being infinite, then  inequality~\eqref{eq::e_detector_cond} is trivially satisfied. Thus, the condition is only really required to hold for stopping times with finite expectation under some $P$. This latter set of stopping times depends on $\mathcal P$, and so in order to not complicate notation, we continue to simply consider all stopping times $\mathcal T$. (Also note that the condition can only be satisfied if process $M$ is integrable under any $P \in \mathcal P$, so this is implicitly assumed to be the case in what follows.)

   By the linearity of expectation and Tonelli's theorem, an average (or ``a mixture'') of e-detectors is also an e-detector. More formally, if $\{M^a\}_{a \in A}$ is a set of e-detectors (where $a$ is a tuning parameter), then so is $\int M^a d\mu(a)$ for any fixed probability distribution $\mu$ over $A$. Later in this paper, we will in particular use finite mixtures of the form $(M^1 + M^2 + \dots + M^K)/K$ in order to adapt to the unknown post-change distribution. (In fact, we will develop more sophisticated mixtures whose number of components grows slowly with time.) For later reference, we state the above as a proposition:
  
  \begin{proposition} \label{prop::mixture_is_e_detector}
     Let $\{M^a\}_{a \in A}$ be a set of e-detectors. Then for any probability measure $\mu$ on $A$, the mixture of e-detectors, $\int M^a d\mu(a)$ forms a valid e-detector.
  \end{proposition}
  
 An e-detector provides a quantification  of evidence for whether a changepoint has occurred or not, and may be continuously monitored, stopped and easily interpreted -- e.g, a steep and steady increase of the process in recent times should be taken as an indication that a change has taken place. 
   The following theorem shows how one can immediately obtain a sequential change detection procedure from an e-detector $M$. 
   
  \begin{theorem}\label{prop::ARL_control_e_detector}
  For any $\alpha \in (0,1)$ and e-detector $M$, if we declare a changepoint at the stopping time
  \begin{equation}\label{eq::basic_cp_st}
      N^* := \inf \left\{n \geq 1 : M_n \geq 1/\alpha\right\},
  \end{equation}
  then we have
  \begin{equation}
     \inf_{P \in \mathcal{P}}\mathbb{E}_{P,\infty} N^* \geq 1/\alpha,
  \end{equation}
 That is, the sequential change detection procedure in \eqref{eq::basic_cp_st} controls the ARL at level $\alpha$.
    \end{theorem} 
The informal proof is one line long: dropping subscripts, the definition of an e-detector implies that $\mathbb{E}N^* \geq \mathbb{E}M_{N^*}$, but by definition of $N^*$, we know that $M_{N*}\geq 1/\alpha$ if $N^*$ is finite. (If $N^*$ is not almost surely finite, the claim holds anyway.)
The full proof is   in~\cref{appen::proofs_for_general}. 

   
   \subsection{Constructing an e-detector based on a sequence of e-processes} \label{subSec::witness_based_on_e_value}
 
   The central building block of our e-detector is called an e-process. E-processes are newly-developed tools that have been shown to play a  fundamental role in sequential hypothesis testing, especially in composite, nonparametric settings. E-processes are generalizations of nonnegative martingales and supermartingales, and in particular, \emph{e-processes are nonparametric and composite generalizations of likelihood ratios}. They have strong game-theoretic roots, and have found utility in the meta-analysis, as well as for the purposes of anytime-valid inference in the presence of continuous monitoring \citep{ramdas2020admissible,ramdas2021testing,grunwald2019safe,howard2020time,howard2021time}. The properties of e-processes have not yet been explored in changepoint analysis, and we undertake this effort here.
   
   To understand their definition, we briefly forget about sequential change detection and consider testing the null hypothesis that $X_1, X_2, \dots \sim P$ for some $P \in \mathcal{P}$. An e-process for $\mathcal{P}$, called a $\mathcal P$-e-process, is a  sequence of nonnegative random variables $(E_t)_{t\geq 1}$ such that for any $P \in \mathcal{P}$ and any stopping time $\tau$, we have $\mathbb{E}_P[E_\tau] \leq 1$. As before, the underlying filtration can be that of the data or that of $(E_t)$, or some enlargement of these. The value of $E_t$ measures evidence against the null (larger values, more evidence). A level-$\alpha$ sequential test can be obtained by rejecting the null as soon as $E_t$ exceeds $1/\alpha$; this is a consequence of Ville's inequality~\citep{ville1939study,howard2020time}. 
   
   As a result of the optional stopping theorem, nonnegative $\mathcal{P}$-martingales (i.e., the process is a nonnegative $P$-martingales simultaneously for every $P \in \mathcal{P}$) and $\mathcal{P}$-supermartingales are examples of e-processes. However, e-processes are a distinct and more general class of processes. In fact, there exist natural classes $\mathcal{P}$ for which the only $\mathcal{P}$-martingales are constants, and the only $\mathcal{P}$-supermartingales are decreasing sequences, but there are e-processes for $\mathcal{P}$ that can increase to infinity when the data are not from $\mathcal{P}$. See \cite{ramdas2021testing} for one such example, arising from  sequentially testing exchangeability of a binary sequence, and \cite{gangrade2023sequential} for another example arising from testing log-concavity. 
   
   In this subsection, we show how to leverage e-processes to build e-detectors. 
\begin{definition}[$\e_j$-process]
      For any $j \geq 1$,  $\Lambda^{(j)}:= \{\Lambda_n^{(j)}\}_{n \geq 1}$ is called an $\e_j$-process for $\mathcal{P}$ if it is a nonnegative adapted process such that $\Lambda_{1}^{(j)}= \cdots =\Lambda_{j-1}^{(j)} = 1$, and
   \begin{equation} \label{eq::e_value_def}
      \sup_{P\in\mathcal{P}}\mathbb{E}_{P, \infty}\left[\Lambda_{\tau}^{(j)} \mid \mathcal{F}_{j-1}\right] \leq 1,~~\forall \tau \in \mathcal{T}.
  \end{equation}
\end{definition}

 When $j = 1$, the $\e_j$-process is simply a standard e-process, which tests whether the data distribution is different from the proclaimed null (pre-change) distribution. For $j > 1$, each $\e_j$-process can be viewed as an e-process that begins at time $j$, and tests whether the data occurring after time $j$ are well explained by the null hypothesis.

\begin{definition}[SR and CUSUM e-detectors]
   Based on a sequence of e-processes $\{\Lambda^{(j)}\}_{j \geq 1}$, define SR and CUSUM e-detectors $M^\SR$ and $M^\CS$, respectively by $M_0^\SR = M_0^\CS:= 0$ and for each $n \geq 1$,
   \begin{align} \label{eq:SER:CS}
       M_n^\SR &:= \sum_{j=1}^n \Lambda_{n}^{(j)}, \text{~ and ~}~
       M_n^\CS := \max_{j \in [n]} \Lambda_{n}^{(j)}.
   \end{align}
   \end{definition}
 It is not hard to check that the above processes are indeed e-detectors, meaning they satisfy~\eqref{eq::e_detector_cond}.

 \begin{remark}
    If each e-process starts with an initial weight smaller than 1 such that the sum of all initial weights is less than or equal to 1, then corresponding SR and CUSUM procedures control the probability of the false alarm, $\sup_{P \in \mathcal P} \mathbb{P}_{P,\infty}(\tau <\infty) \leq \alpha$, which is a much more stringent error metric than the ARL. The price to pay is that if there is a change at time $\nu$, then the detection delay will no longer be independent of $\nu$, and will typically increase logarithmically with $\nu$ (so that the worst average delay is unbounded).
\end{remark}

\subsection{Constructing computationally efficient e-detectors using baseline increments}

 In general, it may take $O(n)$ time to update the aforementioned e-detectors at time $n$. In order to construct e-detectors that can be updated online in sublinear time and memory (or even near-constant time and memory), it turns out to be computationally convenient to use a common ``baseline'' increment in order to build the underlying $\e_j$-processes, as we do below. Effectively, this amounts to using $\e_j$-processes that are $\mathcal P$-supermartingales, which is a special case of particular interest.

  \begin{definition}[Baseline increment] \label{def::baseline_increment}
     A nonnegative, adapted process $L:=\{L_n\}_{n \geq 1}$ is called a baseline increment if for each $n \geq 1$, we have
  \begin{equation} \label{eq::baseline_increment}
      \sup_{P\in\mathcal{P}}\mathbb{E}_{P, \infty}\left[L_n \mid \mathcal{F}_{n-1}\right] \leq 1.
  \end{equation}
  \end{definition}
It is easy to check that if $L^1$ and $L^2$ are baseline increments, and $A^1$ and $A^2$ are nonnegative and predictable processes (meaning that $A^1_n$ and $A^2_n$ are both $\mathcal F_{n-1}$-measurable) such that $A^1 + A^2$ is strictly positive, then the mixture $(A^1 L^1 + A^2 L^2)/(A^1 + A^2)$ also forms a baseline increment. In short, ``predictable mixtures'' retain the baseline increment property.

Comparing \eqref{eq::baseline_increment} to \eqref{eq::e_value_def}, we see that a baseline increment $L$ is not itself an $\e_j$-process, because the expectation in \eqref{eq::baseline_increment} applies only at fixed times $n$, with the conditioning being on the previous step $n-1$, but \eqref{eq::e_value_def} calculates expectations at any stopping time, and conditions on $j-1$. It is best to think about the baseline increment as the multiplicative increment that forms the $\e_j$-process, as follows.

\begin{definition}[Baseline $\e_j$-process] For a given baseline increment $L:= \{L_n\}_{n \geq 1}$, we define the corresponding ``baseline $\e_j$-process'' $\Lambda^{(j)}$,  for each $j, n \in \mathbb{N}$, as below:
\begin{equation} \label{eq::baseline_e_value}
    \Lambda_n^{(j)} := \begin{cases} 
                1 &\mbox{if } n < j  \\
                \prod_{i = j}^{n}L_i & \mbox{otherwise},
            \end{cases}
\end{equation}
\end{definition}
 Under any pre-change distribution $P \in \mathcal{P}$, each $\Lambda^{(j)}$ is a nonnegative supermartingale by definition of the baseline increment $L_i$. Therefore, a straightforward application of the optional stopping theorem implies that  each  $\Lambda^{(j)}$ satisfies condition~\eqref{eq::e_value_def}, and thus is a valid $\e_j$-process. 

As an example of a baseline $\e_j$-process, consider the case where we have iid observations $X_1, X_2, \dots$ from a distribution $p_\theta$ parameterized by $\theta \in \Theta$, and the pre-change distribution is given by $\theta_0$. Then, for any post-change distribution $p_{\theta_1}$ with $\theta_1 \neq \theta_0$, the likelihood ratio between two distributions, $L_n := p_{\theta_1}(X_n) / p_{\theta_0}(X_n)$ yields a baseline increment process with the inequality in \eqref{eq::baseline_increment} being replaced by the equality. Then, each $\Lambda_n^{(j)}$ is the likelihood ratio based on $X_j, \dots, X_n$. Further, instead of using a fixed post-change parameter $\theta_1$, we can also plug-in a running MLE or any other online nonanticipating estimator that is based on the previous history $\mathcal{F}_{n-1}$ only,  say $\hat{\theta}_{n-1}$, into the likelihood ratio. In this case, although the value   $L_n := p_{\hat{\theta}_{n-1}}(X_n) / p_{\theta_0}(X_n)$ at time $n$ of the resulting process may depend on the previous history $\mathcal{F}_{n-1}$, the inequality in \eqref{eq::baseline_e_value} will be satisfied as an equality, yelling a valid baseline $\e_j$-processes.

 \begin{remark}\label{rem:baseline}
While baseline increment processes provide a natural and computationally convenient way to construct $\e_j$-process, we emphasize that any e-detector, even one that does not use baseline increments, will automatically control the ARL by  \Cref{prop::ARL_control_e_detector}.  To elaborate, baseline $\e_j$-processes are composite $\mathcal P$-supermartingales (meaning $P$-supermartingales for every $P\in\mathcal P$), but there exist other $\mathcal P$-e-processes---which are not $\mathcal P$-supermartingales---that naturally arise and these can be used to form e-detectors; for example using universal inference~\cite[Section 8]{wasserman2020universal}. 
\end{remark}

\begin{definition}[Baseline SR and CUSUM e-detectors]
   When an SR or CUSUM e-detector is constructed using a sequence of baseline $\e_j$-processes~\eqref{eq::baseline_e_value}, we call it a ``baseline SR or CUSUM e-detector''. 
\end{definition}
   Each baseline SR or CUSUM e-detector can be computed recursively like their classical analogs:
\begin{align}
       M_n^\SR &= L_n \cdot  \left[M_{n-1}^\SR + 1 \right],\\
      M_n^\CS &=  L_n \cdot \max \left\{M_{n-1}^\CS , 1 \right\},
   \end{align}
with $M_0^\SR = M_0^\CS  = 0$ for each $n \in \mathbb{N}$. 
The above computational benefit is the primary reason to consider baseline e-detectors, but as mentioned in Remark~\ref{rem:baseline} and when introducing e-processes, more general e-detectors are sometimes necessary for certain classes $\mathcal P$.

  We briefly verify below that the processes $M^\SR$ and $M^\CS$ defined above are valid e-detectors. Indeed, for any stopping time $\tau$ and pre-change distribution $P \in \mathcal{P}$, if $\mathbb{P}_{P,\infty}(\tau = \infty) > 0$ then the condition of the e-detector in \eqref{eq::e_detector_cond} holds trivially. If not, then $\tau$ is finite almost surely, and we have by linearity of expectation and the tower rule:
   \begin{align*}
       \mathbb{E}_{P,\infty} M_{\tau}^\CS \leq \mathbb{E}_{P,\infty} M_{\tau}^\SR &= \mathbb{E}_{P,\infty} \sum_{j=1}^\infty \Lambda_{\tau}^{(j)}\mathbbm{1}(j \leq \tau) 
       = \sum_{j=1}^\infty \mathbb{E}_{P,\infty} \Lambda_{\tau}^{(j)}\mathbbm{1}(j \leq \tau) \\
      &=\sum_{j=1}^\infty \mathbb{E}_{P,\infty} \left[\mathbbm{1}(j \leq \tau)\mathbb{E}_{P,\infty}\left[\Lambda_{\tau}^{(j)} \mid \mathcal{F}_{j-1}\right]\right] \leq\sum_{j=1}^\infty \mathbb{E}_{P,\infty} \mathbbm{1}(j \leq \tau) 
       = \mathbb{E}_{P,\infty} \tau,
   \end{align*}

   where the first inequality comes from the nonnegativity of e-processes, and the second inequality comes from the definition of the e-process $\Lambda^{(j)}$ for each $j \geq 1$. Note that this proof is also applicable to general SR and CUSUM e-detectors.

   \subsection{Sequential change detection procedures by thresholding e-detectors} \label{subSec::e-CP_procedure}
   
      The value of any e-detector process, like $M^\SR$  or  $M^\CS$, is directly interpretable without specifying an explicit threshold: a larger value signals an accumulation of evidence of a changepoint. These can be monitored and adaptively stopped. Nevertheless, to explicitly control the ARL at level $\alpha$,
   we define SR and CUSUM-style  change detection procedures, called ``e-SR'' and ''e-CUSUM'' procedures as follows.

   \begin{definition}[e-SR and e-CUSUM procedures]
         Given SR and CUSUM e-detectors $M^\SR$  and  $M^\CS$ , define e-SR and e-CUSUM procedures by the stopping times
    \begin{align}
        N_\SR^* &:= \inf \left\{ n \geq 1 : M_n^\SR  \geq 1/\alpha \right\}, \\
        N_\CS^* &:= \inf \left\{ n \geq 1 : M_n^\CS \geq  c_\alpha\right\}, \label{eq::CUSUM_e_procedure_st}
    \end{align}
    where $c_\alpha$ is a constant chosen to control the ARL of the e-CUSUM procedure by $1/\alpha$ for some $\alpha \in (0,1)$. By \cref{prop::ARL_control_e_detector}, $1/\alpha$ is a valid choice for $c_\alpha$. 
   \end{definition}
   
    We note $c_\alpha=1/\alpha$ may be a very conservative choice for the e-CUSUM procedure. Indeed, suppose we use the trivial e-processes, that is, we set $\Lambda_{n}^{(j)}:=1$ for all $j, n$. In this setting, the SR e-detector is given by $M_n^\SR = n$ for each $n$. In contrast, the CUSUM e-detector, $M_n^\CS$ is equal to 1 for all $n$. Therefore, any valid threshold we can choose for the e-SR procedure must be larger than $\lfloor 1/\alpha \rfloor$. On the other hand, any threshold above $1$ makes $N_\CS^* = \infty$, which of course controls ARL by $1/\alpha$, but the true ARL is much above the target. Building from this trivial example, it is possible to construct nontrivial examples in which letting $\alpha \to 0$ makes the gap between tight thresholds of e-SR and e-CUSUM procedures arbitrarily different. 
     
     \begin{remark} \label{rmk::remark_on_cutoff}
     Unless we assume the pre-change distribution is time-stationary, known and parametric, or we can access a good sample of the pre-change distribution or large enough historical data, computing a tight or even valid threshold $c_\alpha$ can be a challenging task. In the application sections below, we will mainly deal with non-stationary pre-change distributions where pre-change observations may not be identically distributed and thus all observations before the changepoint may follow different distributions. In this case, setting $c_\alpha = 1/\alpha$ seems to be the only reasonable choice, and we recommend using the e-SR procedure rather than e-CUSUM since if we use the same threshold for both procedures, the former always detects the changepoint faster than the latter while provably controlling the ARL at the same level. 
     \end{remark}

    \subsection{Some nontrivial instantiations of e-detectors}

    E-detectors can be thought of as a general reduction of change detection to sequential testing. Given the recent advances in nonparametric, composite sequential testing using nonnegative supermartingales and e-processes, our e-detectors now make new classes of change detection problems possible.
    We detail below some interesting nontrivial examples of change detection problems that can now be solved using e-detectors. 
     
    \paragraph*{Example 1: when likelihood ratios are well-defined.} Consider first the parametric case, when $\mathcal {P,Q}$ have a common reference measure and likelihood ratios are well-defined. When the pre- and post-change distributions are known (meaning $\mathcal {P,Q}$ are singletons), then the standard likelihood-ratio based CUSUM and SR processes from \eqref{eq::CUSUM_formula} and~\eqref{eq::SR_formula} are both e-detectors.  If $\mathcal Q$ is composite, then taking a mixture likelihood ratio yields e-detectors (using either a non-anticipating predictable mixture or using a fixed mixture distribution). If $\mathcal P$ is also composite, but maximum likelihood estimation is efficient over $\mathcal P$, then e-processes can be constructed that take the ratio of mixture likelihoods over the alternative to maximum likelihood under the null, as done in universal inference~\citep[Section 8]{wasserman2020universal} or in~\cite{tartakovsky2014nearly}. If the ``reverse information projection'' (RIPr) is computable (analytically or numerically), then one can use the method of~\cite{grunwald2019safe}, though there are some subtleties: by default the method produces e-values over blocks of observations, which can be multiplied across independent blocks to produce a nonnegative supermartingale (and thus e-process) to be used within our framework. But sometimes, the sequence of RIPr's over increasing sample sizes (nested blocks) automatically produces an e-process, and when this happens, it is more powerful than universal inference (see~\cite{grunwald2019safe,ramdas2022game} for details).
    
    \paragraph*{Example 2: change in distribution.} 
    In this example, $\mathcal P$ is the set of all iid\ product distributions over infinite sequences (or its convex closure, the set of all exchangeable distributions), so $\mathcal P = \{\mu^\infty $ for some probability distribution $\mu\}$.
    The conformal sequential change detection procedures by~\cite{vovk2005algorithmic,vovk2021testing} are designed to test deviations from  exchangeability, meaning that they develop a test martingale (and thus an e-process) for $\mathcal P$, meaning that their procedure fits neatly into our framework.
    Importantly, their filtration is restricted, and is smaller than the natural filtration of the data. This allows nonparametric martingales to exist, but the e-detector property only holds with respect to a smaller class of stopping times. Nevertheless, thresholding our e-detector at $1/\alpha$ still controls the ARL at level $\alpha$, as the latter property is independent of the filtration used to construct the e-detector. As a side note, if one wanted to construct an e-detector that was valid at \emph{all} stopping times with respect to the natural filtration of the data, e-detectors based on martingales provably do not suffice, but e-detectors based on e-processes can be constructed using the techniques from~\cite{ramdas2021testing} (at least for categorical distributions).

    \paragraph*{Example 3: nonparametric two-sample testing.} Suppose we have two streams of (general multivariate) data: $X_1,X_2,\dots \sim P_X$ and $Y_1,Y_2,\dots \sim P_Y$. For simplicity below, assume that at time $t$, we observe one point from each stream ($X_t,Y_t$). Before the changepoint (if one exists), the distributions of $X_t$ and $Y_t$ are equal, meaning that $\mathcal P=\{(P_X,P_Y)^\infty: P_X = P_Y\}$. This is a very nonparametric class, since it specifies nothing about the distributions except for the fact that they are equal before the changepoint. After the changepoint, the streams have different distributions (maybe the distribution of $X$ changes, or that of $Y$ changes, or both), thus $\mathcal Q = \{(P_X,P_Y)^\infty: P_X \neq P_Y \}$. For this very general nonparametric two-sample testing setup, \cite{shekhar2021game} construct test martingales for $\mathcal P$ that are provably consistent against $\mathcal Q$ under minimal assumptions (in particular not requiring any minimum separation between the different distributions after the changepoint, since the tests automatically adapt to the closeness of the unknown alternative). These test martingales fit seamlessly into our e-detector framework, yielding new and practicable e-detectors for detecting a change from homogeneity to non-homogeneity between the streams.

    \paragraph*{Example 4: nonparametric independence testing.} In this problem setting, we observe a pair of random variables $(X_t,Y_t) \sim P_{XY}$ at each step $t$, where $X_t,Y_t$ can each lie in a general space. Before the changepoint (if one exists), the data are independent, meaning that $\mathcal P = \{P_{XY}^\infty: P_{XY} = P_X \times P_Y\}$. Beyond saying that the joint distribution factorizes into the product of marginals, there is no further structure assumed, making this a rich nonparametric composite class. As before, after the changepoint, $(X,Y)$ become dependent, meaning that  $\mathcal Q = \{P_{XY}^\infty: P_{XY} \neq P_X \times P_Y\}$. For this general nonparametric independence testing problem, \cite{podkopaev2022sequential} construct test martingales for $\mathcal P$ that are provably consistent against $\mathcal Q$ under minimal, weak assumptions (as before, not requiring any separation). Again as before, (in the testing problem) the power of these tests automatically adapts to the difficulty of the unknown alternative.  When plugged into our framework, it delivers a novel e-detector for a change from independence to dependence.

\smallskip
    We briefly remark that in the two preceding examples (homogeneity and independence), we can move past the iid\ assumption. The same methods work even when the distribution is allowed to drift within $\mathcal P$ before the changepoint, and drift within $\mathcal Q$ after the changepoint. We refer to the original aforementioned papers for details.

    \paragraph*{Example 5: log-concavity.} Here, the data before the changepoint comes from a log-concave distribution (in a general dimension $d \geq 1$), so $\mathcal P =\{\mu^\infty: \mu \text{ has a log-concave Lebesgue density}\}$. This is a rich, nonparametric, shape-constrained class.  The post-change class of distributions $\mathcal Q$ consists of, for example, any distribution that has a nonzero KL-divergence and Hellinger distance from every distribution in $\mathcal P$. For testing $\mathcal P$ against $\mathcal Q$, \cite{gangrade2023sequential} show that there exists no nontrivial nonnegative supermartingales, but they design a powerful (and computationally efficient) e-process using universal inference. When plugged into our e-detector, this yields a nontrivial procedure that can detect a deviation from log-concavity.

    \paragraph*{Example 6: symmetry.} Suppose $\mathcal P=\{\mu^\infty: \mu \text{ is symmetric around 0}\}$ consists of the set of all distributions (in a general dimension $d \geq 1$) that are symmetric around the origin, while $\mathcal Q$ consists its complement (that is, distributions which are not symmetric around the origin). \cite{ramdas2020admissible} characterize all processes that are nonnegative martingales for $\mathcal P$. When used with our e-detector, these provide a clean way to detect a change from symmetry to asymmetry.

    \paragraph*{Example 7: change in mean.}  Suppose $\mathcal P^C = \{\mu^\infty:  \mathbb E_{X \sim \mu}[X] \leq 0, X \text{ satisfies }C\}$ consists of the set of all univariate distributions with mean less than or equal to zero satisfying come constraint $C$, while $\mathcal Q^C = \{\mu^\infty:  \mathbb E_{X \sim \mu}[X] > 0, X \text{ satisfies }C\}$ consists of those with positive mean. \cite{howard2020time} provides a large variety of nonnegative supermartingales under various conditions $C$, such as when $X$ are subGaussian, or bounded from above, or bounded from below, or have only two or three moments; see also~\cite{wang2022catoni} for heavy-tailed supermartingales. These can be plugged into the e-detector to yield new nonparametric schemes for changes in mean.

    \paragraph*{Example 8: Huber-robust change detection.} As a final example, suppose we wish to detect a change in mean of heavy-tailed data (as above). But now, suppose that an adversary can also arbitrarily corrupt an $\epsilon$ fraction of the data. \cite{wang2023huber} develop Huber-robust supermartingales for this setting, which can be plugged into e-detectors to yield a valid e-detector in the presence of adversarial corruptions.

 \bigskip 
   
   Note that in Examples 3, 4, 5 and 6, if $\mathcal P$ and $\mathcal Q$ were swapped, the testing problem is much harder, and we are not aware of any powerful test or change detection method.
However, if one was simply interested in detecting a change \emph{in homogeneity} or \emph{in dependence}, i.e.\ from some distribution in $\mathcal P \bigcup \mathcal Q$ to some other one, there are two possible change detection methods that come to mind. First, an e-detector based on the conformal change detection methods in Example 2 would detect any change from any distribution to any other, though choosing the conformity score may be tricky. As a second and more direct option, one can choose a measure of homogeneity or dependence (like the kernel maximum mean discrepancy or energy distance, or the Hilbert Schmidt independence criterion or distance covariance), construct a confidence sequence for that measure  (see~\cite{manole2023martingale} for several specific, tight, constructions), and plug it into the recent change detection scheme of~\cite{shekhar2023sequential}.

    Finally, it is worth noting that it is possible to define e-detectors in cases where there is no common reference measure amongst the pre-change and post-change distributions, and thus no easily-defined likelihood ratio process, and also when there is no nontrivial martingale that can be constructed. This is precisely the utility of e-processes, which are nonparametric and composite generalizations of likelihood ratios. We develop one more interesting nonparametric example (not described above) in the simulations section: detecting change in mean of a bounded random variable.

    \subsection{Warm-up: bounds on worst average delays for baseline e-detectors, when $Q$ is known} \label{subSec::delay_bound}
    
    Recall that our objective is to minimize worst average delays  $\mathcal{J}_L(\st^*)$ or $\mathcal{J}_P(\st^*)$, given above in \eqref{eq::worst_delay} and \eqref{eq::ave_delay}  respectively,  while controlling the ARL $\mathbb{E}_{P,\infty}(\st^*) \geq 1/\alpha$.  Note that worst average delays in \eqref{eq::worst_delay} and \eqref{eq::ave_delay} were defined for a fixed pair of pre- and post-change distributions implicitly. In our setting where the pre-change distribution space could be composite, we take an additional supremum over all pre-change distributions when defining both worst average delays for each fixed post-change distribution.
    
    To derive bounds on the worst average delays, we further assume that the post-change observations $X_{\nu+1}, X_{\nu+2}, \dots$ are independent of the pre-change observations and form a strongly stationary process. That is, we assume that, for any finite subset $I \subset \mathbb{N}$ and any $j \in \mathbb{N}$, the joint distributions of $\{X_{\nu + i}\}_{i \in I}$ and $\{X_{\nu + i + j}\}_{i \in I}$ are equal to each other. About the underlying baseline increment, we further assume that there exist a function $f$ and an integer $m \geq 0$ such that $L_n = f(X_n, X_{n-1}, \dots, X_{n-m})$ for each  $n$. In this warm-up section, assume that we know the post-change distribution $Q$. Then, as we shall soon see, an optimal choice for $L_n$ would set $m=0$, but if $m$ is a strictly positive number then we implicitly assume that we can access  $m$ observations $X_{0}, X_{-1}, \dots, X_{1-m}$ from the pre-change distribution in order to build sequential change detection procedures.   Under these conditions, the following theorem provides analytically more tractable upper bounds on worst average delays for e-SR and e-CUSUM procedures. 
    
    \begin{proposition} \label{prop::upperbound_simple_form}
    For a given $\alpha \in (0,1)$, let $N_\SR^*$ and $N_\CS^*$ be e-SR and e-CUSUM procedures using baseline e-detectors. Under the settings described above, their worst average delays are upper bounded as 
     \begin{align}
       \mathcal{J}_P(N_\SR^*) \leq  \mathcal{J}_L(N_\SR^*) &\leq \mathbb{E}_{0,Q} N_{1/\alpha} + m, \label{eq::delay_SR}   \\
            \mathcal{J}_P(N_\CS^*) \leq  \mathcal{J}_L(N_\CS^*) &\leq \mathbb{E}_{0,Q} N_{c_\alpha} + m, \label{eq::delay_CS} 
    \end{align}
    respectively, where $N_c$ is the stopping time defined for any $c>1$ as
     \( N_c := \inf \left\{n\geq 1 : \sum_{i = 1}^n  \log L_i \geq \log(c) \right\},
     \)
    and $c_\alpha \leq 1/\alpha$ is any threshold that ensures the e-CUSUM procedure~\eqref{eq::CUSUM_e_procedure_st} has an ARL no smaller than $1/\alpha$. Furthermore, if the post-change observations are iid and each $L_n$ is a function of $X_n$ only (i.e. $m = 0$) with $\mathbb{E}_{0,Q}\log L_1 > 0$, then
    \begin{align}
        \mathbb{E}_{0,Q} N_c &\leq \frac{\log(c)}{\mathbb{E}_{0,Q}\log L_1} + \frac{\mathbb{V}_{0,Q} \log L_1}{\left[\mathbb{E}_{0,Q}\log L_1\right]^2} + 1. \label{eq::simple_upper_SR} 
    \end{align}
\end{proposition} 

    The proof can be found in \cref{appen::proofs_for_general}. 

    \begin{remark}
    The stopping time $N_{1/\alpha}$ delivers a level-$\alpha$ sequential test for the null hypothesis  $H_0 : P \in \mathcal{P}$. On the other hand, the  stopping time $N_{c_\alpha}$ may not necessarily control the type-1 error by $\alpha$ since the threshold $c_\alpha$ can be significantly smaller than $1/\alpha$ as discussed earlier.
    \end{remark}

The expected stopping time in the upper bounds \eqref{eq::delay_SR} and \eqref{eq::delay_CS} depends on the parameter $m\geq0$. When $Q$ is known, $m=0$ suffices because \eqref{eq::simple_upper_SR} suggests that $L_i$ should simply be chosen to maximize $\mathbb{E}_{0,Q}\log L_1$, which is identical to the log-optimality criterion used for testing $\mathcal P$ against $Q$ (as discussed in many recent works, like~\cite{shafer2019game,grunwald2019safe,WaudbySmith2020EstimatingMO}). 

In applications when $Q$ is unknown, the underlying baseline increment may require a long enough sample history (large $m$) in order to achieve a reasonably small expected stopping time by ``learning'' $Q$ or using an empirical distribution plug-in for $Q$. Then, the above results suggest that a reasonable way to choose the window size $m$ is to pick the one minimizing the upper bound on the worst average delays. However, since the optimal choice of the window size should depend on the unknown post-change $Q$, it remains difficult to minimize the upper bound directly. In our simulations, we often encounter cases where a larger window size is better. In this case, we would choose a window size as large as possible while keeping the procedure computationally tractable. However, the right way to handle unknown $Q$ is dealt with in detail next.

\section{Combining baseline e-detectors using the method of mixtures} \label{sec::mixture_general}

In the previous section, we discussed how one can construct a valid e-detector and derive upper bounds on worst average delays. However, in most composite and nonparametric sequential change detection scenarios, there is no single optimal e-detector but instead there are often several applicable e-detectors to choose from.
In this section, we introduce a practicable and computationally efficient strategy to construct a good e-detector for minimizing the upper bound on worst average delays in \cref{prop::upperbound_simple_form}, especially for the upper bound~\eqref{eq::simple_upper_SR} in the $m =0$ case.

In detail, suppose we have a set of baseline increments $\{L^\lambda\}_{\lambda \in  \Pi}$ parametrized by $\lambda \in  \Pi$. Then, under the additional condition assumed in the second part of \cref{prop::upperbound_simple_form} (namely, that $m = 0$ and the post-change observations are from an iid sequence), an ideal choice of the parameter $\lambda^\op$ for a post-change distribution $Q$ is given by
\begin{equation}
    \lambda^\op(Q) = \argmax_{\lambda \in \Pi} \mathbb{E}_{0,Q} \log L_1^{(\lambda)},
\end{equation}
 minimizing the first term of the upper bound  in \eqref{eq::simple_upper_SR}, which often becomes a leading term especially for small enough $\alpha$. In turn, this term is inversely proportional to
\begin{equation}
    \mathbb{E}_{0,Q} \log L_1^{(\lambda^\op)} := D(Q||\mathcal{P}),
\end{equation}
where the second argument $\mathcal{P}$ in $D(Q||\mathcal{P})$ explicitly refers to the dependency of the baseline increment $L^{\lambda^\op}$ to the class of pre-change distributions $\mathcal{P}$. For the rest of the paper, we will assume that the set of baseline increments, $\{L^\lambda\}_{\lambda \in  \Pi}$ is rich enough such that $D(Q||\mathcal{P}) > 0$ for all $Q \in \mathcal{Q}$. As we observe later, in many canonical cases, we have $D(Q||\mathcal{P}) = \inf_{P \in \mathcal{P}} \KL(Q||P)$ where $\KL(Q||P)$ is the Kullback-Leibler (KL) divergence from $Q$ to $P$. 

Generally, computing $\lambda^\op$ is not feasible since it depends on the unknown post-change distribution $Q$. 
Next, we show how to build a mixture of baseline e-detectors that can detect the changepoint nearly as quickly as the one with $\lambda^\op$, when known lower and upper bounds $\lambda_L$ and $\lambda_U$ on $\lambda^\op$ are available. 

Notice that an average of e-detectors is also a valid e-detector, in the sense of satisfying  condition~\eqref{eq::e_detector_cond}.
Therefore, for any mixing distribution $W$ supported on $[\lambda_L, \lambda_U]$, we can define mixtures of e-SR and e-CUSUM procedures by following stopping times:
\begin{align}
    N^*_\MSR &:= \inf\left\{n \geq 1:  \int \sum_{j=1}^n  \prod_{i = j}^n L_i^{(\lambda)} \mathrm{d}W(\lambda) \geq 1/\alpha \right\}, \\
    N^*_\MCS &:= \inf\left\{n \geq 1: \int \max_{j \in [n]} \prod_{i = j}^n L_i^{(\lambda)}\mathrm{d}W(\lambda) \geq c_{\alpha} \right\},
\end{align}
where $c_{\alpha} > 1$ is a fixed constant which controls the ARL for some $\alpha \in (0,1)$. Since the mixture of e-CUSUM procedure is based on a valid e-detector, we can always set the threshold $c_{\alpha}$ to be equal to $1/\alpha$ as same as the threshold of the mixture of e-SR procedures. 

\begin{remark}
Instead of using mixtures, one may be tempted to consider swapping the above integral with a supremum over $\lambda \in [\lambda_L,\lambda_U]$. However, this does not in general yield a valid e-detector.
\end{remark}


\subsection{Computational and analytical aspects of mixtures of baseline e-detectors} \label{subSec::mixture_comp_validity}
    
    Though any mixing distribution yields a valid e-detector, for computational efficiency, we only consider discrete mixtures where the support of mixing distribution has at most countably many elements. To be specific, let $\{\omega_k\}_{k \geq 1}$ be a set of nonnegative mixing weights with $\sum_{k\geq 1}\omega_k = 1$ and let $\{\lambda_k\}_{k\geq 1}$ be the corresponding supporting set. For ease of notation, we denote $L^{(\lambda_k)} := L(k)$ for each $k \geq 1$. Based on the set of nonnegative mixing weights and the corresponding set of baseline increments,  we define mixtures of SR and CUSUM e-detectors as $M_0^\MSR  = M_0^\MCS := 1$, and for each $n \in \mathbb{N}$,
    \begin{align}
         M_n^\MSR &:= \sum_{k=1}^\infty \omega_k \sum_{j=1}^n  \prod_{i = j}^n L_i(k) := \sum_{k=1}^\infty \omega_k M_n^\SR (k) , \\
         M_n^\MCS & =\sum_{k=1}^\infty  \omega_k \max_{j \in [n]}  \prod_{i = j}^n L_i(k) := \sum_{k=1}^\infty \omega_k M_n^\CS(k).
  \end{align} 
Let $K := |\left\{k: \omega_k > 0\right\}|$ be the number of nonzero mixing weights. 
        
         \paragraph*{Finite mixtures.}
         If $K < \infty$, we may for simplicity assume that the first $K$ weights $\omega_1, \dots, \omega_K$ are the only nonzero values. In this case, we can compute mixtures of SR and CUSUM e-detectors by 
   \begin{align}
       M_n^\MSR &= \sum_{k=1}^K \omega_k M_n^\SR(k), \text{~ and ~}
      M_n^\MCS = \sum_{k=1}^K \omega_k M_n^\CS(k),
   \end{align}
     where $M_n^\SR(k)$ and $M_n^\CS(k)$ are computed recursively as
   \begin{align}
       M_n^\SR(k) &= L_n(k) \cdot  \left[M_{n-1}^\SR(k) + 1 \right],\\
      M_n^\CS(k) &=  L_n(k) \cdot \max \left\{M_{n-1}^\CS(k) , 1 \right\}
   \end{align}
    with $M_0^\SR(k) = M_0^\CS(k)  = 0$ for each $k \in [K]$ and $n \in \mathbb{N}$. Therefore, if each computation of $L_n(k)$ has constant time and space complexities, then the evaluation of mixtures of SR and CUSUM e-detectors at each time $n$ requires $O(K)$ time and space complexity.
    
    \paragraph*{Infinite mixtures, scheduling functions and adaptive re-weighting.} 
    If  $K = \infty$ or if $K$ is to be chosen adaptively as an increasing function of $n$ we modify our strategy as follows. We first choose an increasing function $K: \mathbb{N} \to \mathbb{N}$, and  let $K^{-1}:\mathbb{N}\to\mathbb{N}$ be the generalized inverse function of $K$ defined by $K^{-1}(k):= \inf\left\{j \geq 1: K(j) \geq k \right\}$ for each $k \in \mathbb{N}$. Note that $K^{-1}$ is also an increasing function.  We call such function $K$ as a \emph{scheduling function}. We intentionally overload notation: in what follows, $K(n)$ plays the same role as the constant $K$ in the case of finite support.
    Note that $K^{-1}(k) \leq n$ for any $k \leq K(n)$; we will use this simple fact below when defining nested summations.

    Based on a scheduling function $K$ and its generalized inverse $K^{-1}$, we define \emph{adaptive SR and CUSUM e-detectors}, $M_n^\ASR$ and $M_n^\ACS$, respectively, as 
    \begin{align}
        M_n^\ASR &=\sum_{k=1}^{ K(n)}  \omega_k \sum_{j= K^{-1}(k) }^n \gamma_j \prod_{i = j}^n L_i(k) := \sum_{k=1}^{ K(n)} \omega_k M_n^\SR(k), \label{eq::asr}\\ 
        M_n^\ACS &=\sum_{k=1}^{ K(n)}\omega_k  \max_{ K^{-1}(k)  \leq j \leq n} \gamma_j\prod_{i = j}^n L_i(k) := \sum_{k=1}^{ K(n)} \omega_k M_n^\CS(k), \label{eq::acusum}
    \end{align}
     where each $\gamma_j := 1 / \sum_{k=1}^{K(j)}\omega_k \geq 1$ is the \emph{adaptively re-weighting factor} at time $j$,  ensuring that the mixing weights always sum to one at each time. 
    Here, we restrict not only the space over the index $k$ from $[1, \infty]$ to $[1, K(n)]$ but also the space over the index $j$ from $[1, n]$ to $\left[K^{-1}(k), n\right]$. This choice makes it possible to compute  both $M_n^\ASR$ and $M_n^\ACS$ efficiently since each $M_n^\SR(k)$ and $M_n^\CS(k)$ have following recursive representations
    \begin{align}
         M_n^\SR(k) &= L_n(k) \cdot \left[M_{n-1}^\SR(k) + \gamma_n \right], \\
         M_n^\CS(k) &= L_n(k) \cdot \max\left\{M_{n-1}^\CS(k), \gamma_n \right\},
    \end{align}
    for each $ n \geq K^{-1}(k)$ and $M_n^\ASR(k) = M_n^\ACS(k) = 0$ for all $n = 0, 1, \dots, K^{-1}(k) -1 $. Therefore, if each computation of $L_n(k)$ has constant time and space complexities then the computations of adaptive SR and CUSUM e-detectors at each time $n$ have $O(K(n))$ time and space complexities as well. For the purpose of implementing an online algorithm, we are typically interested in the case $K(n) = O(\log(n))$. 
    \begin{remark}
    Both mixtures of SR and CUSUM e-detectors can be viewed as special cases of their adaptive counterparts where the scheduling function $K$ is understood as a constant function. In this case, we have $\gamma_j = \sum_{k=1}^K \omega_k = 1$ for each $j$, and thus $M_n^\ASR  = M_n^\MSR$ and $M_n^\ACS  = M_n^\MCS$ for each $n \in \mathbb{N}$.  
    \end{remark}
    
    Unlike finite mixtures, the mixing distribution deployed in the adaptive SR and CUSUM e-detectors vary over time. Hence, we cannot simply apply \cref{prop::mixture_is_e_detector} to check whether this adaptive scheme yields valid e-detectors. The following proposition formally states the validity of adaptive SR and CUSUM e-detectors. The proof can be found in \cref{appen::proofs_for_mixture_general}. 
    
    \begin{proposition} \label{prop::adaptive_schme_yields_valid_e_detectors}
      For any mixing weights $\{\omega_k\}_{k \in \mathbb{N}}$ and a scheduling function $K$, adaptive SR and CUSUM e-detectors defined in~\eqref{eq::asr}~and~\eqref{eq::acusum} form valid e-detectors satisfying the condition~\eqref{eq::e_detector_cond}.
    \end{proposition}
    
    Now, based on $M_n^\ACS$ and $M_n^\ASR$, the adaptive e-SR and e-CUSUM procedures are defined by the stopping times:
    \begin{align}
        N_\ASR^* &:= \inf \left\{ n \geq 1 : M_n^\ASR \geq 1/\alpha\right\}, \\
        N_\ACS^* &:= \inf \left\{ n \geq 1 : M_n^\ACS \geq c_\alpha\right\}, \label{eq::ACS_stop}
    \end{align}
    where $\alpha \in (0, 1)$ is a fixed constant and $c_\alpha $ is a positive value controlling ARL of the adaptive e-CUSUM procedure by $1/\alpha$. 
    Similar to the usual e-CUSUM procedure case we discussed before, 
    we can always set $c_\alpha = 1/\alpha$. In this case, from the fact $N_\ASR^* \leq N_\ACS^*$, which is implied by  $ M_n^\ASR \geq M_n^\ACS$, we have
    \begin{equation} \label{eq::ARL_control_adap}
        \mathbb{E}_{P,\infty}N_\ACS^* \geq \mathbb{E}_{P,\infty}N_\ASR^*  \geq 1/\alpha,
    \end{equation}
    where the last inequality comes from \cref{prop::ARL_control_e_detector} with the fact that $M^\ASR$ is a valid e-detector. However, the threshold $c_\alpha$ for the adaptive e-CUSUM procedure can be chosen to be a significantly smaller value if we have enough knowledge about the pre-change distribution, as discussed in  \cref{subSec::e-CP_procedure}.  
    
    \subsection{Worst average delay analysis for adaptive mixtures of e-detectors} \label{subSec::delay_bound_adaptive}
    We now derive general upper bounds on worst average delays of the adaptive e-SR and e-CUSUM procedures. As we did before, we further assume that post-change observations $X_{\nu+1}, X_{\nu+2}, \dots$ are independent of the pre-change observations and form a strong stationary process. Also, we further assume that there exist a function $f_k$ and an integer $m \geq 0$ such that $L_n(k) = f_k(X_n, X_{n-1}, \dots, X_{n-m})$ for each $k$ and $n$. Again, if $m$ is a strictly positive number then we implicitly assume that there exist $m$ observations $X_{0}, X_{-1}, \dots, X_{1-m}$ from the pre-change distribution we can use to build sequential change detection procedures.  Under this additional condition for worst average delay analysis, the following theorem provides analytically more tractable upper bounds on worst average delays for $N_\ASR^*$ and $N_\ACS^*$.
    \begin{theorem} \label{thm::upperbound_general_form}
      Under additional conditions described above, worst average delays for $N_\ASR^*$ and $N_\ACS^*$ can be upper bounded as follows:
             \begin{align}
           \mathcal{J}_P(N_\ASR^*) \leq  \mathcal{J}_L(N_\ASR^*) &\leq \min_{j > m}\left[\mathbb{E}_{0,Q} N_{1/\alpha}(j) + j -1 \right], \label{eq::delay_ASR}   \\
            \mathcal{J}_P(N_\ACS^*) \leq  \mathcal{J}_L(N_\ACS^*) &\leq \min_{j >m}\left[\mathbb{E}_{0,Q} N_{c_\alpha}(j) + j-1\right], \label{eq::delay_ACS} 
    \end{align}
    where, for $j \in \mathbb{N}$ and $c>0$, $N_c(j)$ is the stopping time 
    \begin{align}
      N_c(j) &:= \inf \left\{n\geq 1 : \sum_{k=1}^{ K(j)}  \omega_k\prod_{i = 1}^n L_i(k) \geq c \right\}.
    \end{align}
   Here, $c_\alpha$ is the same threshold used to build the adaptive e-CUSUM procedure in \eqref{eq::ACS_stop}. Note that for mixtures of the SR and CUSUM e-detectors where the scheduling function $K$ is a constant function, the stopping times in the upper bounds do not depend on the index $j$, and thus the upper bounds can be reduced as
     \begin{align}
           \mathcal{J}_P(N_\MSR^*) \leq  \mathcal{J}_L(N_\MSR^*) &\leq\mathbb{E}_{0,Q} N_{1/\alpha} + m , \label{eq::delay_MSR_simple}   \\
            \mathcal{J}_P(N_\MCS^*) \leq  \mathcal{J}_L(N_\MCS^*) &\leq \mathbb{E}_{0,Q} N_{c_\alpha}+ m, \label{eq::delay_MCS_simple} 
    \end{align}
    where, for $c>0$,  $N_c$ is the stopping time 
    \begin{equation}\label{eq::N_sr_sub_psi}
        N_c := \inf \left\{n\geq 1 : \sum_{k=1}^{K} \omega_k\prod_{i = 1}^n L_i(k) \geq c \right\}.
    \end{equation}
    \end{theorem}
   The proof of upper bounds on worst average delays can be found in \cref{appen::proofs_for_mixture_general}.

Unlike the baseline e-detector case, however, due to mixing weights, it is nontrivial to get further simplified upper bounds on worst average delays as we did in  \cref{subSec::delay_bound}. Next, we present specific adaptive e-SR and e-CUSUM procedures based on exponential baseline increments where we can compute both procedures efficiently and derive upper bounds on worst average delays in explicit forms.
 
\section{Exponential baseline e-detectors and their mixtures} \label{sec::mixture_simple_exponential}

Building upon recent advances in time uniform concentration inequalities and  sequential testing developed in \cite{howard2020time} and \citep{shin2021nonparametric}, below we consider an exponential structure on baseline e-detectors. We show that, in this setting,  it is possible to approximate the ``oracle'' e-SR and e-CUSUM procedures based on the knowledge of the optimal (but unknown) $\lambda^\op$ by adaptive procedures built using a mixture of carefully chosen set of baseline increments $\{L^{\lambda_k}\}_{k \geq 1}$ with mixing weights $\{\omega_k\}_{k\geq 1}$. 

To be specific, assume there exists an extended real-valued convex function $\psi$ on $\mathbb{R}$ that is finite and strictly convex on a  set $\Pi \subset \mathbb{R}$  containing $0$ in its interior $\Pi^{\mathrm{o}}$. Furthermore, assume $\psi$ is continuously differentiable on  $\Pi^{\mathrm{o}}$ with  $\nabla\psi(0)=0 = \psi(0)$. Then define the ``exponential baseline increment'' as follows.

\begin{definition}[Exponential baseline increment]\label{def:exp-base-incr}
For each  $n \in \mathbb{N}$  and  $\lambda \in \Pi$, define 
\begin{equation} \label{eq::exponential_base}
    L_n^\lambda = \exp\left\{\lambda s(X_n) - \psi(\lambda) v(X_n)\right\},
\end{equation}
where $s$ is a real-valued function and $v$ is a positive function on the sample space. $L^\lambda:= \{L^\lambda\}_{n\geq1}$ is called an exponential baseline increment if it satisfies condition~\eqref{eq::baseline_increment} in \cref{def::baseline_increment}.
\end{definition}

Above, $s$ and $v$ are mnemonics for sum and variance. For each $Q \in \mathcal{Q}$, define 
\begin{equation}\label{eq:Delta.star}
   \mu(Q) := \mathbb{E}_{0,Q}s(X_1),~ \sigma^2 :=  \mathbb{E}_{0,Q}v(X_1)~\text{and}~\Delta^\op(Q) := \frac{\mu(Q)}{\sigma^2(Q)},
\end{equation}
where we assume that all expectations are finite. The following proposition provides an explicit expression for $D(Q||\mathcal{P}) := \max_{\lambda \in \Pi}\mathbb{E}_{0,Q} \log L_1^{(\lambda)}$ and a sufficient condition to have $D(Q||\mathcal{P}) > 0$ when the underlying baseline increments have the form specified in~\eqref{eq::exponential_base}.

\begin{proposition} \label{prop::nonneg_div}
    For a fixed $Q \in \mathcal{Q}$, suppose there exist $\lambda^\op \in \Pi^{\mathrm{o}}$ such that $\Delta^\op(Q) = \nabla \psi(\lambda^\op)$. Then,
    \begin{equation}
       D(Q||\mathcal{P}) =  \mathbb{E}_{0,Q} \log L_1^{(\lambda^\op)} 
       = \psi^*\left(\Delta^\op(Q)\right) \sigma^2(Q) \geq 0,
    \end{equation}
    where $\psi^*$ is the convex conjugate of $\psi$. Thus,  if $\Delta^\op(Q) \neq 0$, we have $ D(Q||\mathcal{P}) > 0$. 
\end{proposition}
The proof of \cref{prop::nonneg_div} can be found in \cref{appen::proofs_for_mixture_simple_exponential}. 
For the rest of the section, we assume that 
\begin{align*}
&\left\{\lambda \in \mathbb{R}: \Delta^\op(Q) = \nabla \psi(\lambda), Q \in \mathcal{Q}\right\} \subset \Pi,\\
&\Delta^\op(Q) \neq 0,~~\forall Q \in \mathcal{Q}.    
\end{align*}
Then, \cref{prop::nonneg_div} implies $D(Q||\mathcal{P}) > 0$ for all $Q \in \mathcal{Q}$. Also, for ease of notation, we will drop the dependency of $Q$ from related parameters and simply write $\lambda^\op, \mu, \sigma^2$ and $\Delta$. 

The exponential structure of the baseline increment in \eqref{eq::exponential_base} results in a simple form of $\lambda^\op$ such that
\begin{equation}
\lambda^\op= (\nabla\psi)^{-1}(\Delta^\op) = \nabla\psi^*(\Delta^\op),
\end{equation}
where the second equality comes from the fact that $\lambda = \nabla \psi^* \circ \nabla \psi(\lambda)$ for each $\lambda \in \Pi$. Although $\lambda^\op$ still depends on the unknown post-change distribution $Q$ via $\Delta^\op$, in many cases, we can find upper and lower bounds on $\Delta^\op$. In this section, we explain how to use the knowledge of the range of $\Delta^\op$ to build a mixture of exponential baseline e-detectors that has explicit upper bounds on worst average delays.

\subsection{Separated pre- and post-change distributions}

Suppose we have knowledge of upper and lower bounds on the parameter $\Delta^\op$ given in \eqref{eq:Delta.star}, i.e. we know a pair $(\Delta_L, \Delta_U)$ such that $\Delta_L < \Delta^\op < \Delta_U$. It then  follows that $\lambda_L < \lambda^\op < \lambda_U$, where  $\lambda_L = \nabla \psi^*(\Delta_L)$,  $\lambda^\op = \nabla \psi^*(\Delta^\op)$ and  $\lambda_U = \nabla \psi^*(\Delta_U)$.
To simplify presentation, we only consider the \emph{one-sided} and \emph{well-separated} case: $0 < \lambda_L <\lambda_U$. 

 Let $1/\alpha$ be the target level of the ARL control for a fixed $\alpha \in (0,1)$. Let $\{L(k)\}_{k \in [K]}$ and $\{\omega_k\}_{k \in [K]}$ be $K$ exponential baseline increments and mixing weights whose specific values will be defined later in this subsection. Since each $L_n(k)$ is a function of the $n$-th observation $X_n$ for each $k \in [K]$,  \cref{thm::upperbound_general_form} implies that, if the post-change observations form a strong stationary process then the worst average delays for  mixtures of e-SR and e-CUSUM procedures, $N_\MSR^*$ and $N_\MCS^*$ can be upper bounded by $\mathbb{E}_{0,Q} N_{1/\alpha}$ and $\mathbb{E}_{0,Q} N_{c_\alpha}$, respectively,  where $N_c$ is the stopping time defined in \eqref{eq::N_sr_sub_psi}. Furthermore, as we can always set the threshold for the e-CUSUM procedure to be $c_\alpha \leq 1/\alpha$, we have $\mathbb{E}_{0,Q} N_{c_\alpha} \leq \mathbb{E}_{0,Q} N_{1/\alpha}$. Therefore, in this subsection, we construct a set of baseline increments for which we can derive a tight bound on $\mathbb{E}_{0,Q} N_{1/\alpha}$.

    \cref{alg::mixture_SR} describes our methodology for computing mixtures of e-SR procedures in detail. The inputs to the algorithm are the upper and lower bounds $\Delta_U$ and $\Delta_L$ on $
 \Delta^\op$ and the maximal number of baselines processes $K_{\max}$. Mixture of e-CUSUM procedures can be executed similarly by replacing  \cref{line::SR_update_general} by
\begin{equation}
    M_n^{\CS}(k) \leftarrow \exp\left\{\lambda_k s(X_n) - \psi_k v(X_n)\right\}\cdot \max\left\{M_{n-1}^\CS (k),1\right\}.
\end{equation}
Also, for the mixture of e-CUSUM procedures, we can replace the threshold $1/\alpha$ with a smaller value $c_\alpha$ if we have enough information about the pre-change distribution. For both e-SR and e-CUSUM, at each time $n$, updates of mixtures of e-detectors have $O(K_\alpha)$ time and space complexities, which do not depend on $n$. 

 \cref{alg::mixture_SR} relies critically on the function \texttt{computeBaseline} in  \cref{line::computeBaseline.line}, which returns a set of parameters and weights to compute a mixture of e-detectors along with a threshold value $g_\alpha > 0$ that will appear in the upper bound on worst average delays given  in \cref{thm::upper_bound_well_sep_general} that will be explained below.  The details of \texttt{computeBaseline} are fairly technical and are given in \cref{alg::compute_base} in \cref{appen::proofs_for_mixture_simple_exponential}. 
    
    In the main result of this section, we provide bounds on ARL and worst average delays for the mixtures of e-CP procedures obtained with  \cref{alg::mixture_SR} that is a function of the parameter $\lambda^\op$ and the threshold value $g_\alpha$. The proof can be found in \cref{appen::proofs_for_mixture_simple_exponential}.
    \begin{theorem} \label{thm::upper_bound_well_sep_general}
    Let $N_\MSR^*$ and $N_\MCS^*$ be the stopping times corresponding to the mixtures of e-SR and e-CUSUM procedures in \cref{alg::mixture_SR} and its variant, respectively. 
    Then, both procedures control the ARL by $1/\alpha$. If we further assume that the post-change observations $X_{\nu+1}, X_{\nu+2}, \dots$ are iid samples from a post-change distribution $Q$,  then the worst average delays for $N_\MSR^*$ and $N_\MCS^*$ can be  bounded as
    \begin{equation}\label{eq::upper_bound_well_sep_general}
       \max\left\{ \mathcal{J}_L(N_\MSR^*), \mathcal{J}_L(N_\MCS^*) \right\}
        \leq \frac{g_\alpha}{D(Q||\mathcal{P})} +  \frac{\mathbb{V}_{0,Q} \left[\log L_1^{(\lambda^\op)}\right]}{\left[D(Q||\mathcal{P})\right]^2} + 1.
    \end{equation}

 The same bound holds also for $\mathcal{J}_P(N_\MSR^*)$ and $\mathcal{J}_P(N_\MCS^*)$.
    \end{theorem}
    
    In   \cref{prop::upper_bound_on_g} in  \cref{appen::proofs_for_mixture_simple_exponential}, we show that if the number of baseline processes $K_{\max}$ in \cref{alg::mixture_SR} is chosen large enough, then the quantity $g_\alpha$ returned by  \texttt{computeBaseline}  is at most
        \begin{equation} \label{eq::g_upper_explicit_main_sec3}
            \inf_{\eta > 1} \eta \left[\log(1/\alpha) + \log\left(1 + \left\lceil \log_\eta \frac{\psi^*(\Delta_U)}{\psi^*(\Delta_L)}\right\rceil \right)\right],
        \end{equation} 
        which can be easily evaluated numerically. Expression \eqref{eq:Kmax.large.enough} in  \cref{appen::proofs_for_mixture_simple_exponential} provides a precise formula for how large $K_{\max}$ needs to be in order for the above bound to be in effect. In most practical cases, $K_{\max} = 1000$ is a large enough choice.  Also, in many canonical examples we will present later, if we choose large enough $K_{\max}$ satisfying the condition  \eqref{eq:Kmax.large.enough}  then the first term $\frac{g_\alpha}{D(Q||\mathcal{P})}$ of the upper bound of worst average delays in \cref{thm::upper_bound_well_sep_general} become a leading term. In this case, from the inequality~\eqref{eq::g_upper_explicit_main_sec3}, we can check that this leading term is $O\left(\log(1/\alpha) / D(Q||\mathcal{P})\right)$ as $\alpha \to 0$.

    \begin{remark}
                If there is only one pre-change distribution $P$ and one post-change distribution $Q$, both from a natural univariate exponential family, then their likelihood ratio forms an exponential baseline increment. In this case, the above upper bound becomes $O\left(\log(1/\alpha) / \KL(Q||P)\right)$ as $\alpha \to 0$, matching the rate of the known lower bounds \citep{lorden1971procedures}.
    \end{remark}

The bound on worst average delays in \cref{thm::upper_bound_well_sep_general} is obtained by analyzing an auxiliary stopping time 
    \begin{equation} \label{eq::glr_st_main}
        \bar{N}_g:= \inf\left\{n \geq 1: \sup_{\lambda \in (\lambda_L, \lambda_U)} \sum_{i = 1}^n \log L_i^{(\lambda)} \geq g \right\}, \quad g > 1.
    \end{equation}
Using the same arguments as in the proof of \cref{prop::upperbound_simple_form}, we immediately have that if the post-change observations are iid\ from $Q$, then for any $g>1$,
    \begin{equation} \label{eq::glr_bound_general}
        \mathbb{E}_{0,Q} \bar{N}_g 
        \leq   \frac{g}{D(Q||\mathcal{P})} +  \frac{\mathbb{V}_{0,Q} \left[\log L_1^{(\lambda^\op)}\right]}{\left[D(Q||\mathcal{P})\right]^2} + 1 .
    \end{equation}
The  bound \eqref{eq::upper_bound_well_sep_general} is finally established by showing that the stopping time $\bar{N}_{g_\alpha}$ obtained by using the threshold $g_\alpha$ produced by \cref{alg::compute_base} is a deterministic upper bound to the stopping times $N_{c_\alpha}$ and $N_{1/\alpha}$ corresponding to mixtures of SR and CUSUM e-detectors. In detail, it holds that for any stream of observations $X_1,X_2,\ldots$,
\begin{equation}
    N_{c_\alpha} \leq N_{1/\alpha} \leq \bar{N}_{g_\alpha}.
\end{equation}
This nontrivial result is formally stated in  \cref{lemma::alg_1_result} in \cref{appen::proofs_for_mixture_simple_exponential}. Its proof leverages geometric arguments used in \citep[Theorem 2]{shin2021nonparametric} to analyze sequential generalized likelihood ratio tests.

\begin{algorithm*} 
\DontPrintSemicolon
\KwInput{ARL parameter $\alpha \in (0,1)$, Boundary values $0<\Delta_L < \Delta_U$,\newline Maximum number of baselines $K_{\max} \in \mathbb{N}$.}
\KwOutput{Stopping time $N_\MSR^*$ of the mixture of e-SR procedures.}
\KwData{Data stream $X_1, X_2, \dots$ (observed sequentially)
 } 
$\{\lambda_0, \lambda_1,\dots  \lambda_{K_\alpha} \}$, $\{\omega_0,\omega_1, \dots, \omega_{K_\alpha}\}$, $g_\alpha$ $\leftarrow$ \texttt{computeBaseline}($\alpha, \Delta_L, \Delta_U, K_{\max}$)  \; \label{line::computeBaseline.line} 
\For{$k = 0, 1, \dots, K_\alpha$}{
$M_0^{\SR}(k) \leftarrow 0$, $\psi_k \leftarrow \psi(\lambda_k)$\;
}
$M_0^\MSR \leftarrow 0$, $n \leftarrow 0$\;
\While{$M_n^\MSR <1/\alpha$}{
$n \leftarrow n + 1$\;
Observe $X_n$\;
\For{$k = 0, 1, \dots, K_\alpha$}{
$M_n^{\SR}(k) \leftarrow \exp\left\{\lambda_k s(X_n) - \psi_k v(X_n)\right\}\cdot \left[M_{n-1}^\SR (k) + 1\right]$ \; \label{line::SR_update_general}
}
$M_n^\MSR \leftarrow  \sum_{k=0}^{K_\alpha} \omega_k M_n^\SR(k)$\;
}
$N_\MSR^* \leftarrow n$ \;
\Return The stopped time $N_\MSR^*$
 \caption{ Pseudo-code of the mixture of e-SR procedures}
 \label{alg::mixture_SR}
\end{algorithm*}

\subsection{Non-separated pre- and post-change distributions}

The previous subsection discussed how to build mixtures of e-SR and e-CUSUM procedures with an explicit upper bound on worst average delays when we have known and positive boundary values, $\lambda_L$ and $\lambda_U$ on the unknown $\lambda^\op$ via the knowledge of $\Delta_L < \Delta^\op < \Delta_U$. However, in many cases, we may not be fully certain about the boundary values. In this subsection, we discuss how we can generalize the previous argument to the no separation case whereby we only know the sign of $\lambda^\op (>0)$ but do not have specific boundary values.

Recall that, for the well-separated case, we calibrated the mixtures of finitely many exponential baseline e-detectors using the stopping time $\bar{N}_{g_\alpha}$ in \eqref{eq::glr_st_main}, which is in turn based on the maximum of underlying baseline increments over the known upper and lower bounds of $\lambda^\op$. Since we no longer have knowledge of the boundary values $\lambda_L$ and $\lambda_U$, we may use similar stopping times where the range of maximum and the threshold slowly increase over time. In this case, we need an infinite sequence of baseline procedures $\left\{L(k)\right\}_{k \in \mathbb{N}}$ and mixing weights $\{\omega_k\}_{k \in \mathbb{N}}$ to build adaptive e-SR and e-CUSUM procedures. 

The bound in \cref{thm::upperbound_general_form} along with the fact $\gamma_j \geq 1$ for all $j \in \mathbb{N}$ implies that, for any given scheduling function $K : \mathbb{N} \to \mathbb{N}$, if the post-change observations form a strong stationary process then worst average delays for adaptive e-SR and e-CUSUM procedures can be upper bounded by 
    $\min_{j \geq 1} \left[ \mathbb{E}_{0,Q} N_{1/\alpha} (j) + j - 1\right]$ and $\min_{j \geq 1} \left[ \mathbb{E}_{0,Q} N_{c_\alpha} (j) + j - 1\right]$, respectively, where we recall that  $N_c(j)$ is defined for $c>0$ by 
    \begin{equation}
      N_c(j) := \inf \left\{n\geq 1 : \sum_{k=1}^{K(j)}  \omega_k\prod_{i = 1}^n L_i(k) \geq c \right\}.        
    \end{equation}
  Again, since we can set the threshold for the e-CUSUM procedure in such a manner that $c_\alpha \leq 1/\alpha$ (so that $\mathbb{E}_{0,Q} N_{c_\alpha}(j) \leq \mathbb{E}_{0,Q} N_{1/\alpha}(j)$), in this subsection, we focus on constructing a set of baseline increments on which we can derive a tight upper bound on $\min_{j \geq 1} \left[ \mathbb{E}_{0,Q} N_{1/\alpha} (j) + j - 1\right]$. 
        
    To derive the set of baseline increments, we use a time-varying boundary function $g$. Here, we intentionally overload notation: the constant $g$ in the previous subsection for the well-separation case can be viewed as a constant function $g$ in what follows. Let $g: [1,\infty) \to [0,\infty)$ be a nonnegative and nondecreasing continuous function such that the mapping $t \mapsto g(t) / t$ is nonincreasing  and $\lim_{t\to \infty} g(t) /t = 0$. For a chosen positive number $\Delta_0 > 0$, let 
    \[\text{$D_0 := \psi^*(\Delta_0)$ and $V_0 := \inf\left\{t \geq 1: D_0 \geq g(t) / t\right\}$}.\]
    Now, for any fixed $\eta > 1$ and $j \in \mathbb{N}$,  define $\Delta_1 > \Delta_2> \cdots $ as positive solutions of the equations
    \begin{equation} \label{eq::const_of_SR_no_sep}
        \psi^*\left(\Delta_k\right) = \frac{g\left(V_0 \eta^{k} \right)}{V_0\eta^{ k}} ,~~ k=0, 1,2,\ldots.
    \end{equation}
    Finally, based on the sequence $\{\Delta_k\}_{k \geq 0}$, define \[\lambda_k := \nabla \psi^*(\Delta_k),
    \] and set 
    $\omega_0 := \alpha^{-1}e^{-g(V_0)}\mathbbm{1}(g(V_0) > v_{\min} D_0), \omega_k := \alpha^{-1}e^{-g\left(V_0\eta^k\right) / \eta}$ for each $k \in \mathbb{N}$ where $v_{\min} := \min_{x}v(x)$, recalling the function $v$ from Definition~\ref{def:exp-base-incr}.
    
    Based on the quantities defined above, we can construct the stopping time $N_{1/\alpha}(j)$ for each $j$. The following lemma shows that we can upper bound the stopping time $N_{1/\alpha}(j)$ with another stopping time $\bar{N}_g(j)$ from which we can derive an explicit upper bound on its expected stopping time. 
    \begin{lemma} \label{lem::condition_for_SR_to_glrt_general}
        For any fixed $j \geq  1$, $\Delta_0 > 0$, and tuning parameter $\eta > 1$, let $N_{1/\alpha}(j)$ be the stopping time based on the parameters defined above. 
        Then, we have
        \begin{equation}
         N_{1/\alpha} (j) \leq \bar{N}_g(j) ,  
        \end{equation}
        where $\bar N_g(j)$ is a stopping time defined
        by
        \begin{equation} \label{eq::N_G_j}
            \bar N_g(j) := \inf\left\{n \geq 1 : \sup_{\lambda \in (\lambda_{K(j)}, \lambda_0)} \sum_{i = 1}^n \log L_i^{(\lambda)} \geq g\left(V_0 \eta^{K(j)} \right)\right\}.
        \end{equation}
    \end{lemma}
    
    Note that the chosen set of weights $\{\omega_k\}_{k \geq 0}$ yields valid adaptive e-SR and e-CUSUM procedures if 
     \begin{equation} \label{cond::convert_SR_fully_general}
        e^{-g(V_0)}\mathbbm{1}(g(V_0) > v_{\min} D_0) + \sum_{k=1}^{\infty} e^{-g\left(V_0\eta^k\right) / \eta} \leq \alpha.
    \end{equation}
    Once the above condition is satisfied, we can use the worst average delay analysis in \cref{subSec::delay_bound_adaptive} with the bound in \cref{lem::condition_for_SR_to_glrt_general} to get an explicit upper bound on the worst average delay of the adaptive e-SR and e-CUSUM procedures. 
    
In detail, let $j^\op$ be the smallest integer satisfying $\lambda_{K(j^\op)} < \lambda^\op$ and set $K^\op := K(j^\op)$. If we also have $\lambda^\op < \lambda_0$ then \cref{lem::condition_for_SR_to_glrt_general} implies 
    \begin{equation}
        \mathbb{E}_{0,Q} N_{1/\alpha} (j^\op) \leq \mathbb{E}_{0,Q} \bar{N}_{g} (j^\op) \leq \mathbb{E}_{0,Q} N_\op,
    \end{equation}
    where the stopping time $N_\op$ is defined by
    \begin{equation}\label{eq::oracle_upper_bound}
         N_\op := \inf\left\{n \geq 1 : \sum_{i = 1}^n \log L_i^{(\lambda^\op)} \geq g\left(V_0 \eta^{K^\op} \right)\right\},
    \end{equation}
    and the expectation $\mathbb{E}_{0,Q} N_\op$ is typically on the order of $g\left(V_0 \eta^{K^\op} \right)/ D(Q||\mathcal{P})$.   Based on this observation, in the rest of this subsection, we introduce a practical and interpretable way to choose a boundary function $g$ and related tuning parameters which minimize the leading term $g\left(V_0 \eta^{K^\op} \right)$ while satisfying the condition \eqref{cond::convert_SR_fully_general} on the set of mixing weights.
   
First note that, although we have no bounds on $\Delta^\op$ in the no separation case, we can still choose $\Delta_L$ and $\Delta_0$ with $\Delta_L < \Delta_0$ as tuning parameters that represent our initial guess on the range of the unknown $\Delta^\op$. Since it is possible that the unknown parameter $\Delta^\op$ of the post-change distribution is outside of the boundary $(\Delta_L, \Delta_0)$, instead of assigning the entire $\alpha$ to the inside of the guessed interval, we split it into two parts by $r\alpha$ and $(1-r)\alpha$, respectively where $r \in (0,1)$ is another tuning parameter called the importance weight. Roughly speaking, larger $r$ implies we make a higher bet on that the unknown $\Delta^\op$ is inside of our chosen boundaries $(\Delta_L, \Delta_0)$.  

Now, given tuning parameters $\Delta_L, \Delta_0$ and $r$, we compute the set of $\{g_{r\alpha}, K_L, \eta\}$ by executing the function \texttt{computeBaseline}, just like in  \cref{alg::mixture_SR}, except that $\alpha$ is replaced replaced by $r\alpha$. Then, we can extend the boundary function $g$ to accommodate the case in which the unknown $\Delta^\op$ is not inside the initial interval we had guessed. To be specific, we use the boundary function 
\begin{equation} \label{eq::SR_boundary}
    t \in [1,\infty) \mapsto g(t) := g_{r\alpha} + s \eta \log\left( 1 + \log_\eta \left(\frac{t}{V_0 \eta^{K_L}} \vee 1 \right)\right),
\end{equation}
where $V_0 := g_{r\alpha} / D_0$ and  $s > 1$ is a constant obtained as the solution of the equation
\begin{equation} 
    \zeta(s) -1 := \sum_{k=1}^\infty \frac{1}{(1 + k)^s} = e^{g_{r\alpha} / \eta} \left[\alpha - \left\{e^{-g_{r\alpha}}\mathbbm{1}(g_{r\alpha} > D_0) + K_L e^{-g_{r\alpha} / \eta}\right\}\right].
\end{equation}
Note that the right hand side of the above equation is approximately equal to $(1-r)\alpha e^{g_{r\alpha} / \eta}$. Therefore, 
\begin{equation}
 s \approx \zeta^{-1}\left(1 + [1-r]\alpha e^{g_{r\alpha} / \eta} \right).
\end{equation}

\begin{algorithm*}[ht!] 
\DontPrintSemicolon
\KwInput{ARL parameter $\alpha \in (0,1)$, Tuning parameters $\Delta_L < \Delta_0$, importance weight $r \in (0, 1)$, scheduling parameter $m \geq 1$, Number of baselines for the well-separated regime $K_0 \in \mathbb{N}$.}
\KwOutput{Stopping time $N_\ASR^*$ of the adaptive e-SR procedure.}
\KwData{Data stream $X_1, X_2, \dots $ (observed sequentially)
 } 
Obtain $\{\lambda_0, \lambda_1,\dots  \lambda_{K_L}\}$, $\{\omega_0,\omega_1, \dots, \omega_{K_L}\}$, $\left\{g_{r\alpha}, K_L, \eta,W\right\}$ by executing \texttt{computeBaseline}($r\alpha, \Delta_L, \Delta_0, K_0$) in \cref{alg::compute_base}. \;
$s \leftarrow \zeta^{-1}\left(1 +  \left[\alpha - W\right] e^{g_{r\alpha} / \eta}\right) $ \tcc{$\zeta (s) \approx1 +  \left[1-r\right]\alpha e^{g_{r\alpha} / \eta}$} 
$M_0^{\SR}(k) \leftarrow 0$ , $\psi_k \leftarrow \psi(\lambda_k)$, $\omega_k \leftarrow \omega_k \frac{W}{\alpha},~~\forall k = 0, 1, \dots, K_L$ \;
$M_0^\ASR \leftarrow 0$, $\gamma \leftarrow 1/\sum_{k = 0}^{K_L} \omega_k$, $n \leftarrow 0$\;

\While{$M_n^\ASR < 1/\alpha$}{
$n \leftarrow n + 1$\;
\tcc{Occasionally add a new baseline increment.}
$K_n \leftarrow K_L + \left\lceil m \log_{\eta} n \right\rceil$ \;
\If{$K_n  > K_{n-1}$}{
\For{$k = K_{n-1} + 1, \dots,K_n$}{
Compute $\Delta_{k}$ as the solution of 
        \(
                \psi^*\left(z\right) = \frac{g_k}{V_0\eta^{ k}} 
        \)
with respect to $z (> 0)$, 
where $V_0 := g_{r\alpha} / D_0$ and $g_k := g_{r\alpha} + s\eta\log\left(1 + k - K_L\right) $.\;
$M_{n-1}^\ASR(\mu_{k}) \leftarrow 0$, $\psi_k \leftarrow \psi(\lambda_k)$, $\omega_k \leftarrow \alpha^{-1}e^{- g_k/\eta}$ \;
}
$\gamma  \leftarrow \left(1/\gamma + \sum_{k=K_{n-1}}^{K_n}\omega_k\right)^{-1}$\;
}
Observe $X_n$\;
$M_n^{\SR}(k) \leftarrow \exp\left\{\lambda_k s(X_n) - \psi_k v(X_n)\right\}\cdot \left[M_{n-1}^\SR (k) + \gamma\right],~~\forall k = 0,1,\dots, K_n$ \; \label{line::adap_SR_update}
$M_n^\ASR \leftarrow \sum_{k=0}^{K_{n}} \omega_k M_n^\SR(k)$\;
}
$N_\ASR^* \leftarrow n$ \;
\Return The stopped time $N_\ASR^*$
 \caption{ Pseudo-code of the adaptive e-SR procedures}
 \label{alg::adap_SR}
\end{algorithm*}

In \cref{alg::adap_SR}, we provide the detailed steps for the adaptive e-SR procedure based on the boundary function in \eqref{eq::SR_boundary}. The algorithm can be easily modified for the adaptive e-CUSUM procedure by replacing the update in \cref{line::adap_SR_update} with the rule
\begin{equation}
    M_n^{\CS}(k) \leftarrow \exp\left\{\lambda_k s(X_n) - \psi_k v(X_n)\right\}\cdot \max\left\{M_{n-1}^\CS (k), \gamma\right\}.
\end{equation}
Also, for the adaptive e-CUSUM procedure, we can replace the threshold $1/\alpha$ with a smaller value $c_\alpha$ if we have enough information about the pre-change distribution.

In term of computational complexity, in \cref{alg::adap_SR} we set the scheduling function $K : \mathbb{N} \to \mathbb{N}$ as
\begin{equation}
    K(n) := K_L + \lceil m\log_\eta n \rceil,
\end{equation}
where $m \geq 1$ is a tuning parameter.  Therefore, for both adaptive e-SR and e-CUSUM procedures,  updates of statistics have $O(m\log_\eta n)$ time and space complexities at each time $n$. Although it is not a fully online algorithm, logarithm time and space complexities make it feasible to run adaptive e-SR and e-CUSUM procedures in most practical online settings.

From \cref{subSec::mixture_comp_validity}, we know that both procedures control the ARL by $1/\alpha$. The following theorem introduces explicit bounds on the worst average delays for both procedures.

\begin{corollary} \label{cor::explicit_upper_no_sep}
     Let $N_\ASR^*$ and $N_\ACS^*$ be stopping times corresponding to the adaptive e-SR  procedures in \cref{alg::adap_SR} and its and e-CUSUM variant, respectively. 
    Then, both procedures control ARL by $1/\alpha$. If we further assume that post-change observations $X_{\nu+1}, X_{\nu+2}, \dots$ are iid samples from a post-change distribution then the worst average delays for $N_\ASR^*$ and $N_\ACS^*$ can be upper bounded as
    \begin{equation}\label{eq::upper_bound_no_sep}
    \begin{aligned}
        &\max\left\{ \mathcal{J}_L(N_\ASR^*), \mathcal{J}_L(N_\ACS^*) \right\}\\ &\leq     \begin{cases} 
     \frac{g_{r\alpha}}{D(Q||\mathcal{P})}\frac{\psi^*\left(\Delta^\op\right)}{\psi^*\left(\Delta_0\right)} +  \frac{\mathbb{V}_{0,Q} \left[\log L_1^{(\lambda_0)}\right]}{\left[D(Q||\mathcal{P})\right]^2}\left[\frac{\psi^*\left(\Delta^\op\right)}{\psi^*\left(\Delta_0\right)}\right]^2 + 1 &\mbox{if } \Delta^\op\geq \Delta_0 \\
   \frac{g_{r\alpha}}{D(Q||\mathcal{P})} +  \frac{\mathbb{V}_{0,Q} \left[\log L_1^{(\lambda^\op)}\right]}{\left[D(Q||\mathcal{P})\right]^2} + 1    & \mbox{if } \Delta^\op \in (\Delta_L, \Delta_0) \\
    \frac{g_{r\alpha} + s\eta\log\left(1 + K^\op - K_L\right)}{D(Q||\mathcal{P})} +  \frac{\mathbb{V}_{0,Q} \left[\log L_1^{(\lambda^\op)}\right]}{\left[D(Q||\mathcal{P})\right]^2} +  \left[\frac{\psi^*\left(\Delta_L\right)}{\psi^*\left(\Delta^\op\right)} \frac{g_{r\alpha} + s\eta\log\left(1 + K^\op - K_L\right)}{g_{r\alpha}}\right]^{1/m}   &\mbox{if } \Delta^\op\leq \Delta_L
    \end{cases}.
    \end{aligned}
    \end{equation}
Note that $\eta, s > 1$ and $r \in (0,1)$ do not depend on the unknown $\Delta^\op$. 
\end{corollary}

\section{Application to real data and simulation study}\label{sec::examples}
\subsection{Bernoulli random variables with dependent, time-varying means} \label{subSec::example_bernoulli}

\paragraph*{Winning rates of the Cavaliers.} To illustrate how sequential change detection procedures based on e-detectors work, we revisit the example of the Cleveland Cavaliers, an American professional basketball team introduced in \cref{subSec::intro_Cav_example}. Instead of using Plus-Minus stats, in this example, we are monitoring the performance of the Cavaliers by keeping track of wins and losses over all the games. Let $X_1, X_2, \dots \in \{0,1\}$ be the sequence of win indicators during 2010-11 to 2017-18 regular seasons, where $X_i=1$ if the Cavaliers won game $i$. Though \cref{fig::cavs_win_rate} presents monthly and seasonal averages for the purpose of visualization, we use the underlying binary sequence to build a sequential change detection procedure. 

\paragraph*{Modeling winning probabilities as a dependent sequence of Bernoullis.} 

To detect a significant improvement of the performance of the Cavaliers, we assume that before an unknown changepoint $\nu \in \mathbb{N}\cup\{\infty\}$, the conditional average of winning probability given the sample history is less than or equal to $p_0:= 0.49$. That is, under any pre-change distribution $P$ we have $p_n := \mathbb{E}_{P,\infty}[X_n \mid \mathcal{F}_{n-1}] \leq p_0$. (For simplicity, $\mathcal F$ is taken to be the natural filtration of the data.) Thus, the pre-change class of distributions is
\[
\mathcal P := \{ (p_1,p_2,\dots) : p_i \leq p_0, ~ \forall i \geq 1\},
\]
where we parameterize each distribution $P$ over binary sequences by the sequence of conditional probabilities.

Our objective is to build mixtures of e-SR and e-CUSUM procedures tuned to quickly detect any significantly improved win rate larger than $q_0:=0.51$ after the changepoint. This can be modeled by assuming that after some changepoint $\nu$, the distribution $Q$ is such that $\mathbb{E}_{P,\nu,Q}\left[X_n \mid \mathcal F_{n-1}, n > \nu \right] := q_n \geq q_0$. Thus, we may think of the post-change class of distributions as being
\[
\mathcal Q := \{(q_1,q_2,\dots) : q_i \geq q_0, ~ \forall i \geq 1\}.
\]
In particular, this formalization allows for the winning probabilities to fluctuate over time before and after the changepoint (accounting for factors like form, injuries, etc.).

\paragraph*{Deriving exponential baseline processes.}
For each $\lambda > 0$, define a baseline increment process $L^{(\lambda)} := \{L_n^{(\lambda)}\}_{n \geq 1}$ as
\begin{equation} \label{eq::base_line_bernoulli}
    L_n^{(\lambda)} := \exp\left\{\lambda\left(X_n - p_0\right) - \psi_B(\lambda) \right\},
\end{equation}
where $\psi_B(\lambda) := \log\left(1 - p_0 + p_0 e^\lambda\right) - \lambda p_0$ is the Bernoulli cumulant generating function. Note that each $L^{(\lambda)}$ is a valid baseline increment as it satisfies the inequality~\eqref{eq::baseline_increment} in \cref{def::baseline_increment}. That is, under any pre-change distribution $P$, we have 
\begin{align*}
\mathbb{E}_{P,\infty}\left[L_n^{(\lambda)}\mid \mathcal{F}_{n-1}\right] 
&= \mathbb{E}_{P,\infty}\left[\exp\left\{\lambda X_n - \log\left(1 - p_0 + p_0 e^\lambda\right)\right\}\mid \mathcal{F}_{n-1}\right] \\
&= \frac{\mathbb{E}_{P,\infty}\left[e^{\lambda X_n} \mid \mathcal{F}_{n-1}\right] }{1 - p_0 + p_0 e^\lambda} = \frac{1 - p_n + p_n e^\lambda}{1 - p_0 + p_0 e^\lambda} \leq 1,~~\forall \lambda \geq 0.
\end{align*}

To derive exponential baseline processes, we first consider a simplified post-change distribution $Q$ where each post-change observation is identically distributed with $\mathbb{E}_{0,Q}\left[X_1\right] := q \geq q_0 > p_0$. In this case,  the optimal choice of $\lambda \geq 0$ given by
\begin{equation}
    \lambda^\op := \argmax_{\lambda \geq 0} \mathbb{E}_{0,Q}\exp\left\{\lambda(X_1 - p_0) - \psi_B(\lambda)\right\}.
\end{equation}
Since the baseline increment has the exponential structure, by \cref{prop::nonneg_div}, we have that
\begin{equation} \label{eq::divergence_bernoulli}
    D(Q||\mathcal{P}) := \mathbb{E}_{0,Q}\log L_1^{(\lambda^\op)} = \psi_B^*\left(q - p_0\right) = \mathrm{KL}(q||p_0),
\end{equation}
where $\mathrm{KL}(q||p_0)$ is the Kullback-Leibler (KL) divergence of Bernoulli distributions with parameters $q$ and $p_0$ written as
\begin{equation}
    \mathrm{KL}(q||p_0) := q \log \frac{q}{p_0} + (1-q) \log \frac{1-q}{1-p_0},
\end{equation}
for $q, p_0 \in (0,1)$. The appearance of the KL divergence in \eqref{eq::divergence_bernoulli} is not a coincidence as the baseline increment can be viewed as a re-parametrized likelihood ratios between two Bernoulli processes. However, the simple geometric structure of the baseline increment make it possible to utilize a prior knowledge about the post-change distribution via \cref{alg::mixture_SR}~and~\ref{alg::adap_SR}.

For instance, suppose we know upper and lower bounds of conditional means of the post-change distribution as $q_n \in (q_L, q_U), \forall n > \nu$. Let $N_{\MSR}^*$ and $N_{\MCS}^*$ be stopping times of mixtures of e-SR and e-CUSUM procedures in \cref{alg::mixture_SR}. In this case, derived sequential change detection procedures do not rely on a specific choice of a post-change distribution $Q \in \mathcal{Q}$. However, these procedures can still perform almost as well as the one optimized to a specific choice of the post-change distribution within the same range $(q_L, q_U)$. Typically, if the post-change observations are iid samples from a post-change distribution $Q$ with $\mathbb{E}_{0,Q}\left[X_1\right] := q \in (q_L, q_U)$, then by \cref{thm::upper_bound_well_sep_general}, the worst average delays have the following explicit bound:
\begin{align*}
   \max\left\{ \mathcal{J}_L(N_\MSR^*), \mathcal{J}_L(N_\MCS^*) \right\}
    \leq \frac{g_\alpha}{\mathrm{KL}(q||p_0)} +  \frac{q(1-q)\left[\log\left(\frac{1-p_0}{p_0}\frac{q}{1-q}\right)\right]^2}{\left[\mathrm{KL}(q||p_0)\right]^2} + 1.
\end{align*}

Typically for small $\alpha \ll 1$, from \cref{prop::upper_bound_on_g}, we can simplify the above upper bound as
\begin{align*}
   \max\left\{ \mathcal{J}_L(N_\MSR^*), \mathcal{J}_L(N_\MCS^*) \right\}
    \lessapprox\frac{\inf_{\eta > 1} \eta \left[\log(1/\alpha) + \log\left(1 + \left\lceil \log_\eta \frac{\mathrm{KL}(q_U||p_0)}{\mathrm{KL}(q_L||p_0)}\right\rceil \right)\right]}{\mathrm{KL}(q||p_0)},
\end{align*}
which matches the rate of the worst average delays, $O\left(\log(1/\alpha) / \KL(q||p_0)\right)$ of the oracle sequential change detection procedure as $\alpha \to 0$.

\begin{figure*}
    \begin{center}
    \includegraphics[scale =  0.5]{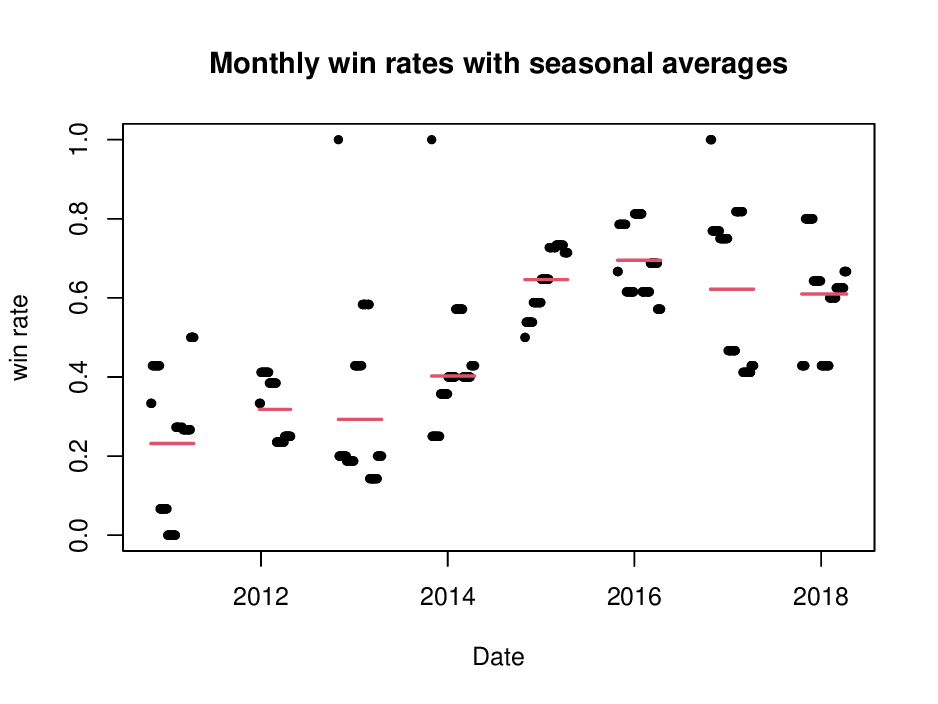}%
    \includegraphics[scale =  0.5]{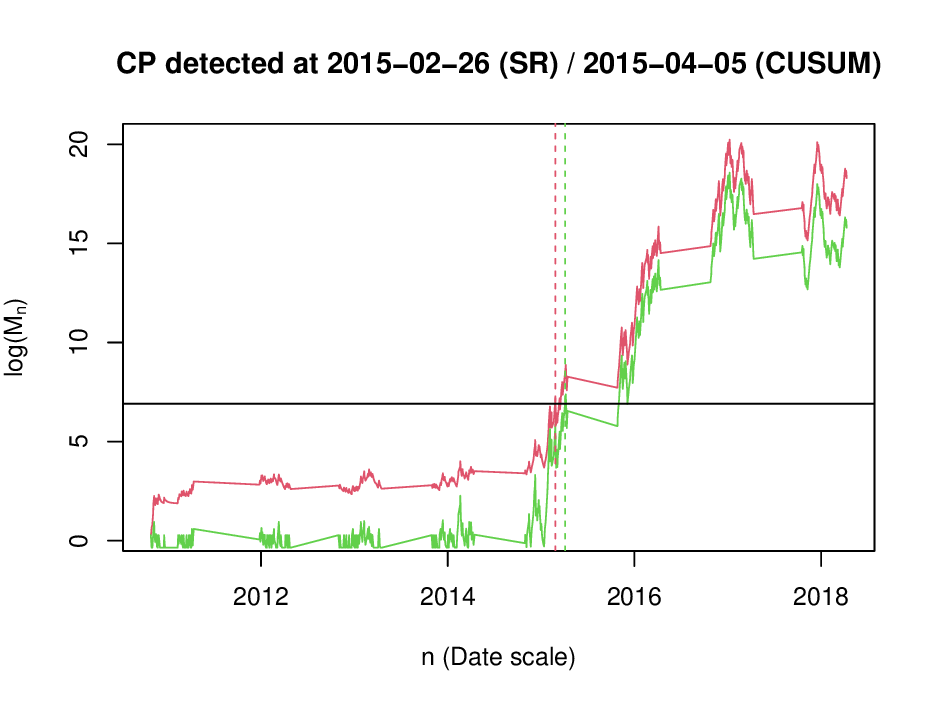}
    \end{center}
    \caption{\em Left: Monthly win rates of the Cavaliers from 2010-11 to 2017-18 seasons (the raw data is Bernoulli, which is harder to visualize). Each red line corresponds to the seasonal average. Right: Paths of log e-detectors (SR: red; CUSUM: green). The horizontal line is the threshold (common to both procedures) equal to $\log(1/\alpha)$, ensuring that the ARL is at least $1/\alpha = 10^{3}$, larger than the number of games in 12 seasons (82 per season). The e-SR procedure detects a changepoint during the 2014-15 season.}
    \label{fig::cavs_win_rate}
\end{figure*}

\paragraph*{Implementation of \cref{alg::mixture_SR} and its results.} The lower bound $q_L$ can be chosen as $q_0 = 0.51$ since it is the minimum winning rate we consider as a significant improvement from before the changepoint, when the rates are upper bounded by $p_0 = 0.49$. We can also safely assume that the win rate cannot be too  high given the competitiveness of the NBA, so that the improved win rates cannot be larger than $0.9$. In our framework, these considerations can be encoded by setting $\Delta_L :=q_0 - p_0 =0.02$ and $\Delta_U := 0.41$ as input parameters of \cref{alg::mixture_SR}. As in~\cref{subSec::intro_Cav_example}, we set $\alpha := 10^{-3}$ to ensure that the ARL is at least $1/\alpha := 10^3$, which is more than the total number of games over 12 years of regular seasons. 
Finally, we set the maximum number of baselines $K_{\max} := 1000$. In fact, the \texttt{computeBaseline} function of \cref{alg::compute_base} returns only 69 baseline processes, and thus the resulting mixtures of e-SR and e-CUSUM procedures of \cref{alg::mixture_SR} can be computed efficiently in an online fashion.

The right plot in \cref{fig::cavs_win_rate} presents the log e-detector values using mixtures of e-SR (red) and e-CUSUM (green) procedures. Although there were a few months in which monthly win rates were higher than $p_0$, overall log e-detectors remained at a stable level over the first four seasons. However, after the 2014-15 season starts, the  log e-detectors increase rapidly and both procedures detect a changepoint during the 2014-15 season, which is the season that marked the return of LeBron James to the Cavaliers.

\subsection{Mean-shift detection in general bounded random variables} \label{subSec::bounded_rv}
\paragraph*{Plus-Minus of the Cavaliers revisited.}
We return to the Cavaliers 2011-2018 example from~\cref{subSec::intro_Cav_example}. Let $\tilde X_1, \tilde X_2,\dots$ be the sequence of Plus-Minus stats from each game. We assume that the average Plus-Minus of the team is less than or equal to $\mu_{<} := -1$ before the changepoint (if any), while after the changepoint it is greater than $\mu_{>}:= 1$. Here, the gap $|\mu_> - \mu_<|$ between averages of Plus-Minus in pre- and post-changes refers to the degree of improvement we consider as significant. 

For convenience, we first normalize the observed sequence. We assume that the absolute value of each Plus-Minus is bounded by $80$, meaning that no team beats another by over 80 points (such an extreme game has never happened in NBA history). Accordingly, define the normalized Plus-Minus, $X_n := (\tilde X_n + 80) / 160 \in [0,1]$ for each $n$. Then, the pre-change observations have conditional mean at most $\mep := (\mu_< + 80) / 160 = 0.494$ and the minimum gap to detect is equal to $\delta := |\mu_> - \mu_<| / 160 = 0.0125$.

\paragraph*{Modeling plus-minus stats as a bounded sequence with time-varying, dependent means.} After the normalization above, the Plus-Minus stats form sequence of bounded random variables $X_1, X_2,...$ on $[0,1]$. Each observation may have different distribution (due to seasonal effects, injuries, form, etc.), but we assume that all observations before an unknown changepoint $\nu$ have a mean less than or equal to a known boundary $\mep \in (0,1)$, when conditioned on the past sample history. That is, under any pre-change distribution $P$, we have $\mu_n := \mathbb{E}_{P,\infty}\left[X_n\mid \mathcal{F}_{n-1}\right] \leq m,~~\forall n\geq 1$. In other words, we use
\[
\mathcal{P} := \{ P: \mu_n \leq m, \forall n \geq 1 \},
\]
where other characteristics about $P$ (outside of its sequence of conditional means) are irrelevant. But after the changepoint, all observations have (conditional) mean larger than the boundary $\mep$ with the minimum gap equal to $\delta$. Thus,
\[
\mathcal{Q} := \{ P: \mu_n \geq m+\delta, \forall n \geq 1 \},
\]

To build an e-SR or e-CUSUM procedure, we need to choose a baseline increment. To derive it, we first consider a simplified setting where both pre- and post-change observations are independently and identically distributed with $\mathbb{E}_{P,\infty}[X] \leq m$ and $\mathbb{E}_{0,Q}[X] \geq m+\delta$, respectively. In this simplified case, we simply refer $P$ and $Q$ to marginal pre- and post-change distributions and $\mathcal{P}$ and $\mathcal{Q}$ to their collections. Then, define $\KL_{\inf}(Q;\mep) := \inf_{P \in \mathcal{P}} \KL(Q, P)$ to be the smallest KL divergence between $Q$ and $\mathcal P$. It is known (see, e.g., \citep{honda2010asymptotically, honda2015non}) that $\KL_{\inf}$ has the following variational representation:
\begin{equation}
   \KL_{\inf}(Q, \mep) =  \sup_{\lambda \in (0,1)}\mathbb{E}_{0,Q}\log\left(1 + \lambda\left(\frac{X}{\mep} - 1\right)\right) =: D(Q||\mathcal{P}).
\end{equation}
Accordingly, for each $\lambda \in (0,1)$, define the baseline increment $L^\lambda := \{L_n\}_{n \geq 1}$ as
\begin{equation}\label{eq:baseline.n=bounded.mean}
    L_n^\lambda := 1 + \lambda \left(\frac{X_n}{\mep} - 1\right),
\end{equation}
for each $n \in \mathbb{N}$. Though the baseline increment above has been derived in the simplified iid setting, it can be checked that $L^\lambda$ is also a valid baseline increment for the general time-varying, dependent means case since it is nonnegative whenever $X_n,\mep \in [0,1]$ as assumed in our setup, and for each pre-change distribution $P \in \mathcal{P}$, we have
\begin{align*}
\mathbb{E}_{P,\infty}\left[L_n^\lambda \mid \mathcal{F}_{n-1}\right] 
    &=  1 + \lambda \left(\frac{ \mathbb{E}_{P,\infty}\left[X_n \mid \mathcal{F}_{n-1}\right]}{\mep} - 1\right) \leq 1,
\end{align*}
where the inequality comes from the condition $\mu_n \leq \mep$ for any pre-change distribution. 

Interestingly. the the baseline increments in \eqref{eq:baseline.n=bounded.mean} correspond to rescaled increments of the capital process used in \cite{WaudbySmith2020EstimatingMO} to design test martingales for  confidence sequences of means of bounded random variables. Though the expressions are essentially identical, ours was obtained via a variational representation of the KL divergence between distributions of bounded random variables, while the derivation presented in \cite{WaudbySmith2020EstimatingMO} is based on a betting interpretation of hypothesis testing.

For any $Q \in \mathcal Q$, let $\lambda^\op$ be the optimal choice of $\lambda \in [0,1]$ given by
\begin{equation}
    \lambda^\op = \argmax_{\lambda \in [0,1])}\mathbb{E}_{0,Q}\log\left(1 + \lambda\left(\frac{X}{\mep} - 1\right)\right).
\end{equation}
Unfortunately, it is typically difficult to compute the optimal $\lambda^\op$ since it depends on the unknown post-change distribution $Q$ in a complicated way. In this case, we use a sub-exponential lower bound from~\cite{fan2015exponential,howard2021time}, given by
\begin{equation} \label{eq::lower_bound_bounded}
    \tilde{L}_n^\lambda := \exp\left\{\lambda \left(\frac{X_n}{\mep} - 1\right) - \psi_E(\lambda)\left(\frac{X_n}{\mep} - 1\right)^2\right\} \leq 1 + \lambda \left(\frac{X_n}{\mep} - 1\right)  = L_n^\lambda,
\end{equation}
where $\psi_E(\lambda) := -\log(1-\lambda) - \lambda$ for $\lambda \in (0,1)$. For each $\lambda \in (0,1)$, the process $\tilde{L}^\lambda$ is itself a valid exponential baseline increment with $s(x) := x/\mep - 1 $ and $v(x) := (x/\mep -1)^2$.

 The  lower bound in \eqref{eq::lower_bound_bounded} also implies the lower bound 
\begin{equation}
     \KL_{\inf}(Q, \mep) \geq \sup_{\lambda \in [0,1]} \left\{\lambda \mu - \psi_E(\lambda)\sigma^2\right\} = \sigma^2 \psi_E^*\left(\Delta^\op\right), 
\end{equation}
where $\psi_E^*(u):=u-\log(1+u)$ is the convex conjugate of $\psi_E$, while $\mu, \sigma^2$ and $\Delta^\op$   from~\cref{sec::mixture_simple_exponential} are:
\begin{align*}
    \mu &:=\mathbb{E}_{0,Q}s(X)  =\frac{\mathbb{E}_{0,Q}X - \mep}{\mep},\\ \sigma^2&:=\mathbb{E}_{0,Q}v(X) = \frac{ \mathbb{E}_{0,Q}(X-\mep)^2}{\mep^2},\\
    \Delta^\op &:= \frac{\mu}{\sigma^2} =  \frac{\mep\left[\mathbb{E}_{0,Q}X-\mep\right] }{\mathbb{E}_{0,Q}(X-\mep)^2}. 
\end{align*}
Noting that $\psi_E^*(u)\approx u^2/2$ for small $u$, we see that for small $\Delta^\op \ll 1$, one has
\begin{equation}
  \KL_{\inf}(Q, \mep) \gtrsim \frac{\left[\mathbb{E}_{0,Q}X-\mep\right]^2}{2\mathbb{E}_{0,Q}(X-\mep)^2}.
\end{equation}

Note that the oracle $\Delta^\op$ depends on the unknown post-change distribution only via first and second moments. Therefore, in contrast to the original set of baseline increments $\{L^\lambda\}_{\lambda \in (0,1)}$, the exponential baseline increments $\{\tilde{L}^\lambda\}_{\lambda \in (0,1)}$ that lower bound them allow us to more easily set a range $(\Delta_L, \Delta_U)$ to build mixtures of the e-SR and e-CUSUM procedures. For example, if we assume that the post-change distribution has mean at least $\mep + \delta$ for a positive $\delta$ then we can  upper and lower bound $\Delta^\op$ by
\begin{equation} \label{eq::bounds_of_delta_for_bounded_case}
 \Delta_L := \frac{\mep\delta}{(1-\mep)^2} \leq \Delta^\op \leq \frac{\mep(1-\mep)}{\delta^2} =: \Delta_U.   
\end{equation}

Now, given $\Delta_L$ and $\Delta_U$, we can use \cref{alg::mixture_SR} to run the mixture of e-SR or e-CUSUM procedure to detect the changepoint based on the exponential baseline baseline increments $\{\tilde{L}^\lambda\}_{\lambda \in (0,1)}$. It is also straightforward to build the corresponding mixtures of e-SR and e-CUSUM procedures for the original baseline increment $\{L^\lambda\}_{\lambda \in (0,1)}$ which is always more sample-efficient. 

\begin{figure*}
    \begin{center}
    \includegraphics[scale =  0.5]{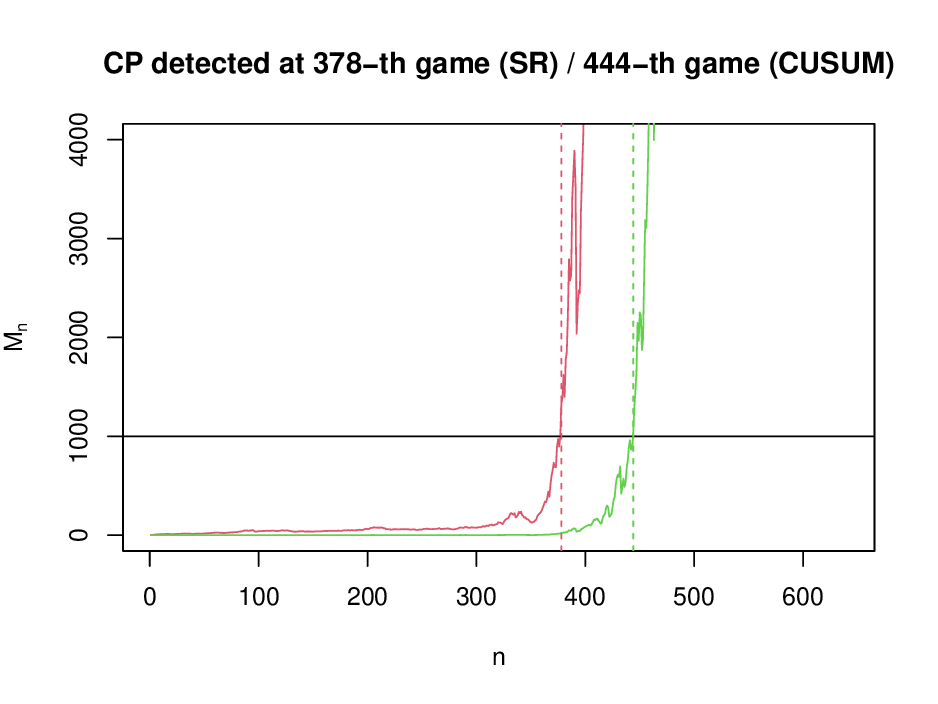}%
    \includegraphics[scale =  0.5]{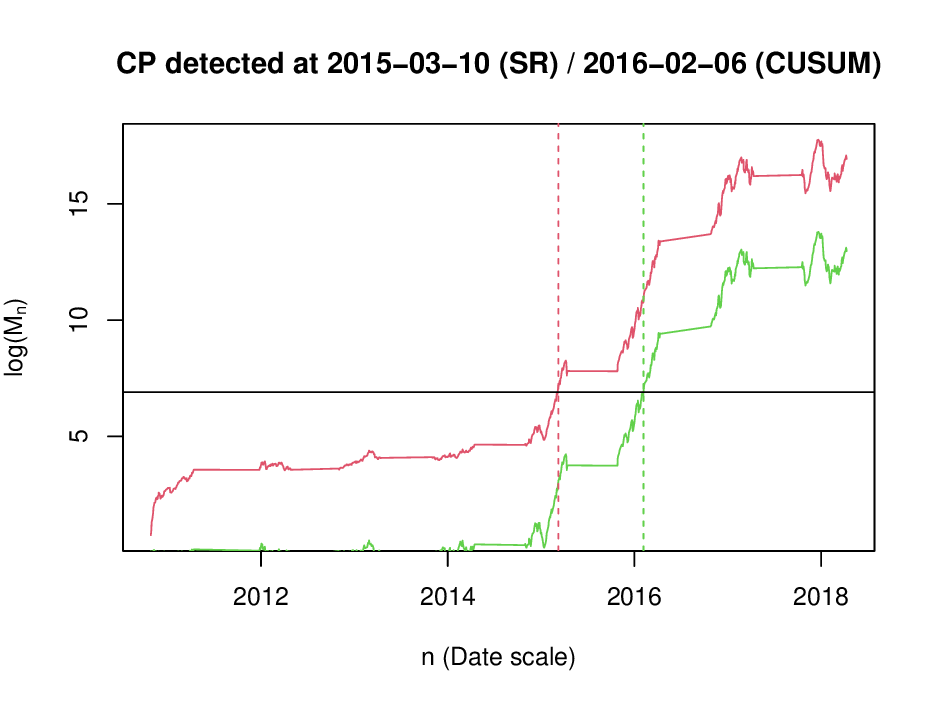}
    \end{center}
    \caption{\em Left: E-detectors (SR: red; CUSUM: green) over eight seasons (82 games per season). Right: Logarithm of e-detectors against date (the sharp rise of e-SR is simply due to the log scale). In both plots, horizontal lines are thresholds equal to $1/\alpha$ (left) and $\log(1/\alpha)$ (right) controlling the ARL by $1/\alpha = 10^{3}$. The e-SR procedure detects a change during the 2014-15 season, while e-CUSUM takes longer (as expected).}
    \label{fig::cavs_score_in_sec5}
\end{figure*}

\paragraph*{Implementation of \cref{alg::mixture_SR} and its results.}

Recall that in the plus-minus stats running example, we use pre-change mean $m = 0.494$ and the minimum gap $\delta = 0.0125$, which bounds $\Delta^\op$ by $\Delta_L :=  0.024$ and $\Delta_U:= \frac{\mep(1-\mep)}{\delta^2} = 1600$. As before, we choose $\alpha = 10^{-3}$ to make the ARL larger than 12 regular seasons and set the maximum number of baselines $K_{\max} = 1000$. Based on these parameters, we can build mixtures of e-SR and e-CUSUM procedures. Though the difference between $\Delta_L$ and $\Delta_U$ may seem to be large, the actual number of baselines returned by~the function \texttt{computeBaseline} in \cref{alg::mixture_SR} is 190, which is small enough to update the procedure efficiently on the fly.

 \cref{fig::cavs_score_in_sec5} shows e-detectors (left) and their logarithms (right). The horizontal line corresponds to the detection boundary given by $1/\alpha$ (left) and $\log\alpha^{-1}$ (right). Similar to the winning rate example, the log e-detectors remained stable during the first four regular seasons, although the difference between SR and CUSUM e-detectors is larger than before. After 2014-15 season started, both e-detectors increased rapidly, and the e-SR procedure detects a changepoint during the 2014-15 season, but e-CUSUM detects the changepoint only in the following season (as expected, since both procedures use the same threshold).

\subsection{Simulation-based comparison with parametric methods} \label{subSec::simulation}

In the Bernoulli example of \cref{subSec::example_bernoulli}, we showed that, in the simple i.i.d\ Bernoulli setup, our mixtures of  e-SR and e-CUSUM procedures match the rate of the worst average delays $O\left(\log(1/\alpha) / \KL(q||p_0)\right)$ of the oracle sequential change detection procedure as $\alpha \to 0$. 
In this subsection, we conduct a simulation study to compare the efficiency of our e-SR procedure with the oracle CUSUM procedure with the exact threshold \citep{moustakides1986optimal, ritov1990decision}, given by
\begin{equation} \label{eq::oracle_CS_formula}
\begin{aligned}
    N^*_\CS &:= \inf\left\{n \geq 1: \max_{0\leq j < n} \sum_{i=j}^n \log \frac{f_{p^*}(X_i)}{f_{p_0}(X_i)} \geq c_{\alpha}^*\right\},
\end{aligned}
\end{equation}
where $p^*$ is the true post-change distribution parameter (hence the oracle designation) and $c^*_\alpha$ is the value of the threshold so that the ARL is exactly $1/\alpha$.
It is well known that the oracle CUSUM procedure (with the appropriate choice of the stopping threshold that controls ARL exactly) minimizes the worst average delay. 

We also compare our method with a version of the GLR procedure based on the stopping time
\begin{equation} \label{eq::glr_formula}
\begin{aligned}
    N^*_\GLR &:= \inf\left\{n \geq 1: \max_{0\leq j < n} \sup_{p}\sum_{i=j}^n \log \frac{f_{p}(X_i)}{f_{p_0}(X_i)} \geq c_{\alpha}\right\} \\
   &= \inf\left\{n \geq 1: \max_{0\leq j < n}\sum_{i=j}^n \log \frac{f_{\hat{p}_{j:n}}(X_i)}{f_{p_0}(X_i)} \geq c_{\alpha}\right\},
\end{aligned}
\end{equation}
where each $\hat{p}_{j:n}$ is the MLE of the post-change parameter and the exact threshold $c_{\alpha}$ is tuned to control ARL exactly at $1/\alpha$ (this is typically only possible in such simple parametric settings, either by analytic derivations or simulations). Unlike the oracle CUSUM procedure, the GLR procedure does not have an iterative update rule, as we need to recompute the MLE of the post-change parameter at each time. As a result, its computational cost at time $n$ is $O(n)$, which makes an online implementation very costly. In practice, we may want to use a window-limited GLR procedure to overcome the computational challenge. However, in our study, we deploy the GLR procedure to avoid the additional challenge of picking a window size. 

\begin{figure*}
    \begin{center}
    \includegraphics[scale =  0.6]{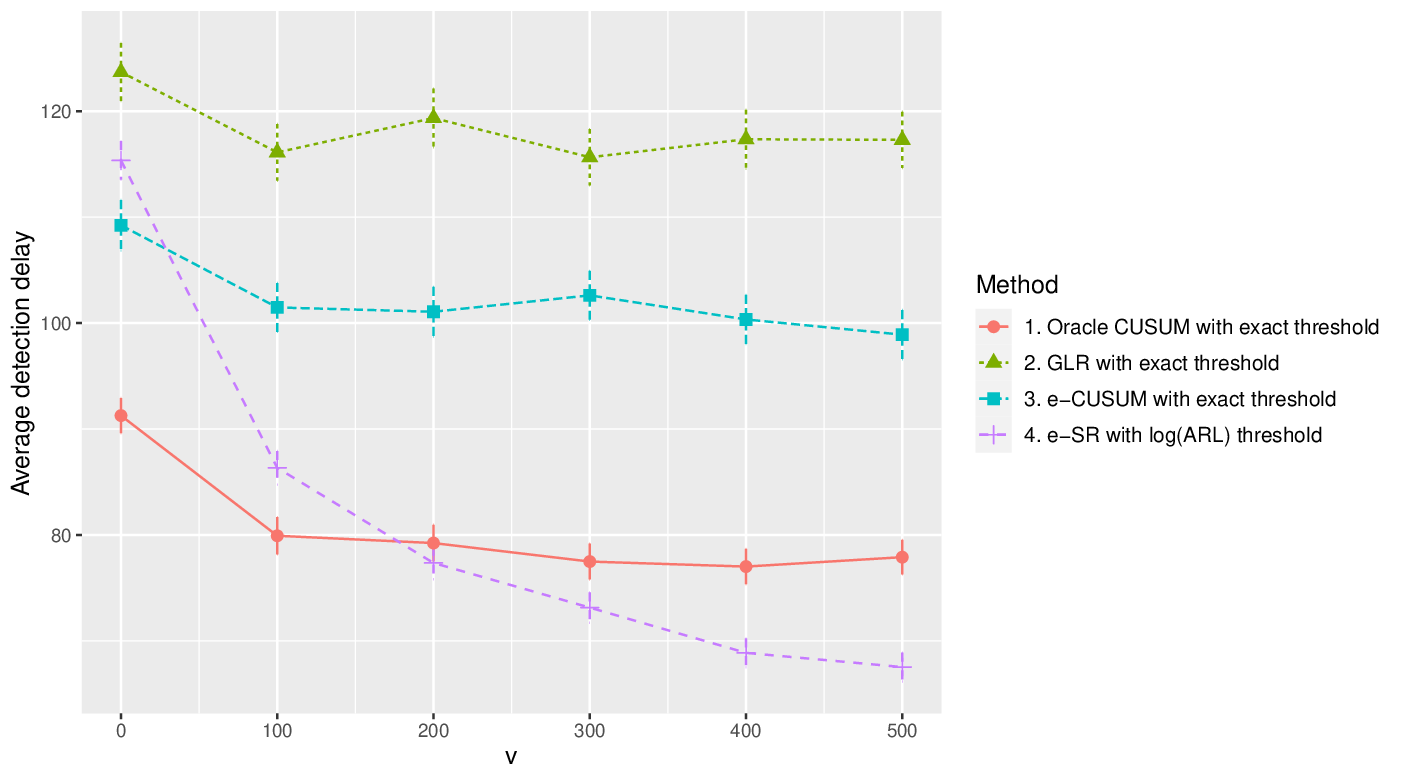}
    \end{center}
    \caption{\em  Average detection delay for each changepoint $\nu=0, 100, \dots, 500$ (each experiment has exactly one changepoint at $\nu$). Three of the methods use an exact threshold calculated via simulation (only possible in this simple, parametric example). Only the Oracle CUSUM method knows the post-change distribution. Even though e-SR uses the conservative $\log(1/\alpha)$ threshold, its detection delay is excellent, often even better than the Oracle CUSUM method (which has optimal average \emph{worst-case} (across $\nu$) delay).}
    \label{fig::simulation_wad}
\end{figure*}

\begin{figure*}
    \begin{center}
    \includegraphics[scale =  0.6]{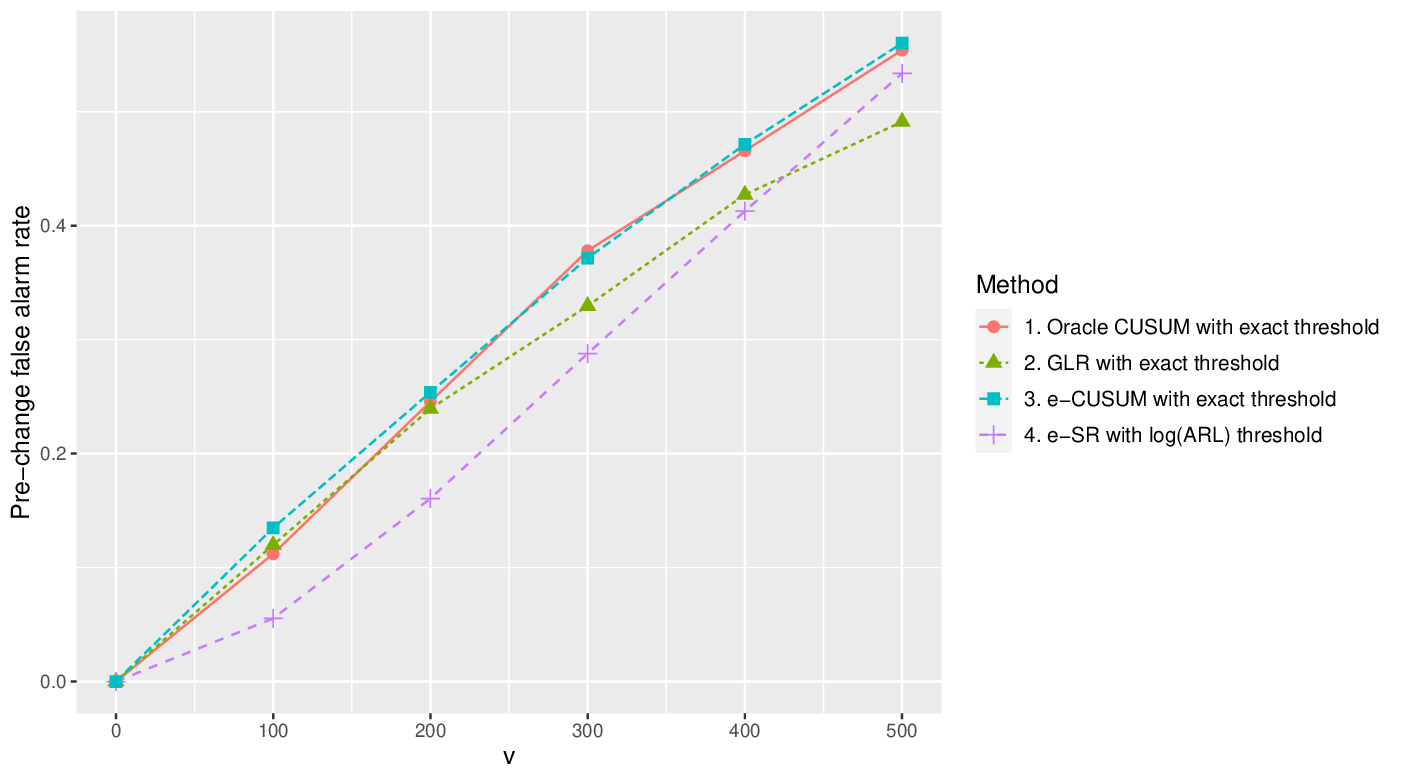}
    \end{center}
    \caption{\em Pre-change false alarm rates of detection methods for each changepoint $\nu=0, 100, \dots, 500$ (each experiment has exactly one changepoint at $\nu$). Three of the methods use an exact threshold calculated via simulation (only possible in this simple, parametric example). Only the Oracle CUSUM method knows the post-change distribution exactly. e-SR has the smallest false alert ratio (defined in the text).}
    \label{fig::simulation_false_alert_ratio}
\end{figure*}

\paragraph*{Simulation details.}
Throughout this simulation, we draw pre-change observations as iid Bernoulli random variables with $p_0 = 0.5$, and post-change observations using $p_1 = 0.6$. For non-oracle methods, we will only assume that the post-change parameter is known to be in the interval $[0.51, 0.99]$. The e-SR and e-CUSUM procedures in \cref{alg::mixture_SR} will use this range to set $\Delta_{L}:=0.01$ and $\Delta_{U}:= 0.49$. For the GLR procedure, the MLE of the post-change parameter is
\begin{equation}
    \hat{p}_{j:n} := \min\left\{\max\left\{\bar{X}_{j:n}, 0.51\right\}, 0.99\right\},
\end{equation}
where $\bar{X}_{j:n}$ is the sample average over last $n-j+1$ observations. Our ARL target is equal to $1/\alpha := 500$, and each simulation is repeated  $5000$ times to estimate average delays. The time of the changepoint $\nu$ varies in $\{0, 100, \dots, 500\}$. For simplicity, each run will end no later than time $n=1000$. 

For the oracle CUSUM, GLR, and e-CUSUM procedures (but not e-SR), we use the same simulation setup to find the exact threshold value that controls the ARL exactly $1/\alpha$ for each method. For the e-SR procedure, we simply use the universal threshold of $\log(1/\alpha)$ to demonstrate its efficiency. In particular, practitioners are not required to resort to expensive simulation to identify a good threshold value.

\cref{fig::simulation_wad} shows the average delays of oracle CUSUM, GLR, e-CUSUM and e-SR procedures for each changepoint $\nu\in \{0, 100, \dots, 500\}$ (each experiment has only one change at $\nu$). The vertical bar on each point represents 95\% confidence interval of the average delay. As the theory guarantees, the oracle CUSUM procedure with an exact threshold results in the smallest worst average delay ($91.3 \pm 1.7$), 
while surprisingly the GLR procedure with an exact threshold shows the largest worst average delay ($123.7 \pm 2.7$) despite its high computational complexity. The two e-detectors perform reasonably well, and the e-SR detector in particular performs quite favorably overall despite using its conservative $\log(1/\alpha)$ threshold. 

As the value of the changepoint approaches the ARL target of $500$, the average delays tend to decrease sharply for e-SR procedure, even falling below the oracle CUSUM method. This is plausibly because e-SR sums the underlying e-processes, in contrast to the other CUSUM-style procedures which take the maximum of the underlying e-processes. We may also intuitively expect the oracle CUSUM delay to be relatively flat across $\nu$ because it is known to be minimax optimal, in the sense of minimizing the worst case delay, and minimax procedures often have constant ``risk profiles''.  Proving these fine-grained behaviors is beyond the scope of the current paper.

\cref{fig::simulation_false_alert_ratio} illustrates the ``pre-change false alarm rate'': the fraction of simulation runs in which the detection procedures stopped \emph{before} the changepoint at $\nu$ (if a change occurs at $\nu = 0$, then it is zero by definition). This is not a common metric, since we provably control ARL at the target level. However, it is an interesting metric, so we plot it.  As the time of the true changepoint $\nu$ becomes closer to the ARL target of 500, false alarm rates increase across all methods. We notice that the e-SR procedure with the $\log(1/\alpha)$ threshold results in the smallest false alarm rates in most cases.

\section{Discussion}

\subsection{Game-theoretic interpretation of an e-detector} 
We briefly mention here a game-theoretic interpretation of an e-detector along the lines of the game-theoretic interpretations of martingales and supermartingales as the wealth of a gambler playing a fair game (well known since the time of~\cite{ville1939study}).
We first summarize the game-theoretic interpretation of a $\mathcal P$-e-process, as described in~\cite{ramdas2021testing}. 

The standard game-theoretic setup of~\cite{shafer2019game} involves three players: a forecaster, a skeptic, and reality. The forecaster claims at the beginning that $\mathcal P$ is a plausible model for the yet-to-be-observed data; meaning that the observations are in accordance with (or generated by) some $P \in \mathcal P$. The skeptic plays (in parallel) a family of games indexed by $P\in\mathcal P$ against nature and begins with one dollar in each game. The objective of the skeptic in the $P$-th game is to sequentially test whether $P$ is a good explanation for the data by betting against $P$. At each time step, the skeptic places fair bets (relative to $P$, in the $P$-th game) about the next outcome. Then nature reveals the next outcome, and the skeptic's wealth in every game is updated. The magnitude of the skeptic's wealth in the $P$-th game is direct evidence against $P$ being a good explanation; the higher the wealth, the more unlikely the data came from $P$. Thus in each game, the gambler places different bets, but nature's moves (the outcomes) are  identical across all games. The skeptic's overall evidence against $\mathcal P$ is measured by their \emph{worst} wealth across all the games. If this evidence exceeds $1/\alpha$, it means that the skeptic multiplied their initial capital by at least $1/\alpha$ in every game, and if we reject $\mathcal P$ when this happens, Ville's inequality implies that we have a valid level-$\alpha$ sequential test. 

Since our e-detectors are constructed to be cumulative sums of e-processes started at consecutive times, their game-theoretic interpretation builds on the aforementioned one. Informally, the forecaster not only claims that the data sequence follows $\mathcal P$ from the start, but that this will not change after some amount of time. The skeptic now wishes to detect a change, if one occurs, as soon as possible. To accomplish this task, the skeptic is provided with one extra dollar every day that they invest (using a $\mathcal P$ e-process) into testing whether the data from that day onwards is still explained well by $\mathcal P$.  E-detectors use the wealth in all these games (one against each $P \in \mathcal P$, starting at each time) as a measure of evidence against the forecaster's claims. The SR e-detector uses the sum (across time) of the minimum wealth (across $P$ at each time), though it could use the amount that this wealth exceeds $n$, which is the total dollar amount invested up to time $n$. The CUSUM e-detector uses the max-min wealth; the maximum (across time) of the minimum wealth (across $P$). These are only two ways of constructing e-detectors, and we leave other constructions to future work.

\subsection{Viewing Lorden's reduction to sequential testing as an e-detector}

Lorden \citep{lorden1971procedures} proposed a simple method to construct a change detection method with ARL control via a reduction to sequential testing. We describe this below, first defining a sequential test formally.

A {\it sequential test} $\phi$ is a mapping from increasing amounts of data to a sequence of zeros and ones, where a one represents a rejection of the null hypothesis, and a zero means ``continue collecting data''. Formally, define the decision at time $t$ as $\phi_t: \mathcal X^t \to \{0,1\}$, and let $\phi := \{\phi_t\}_{t \geq 0}$ be the collection of such decisions made one at a time based on the first $t$ datapoints, with $\phi_0=0$ by default. The sequence of tests $\phi$ is called a level-$\alpha$ sequential test for $\mathcal{P}$ if 
\[
\sup_{P \in \mathcal P} P(\exists t \geq 1: \phi_t(X_1,\dots,X_t) = 1) \leq \alpha ,
\]
i.e. if the probability of ever falsely rejecting the null is at most $\alpha$. By convention, if $\phi_t = 1$, we set $\phi_s = 1$ for $s > t$. This is equivalent to requiring that, for each $P \in \mathcal P$,
\[
 P(\phi_\tau(X_1,\dots,X_\tau) = 1) \leq \alpha \text{ for any stopping time }\tau.
\]

Let $\phi^{(j)}$ denote a sequential test is started at time $j$; that is, for $\phi^{(j)}$, the first observation is actually $X_j$ (and not $X_1$), but the test itself can depend on the first $j-1$ observations (for example, we can choose our betting strategy based the first $j-1$ points, even though our betting score will only be evaluated from time $j$ onwards). Note that $\phi^{(1)}$ is simply a standard sequential test as defined above. Ideally, these tests are powerful against alternative $\mathcal Q$.

Lorden's change detection procedure is simple and works as follows. At each time $t$, start a new sequential test $\phi^{(t)}$, in addition to the ones that are already running. In other words, consider a sequence of level-$\alpha$ sequential tests $\phi^{(1)}, \phi^{(2)}, \dots$, starting at consecutive times. Lorden declares a change if any of those sequential tests rejects the null $\mathcal P$:
\begin{align*}
    &N^{\text{Lorden}}:= \inf\{n \geq 1:\max_{1\leq j\leq n}\phi^{(j)}_n(X_{n-j+1},\dots,X_n) = 1 \}.
\end{align*}
Lorden proved that this method controls the ARL at $1/\alpha$ if the data are iid, and if  the same test $\phi^{(j)}$ is deployed at each $j$ (i.e.\ apart from the delayed start, the tests are identical).

We first observe that Lorden's method is a special case of an e-detector. Indeed, with each level-$\alpha$ sequential test $\phi \equiv \{\phi_t\}_{t \geq 0}$, we can associate an e-process $\Lambda^{\text{Lorden}}_t := \frac{\mathbf{1}(\phi_t = 1)}{\alpha} = \frac{\phi_t}{\alpha}$. Note that $\Lambda$ only takes on two values: $0,1/\alpha$, and when it reaches the latter, it stays there. Furthermore, note that $\mathbb{E}[\Lambda^{\text{Lorden}}_\tau] \leq 1$ for any stopping time $\tau$, which makes it an e-process as claimed.

Last, note that if we form a ``Lorden e-detector'' using either the SR or CUSUM methods in~\eqref{eq:SER:CS}, then both e-detectors start out at 0, and either e-detector jumps to level $1/\alpha$ if and only if one of the (delayed start) sequential tests rejects the null, and further our e-detector declares a changepoint at exactly the same instant that Lorden's does. Thus, Lorden's procedure can be subsumed within our e-detector framework without any loss of generality or performance. 

There are two benefits to viewing Lorden's method as an e-detector. First, we can dispense with both the aforementioned conditions that Lorden required to prove ARL control: the iid\ assumption, and the condition that the underlying tests $\phi^{(j)}$ are identical across $j$. Indeed, our main result guarantees ARL control for any e-detector no matter what the underlying $e_j$-processes are, or whether the data are iid\ or not.

Second, this viewpoint allows us to see why e-detectors could have much smaller detection delay than Lorden's method (that is, Lorden's e-detector). In Lorden's e-detector, there is no sharing of evidence across different $e_j$-processes: each sequential test acts alone without help from the others, and we need a single $e_j$-process to reach $1/\alpha$ before we can declare a change. When a general (say SR-type) e-detector crosses $1/\alpha$, the reason it does so will usually be because of a collaborative effect across various e-processes (caused by the nontrivial sum of $e_j$-processes in the definition of the e-detector), none of which have yet individually reached $1/\alpha$. This will happen much sooner than any individual one reaches that threshold, causing an earlier detection than Lorden's method. In fact, every level-$\alpha$ sequential test in some sense must be based on threshold an e-process at level $1/\alpha$ \citep{ramdas2020admissible} and using our e-detector with those underlying e-processes will be more statistically efficient (shorter delay) than directly using Lorden's reduction.

For the sake of future reference, we summarize the above observations below in a ``generalized Lorden's lemma'', whose proof follows immediately from the properties of an e-detector and the discussion above.

\begin{lemma}[Generalized Lorden's Lemma]
    Suppose the data initially come from a distribution in the pre-change class $\mathcal P$ and, if a change occurs, they later come from a distribution in the post-change class $\mathcal Q$ (note the lack of any iid\ assumption). For each $j$, let $\phi^{(j)}$ denote a (one-sided) level-$\alpha$ sequential test for $\mathcal P$ against $\mathcal Q$ that is started at time $j$ (but need not be identical or related in any way to any other $\phi^{(k)}$, for $k \neq j$). If we declare a change at the first time when any one of these sequential tests rejects the null, the resulting change detection procedure has ARL at most $1/\alpha$. Further, this generalization of Lorden's change detector is a special case of an e-detector.
\end{lemma}

\subsection{Future directions} 
There remain a whole host of follow-up directions; we mention only a few below. First, our sequential change detection framework can be straightforwardly generalized to the multi-stream setting where we are monitoring a large number of data streams. In the classical parametric setting, minimum or summation of local CUSUM statistics for multi-stream data were proposed and their asymptotic optimality was studied \citep{hadjiliadis2009one, mei2010efficient}. Since either minimum or scaled summation (average) of e-detectors also forms a valid e-detector, we can apply the framework in this paper to the multi-stream setting seamlessly. It is interesting to investigate how the framework can be even further generalized to structural multi-stream settings \citep{xie2013sequential, zou2019quickest, chen2020high}.

The kernel sequential change detection is an important class of sequential change detection methods \citep{desobry2005online,NIPS2015_eb1e7832}. 
It is an interesting open direction how to instantiate  existing kernel-based methods into our general framework to make it possible to analyze kernel sequential change detection algorithms in a nonasymptotic way. As it  currently stands, neither framework is more general than the other, because the kernel methods often assume iid data before the changepoint, while we abstain from such strong assumptions.

Last, throughout this paper, we have only focused on detecting \emph{whether} a changepoint happened or not but have not dealt with inferential questions surrounding \emph{when} the change occurred. 
Future work could study how to perform such inference with e-detectors in our nonparametric settings, either online or post-hoc.

\section{Summary} 

We have presented a general framework for sequential change detection based on a new concept called e-detectors. The proposed framework is nonparametric as it does not rely on a parametric assumption on the data-generating distribution (though, when such assumptions are made, we recover well-known parametric methods as special cases). Also, the framework comes with nonasymptotic guarantees, since every component of the framework can be chosen and analyzed explicitly without any asymptotic approximations. By introducing additional structures such as baseline increments and exponential e-detectors on the top of the general framework, we can construct computationally and statistically efficient online algorithms that have explicit upper bounds on worst average delays. Finally, through examples involving Bernoulli and bounded random variables, we explained how one can apply the presented framework in practical settings, with NBA data serving as a running case study.

\subsection*{Acknowledgments}
A. Ramdas was partially supported by NSF grants IIS-2229881 and  DMS-2310718. J. Shin and A. Rinaldo were partially supported by NSF grant DMS-EPSRC 2015489. 

\bibliography{myref}
\bibliographystyle{plainnat}

\newpage

\appendix

\section{Main proofs}

\subsection{Proofs for statements in Section~\ref{sec::general}} \label{appen::proofs_for_general}
 
\begin{proof}[Proof of \cref{prop::ARL_control_e_detector} (ARL control)]
For any given $\alpha \in (0,1)$ let $N^*$ be the stopping time defined by
\begin{equation}
    N^* := \inf\left\{n \geq 1 : M_n \geq 1/\alpha\right\}. 
\end{equation}
For any pre-change distribution $P \in \mathcal{P}$,  we may assume $N^* < \infty$ with probability one under $P$, without loss of generality. (If not, we have $\mathbb{E}_{P, \infty} N^*= \infty$, which immediately proves the claim.) Then, from the definition of an e-detector, we have
\begin{equation}
    \mathbb{E}_{P, \infty} \left[ M_{N^*}\right] \leq \mathbb{E}_{P, \infty}  N^* ,
\end{equation}
which implies that, for any $P \in \mathcal{P}$,
\begin{align*}
     \mathbb{E}_{P,\infty} N^* &\geq  \mathbb{E}_{P, \infty}\left[ M_{N^*}\right]\\
     & = \mathbb{E}_{P,\infty}  \left[ M_{N^*} \mathbbm{1}\left(N^* < \infty\right)\right]~~
     \\
     & \geq  \mathbb{E}_{P,\infty} \left[ \frac{1}{\alpha}\mathbbm{1}(N^* < \infty)\right]~~~  ~ ~ ~ \text{(by definition of $N^*$)} \\
     & = \frac{1}{\alpha} \mathbb{P}_{P,\infty}\left(N^* < \infty \right) 
     ~ = \frac{1}{\alpha},
\end{align*}
 as desired. 
\end{proof}

 \begin{proof}[Proof of \cref{prop::upperbound_simple_form} (Bounds on worst average delay)]
  We first prove \eqref{eq::delay_SR}. Note that 
 \begin{align*}
     N_\SR^* &= \inf \left\{n \geq 1 :  \sum_{j=1}^n  \prod_{i = j}^n L_i\geq 1 /\alpha \right\}\\
     & \leq \min_{j \geq 1}\inf \left\{n \geq j : \prod_{i = j}^n L_i\geq 1 /\alpha \right\}  := \min_{j \geq 1} N_j.
 \end{align*}
 For each fixed $P \in \mathcal{P}$ and $Q \in \mathcal{Q}$, since $N_j \indep \mathcal{F}_{\nu}$ for all $j > \nu + m $, we have 
 \begin{align*}
     & \mathbb{E}_{P,\nu,Q} \left[[N_\SR^* - \nu]_+ \mid \mathcal{F}_\nu \right] \\
     & \leq  \min_{j \geq 1} \mathbb{E}_{P,\nu,Q} \left[[N_j - \nu]_+ \mid \mathcal{F}_\nu \right] \\
      & \leq  \min_{j > \nu + m} \mathbb{E}_{P,\nu,Q} \left[N_j -\nu  \mid \mathcal{F}_\nu \right] ~~~~\text{(Since $N_j \geq j > \nu+m$)}\\
       & =  \min_{j > \nu + m} \mathbb{E}_{P,\nu,Q} \left[N_j\right] - \nu  ~~~~\text{(Since $N_j \indep \mathcal{F}_{\nu}, \forall j > \nu + m $)} \\  
      & = \min_{j > \nu + m}  \mathbb{E}_{P,\nu,Q} \left[ \inf \left\{n \geq j : \prod_{i = j}^n L_i \geq 1 /\alpha \right\}\right] - \nu\\ 
    & = \min_{j -\nu>  m}  \mathbb{E}_{0,Q} \left[ \inf \left\{n \geq j - \nu : \prod_{i = j-\nu}^n L_i \geq 1 /\alpha \right\}\right]\\
      & = \min_{j'> m}  \mathbb{E}_{0,Q} \left[ \inf \left\{n \geq j' : \prod_{i = j'}^n L_i \geq 1 /\alpha \right\}\right] ~~\text{(by $j' := j -\nu)$}\\
      & = \min_{j'> m}  \mathbb{E}_{0,Q} \left[ \inf \left\{n \geq  1 : \prod_{i = 1}^{n} L_i \geq 1 /\alpha \right\}\right] + j' - 1 \\
        & = \min_{j'> m}  \mathbb{E}_{0,Q} N_{1/\alpha} + j' - 1 \\
        & =  \mathbb{E}_{0,Q} N_{1/\alpha} + m,
 \end{align*}
where the third equality comes from the fact that the distribution of $X_{j-m}, X_{j-m+1}, \dots$ under $\mathbb{P}_{P, \nu, Q}$ is equal to the one of $X_{j-\nu-m}, X_{j-\nu-m+1}, \dots$ under $\mathbb{P}_{0,Q}$ provided that $j-m > \nu$, and the fifth equality is based on the strong stationarity of the post-change observations. Since the last term does not depend on $P$, $\nu$ or $\mathcal{F}_\nu$, we obtain the claimed result
 \begin{align*}
      \mathcal{J}_P(\st_\SR^*) \leq \mathcal{J}_L(\st_\SR^*)= \sup_{P \in \mathcal{P}, \nu \geq 0} \esssup  \mathbb{E}_{P,\nu,Q} \left[[\st_\SR^* - \nu]_+ \mid \mathcal{F}_\nu \right] \leq  \mathbb{E}_{0,Q} N_{1/\alpha} + m,
  \end{align*}
 as desired. 
 
 To prove \eqref{eq::delay_CS}, first note that 
  \begin{align*}
     N_\CS^* &= \inf \left\{n \geq 1 :  \max_{j \in [n]} \prod_{i = j}^n L_i   \geq c_\alpha \right\}\\
     & \leq \min_{j \geq 1}\inf \left\{n \geq j :\prod_{i = j}^n L_i \geq c_\alpha \right\} \\
     & := \min_{j \geq 1} N_j.
 \end{align*}
  The remaining part of the proof for \eqref{eq::delay_CS} is followed by the same argument for \eqref{eq::delay_SR}. 
 
 Finally, to prove \eqref{eq::simple_upper_SR}, it is enough to show the following inequality holds:
 \begin{equation}
 \label{eq::simple_upper_general}
      \mathbb{E}_{0,Q} N(g) \leq \frac{g}{\mathbb{E}_{0,Q}\log L_1} + \frac{\mathbb{V}_{0,Q} \log L_1}{\left[\mathbb{E}_{0,Q}\log L_1\right]^2} + 1,~~\forall g > 0,
 \end{equation}
 where $N(g) := \inf \left\{n \geq 1 : \sum_{i=1}^n \log L_i \geq g\right\}$ for each $g > 0$. The proof of the above upper bound \eqref{eq::simple_upper_general} is based on the Lorden's inequality \cite{lorden1970excess} which can be stated as follows:
\begin{fact}[Lorden's inequality \cite{lorden1970excess}]
Suppose $X_1, X_2, \dots$ are i.i.d.\ samples with $\mathbb{E}X_1 = \mu > 0$ and $\mathbb{E}X_1^2 <\infty$. For each $g >0$, set $N(g) := \inf\left\{\n: S_n := \sum_{i=1}^\n X_i \geq g \right\}$ and $R_g := S_{N(g)} - g$. Then, the following inequality holds:
\begin{equation}
\sup_{g > 0}\mathbb{E} \left[ R_g\right] \leq \frac{\mathbb{E}X_1^2}{\mu} = \mu + \frac{\sigma^2}{\mu}, 
\end{equation}
where $\sigma^2 = \mathbb{V}X_1$.
\end{fact}
 
Now, to prove the upper bound~\eqref{eq::simple_upper_general}, fix a constant $g > 0$. Since $\mathbb{E}_{0,Q}\log L_1 > 0$, we have $\mathbb{E}_{0,Q}N(g) < \infty$. Therefore, by Wald's equation, 
\begin{equation}
    \mathbb{E}_{0,Q}\log L_1 \mathbb{E}_{0,Q}N(g) 
    = \mathbb{E}_{0,Q} \left[\sum_{i=1}^{N(g)}\log L_i\right].
\end{equation}
For each $g > 0$, set $R_g := \sum_{i=1}^{N(g)}\log L_i - g$. Then, from the Lorden's inequality, we have
\begin{align}
    \mathbb{E}_{0,Q}\log L_1 \mathbb{E}_{0,Q}N(g) 
    &= \mathbb{E}_{0,Q} \left[\sum_{i=1}^{N(g)}\log L_i\right]\nonumber\\
    & \leq c + \sup_{g > 0}\mathbb{E}_{0,Q} \left[ R_g\right]\nonumber \\
    & \leq g + \frac{\mathbb{V}_{0,Q} \log L_1}{\mathbb{E}_{0,Q}\log L_1} + \mathbb{E}_{0,Q}\log L_1, \label{eq::sample_complex_inner}
\end{align}
where the first inequality comes from the definition of $N(g)$. By multiplying $1  /\mathbb{E}_{0,Q}\log L_1$ on both sides of the inequality~\eqref{eq::sample_complex_inner}, we have the claimed upper bound, completing the proof. 
\end{proof}

\subsection{Proofs for statements in Section~\ref{sec::mixture_general}} \label{appen::proofs_for_mixture_general}

 \begin{proof}[Proof of \cref{prop::adaptive_schme_yields_valid_e_detectors} (Validity of adaptive e-detectors)]
    To see adaptive SR and CUSUM e-detectors are actually valid e-detectors, first note that $M_n^\ACS \leq M_n^\ASR$ for each $n \in \mathbb{N}$. Therefore, it is enough to show that $\mathbb{E}_{P,\infty} M_{\tau}^\ASR \leq \mathbb{E}_{P,\infty} \tau$ for any stopping time $\tau$ and pre-change distribution $P \in \mathcal{P}$. If $\mathbb{P}_{P,\infty}(\tau = \infty) > 0$ then the above inequality holds trivially. Otherwise if $\mathbb{P}_{P,\infty}(\tau = \infty) = 0$ , we have that
    \begin{align*}
        \mathbb{E}_{P,\infty} M_{\tau}^\ASR &= \mathbb{E}_{P,\infty} \sum_{k=1}^{ K(\tau )}  \omega_k \sum_{j= K^{-1}(k) }^{\tau } \gamma_j \prod_{i = j}^{\tau } L_i(k)  \\
        &:= \mathbb{E}_{P,\infty} \sum_{k=1}^{ K(\tau )}  \omega_k \sum_{j= K^{-1}(k) }^{\tau } \gamma_j \Lambda_{\tau }^{(j)}(k)\\ 
    & = \mathbb{E}_{P,\infty} \sum_{j=1}^{\tau } \sum_{k=1}^{K(j)}    \gamma_j \omega_k \Lambda_{\tau }^{(j)}(k)
        ~ =  ~ \sum_{j=1}^{\infty} \sum_{k=1}^{K(j)}    \gamma_j \omega_k  \mathbb{E}_{P,\infty}\Lambda_{\tau }^{(j)}(k)\mathbbm{1}(j \leq   \tau  ) \\
& = \sum_{j=1}^{\infty} \sum_{k=1}^{K(j)}    \gamma_j \omega_k  \mathbb{E}_{P,\infty}\mathbbm{1}(j \leq   \tau )\mathbb{E}_{P,\infty}\left[\Lambda_{\tau }^{(j)}(k)\mid \mathcal{F}_{j-1}\right] \\
& \leq \sum_{j=1}^{\infty} \sum_{k=1}^{K(j)}    \gamma_j \omega_k  \mathbb{E}_{P,\infty}\mathbbm{1}(j \leq   \tau) ~ = ~ \sum_{j=1}^{\infty} \mathbb{E}_{P,\infty}\mathbbm{1}(j \leq   \tau  )  ~ = ~ \mathbb{E}_{P,\infty} \tau ,
    \end{align*}
    as desired. Above, the sole inequality  follows since each $\Lambda^{(j)}(k)$ is an $e_j$-process, and the following equality invokes the definition of $\gamma_j$ for each $j$. 
 \end{proof}

 \begin{proof}[Proof of \cref{thm::upperbound_general_form} (Delay bounds for adaptive e-detectors)]
  We first prove the upper bound for the adaptive e-SR procedure in  \eqref{eq::delay_ASR}. Note that 
 \begin{align*}
     N_\ASR^* &= \inf \left\{n \geq 1 : \sum_{k=1}^{K(n)}  \omega_k \sum_{j=K^{-1}(k)}^n \gamma_j \prod_{i = j}^n L_i(k)\geq 1 /\alpha \right\}\\
     &\leq \inf \left\{n \geq 1 : \sum_{k=1}^{K(n)}  \omega_k \sum_{j=K^{-1}(k)}^n  \prod_{i = j}^n L_i(k)\geq 1 /\alpha \right\} \\
     &= \inf \left\{n \geq 1 : \sum_{j=1}^n \sum_{k=1}^{K(j)} \omega_k\prod_{i = j}^n L_i(k) \geq 1 /\alpha \right\}\\
     & \leq \min_{j \geq 1}\inf \left\{n \geq j : \sum_{k=1}^{K(j)} \omega_k\prod_{i = j}^n L_i(k) \geq 1 /\alpha \right\} \\
     & := \min_{j \geq 1} N_j,
 \end{align*}
 where the first inequality follows because  $\gamma_j \geq 1$ for each $j$.
 For each fixed $P \in \mathcal{P}$ and $Q \in \mathcal{Q}$, since $N_j \indep \mathcal{F}_{\nu}$ for all $j > \nu + m $, we have 
 \begin{align*}
     & \mathbb{E}_{P,\nu,Q} \left[[N_\ASR^* - \nu]_+ \mid \mathcal{F}_\nu \right] \\
      & \leq  \min_{j > \nu + m} \mathbb{E}_{P,\nu,Q} \left[N_j -\nu  \mid \mathcal{F}_\nu \right] ~~~~\text{(Since $N_j \geq j > \nu+m$)}\\
      & = \min_{j > \nu + m}  \mathbb{E}_{P,\nu,Q} \left[ \inf \left\{n \geq j : \sum_{k=1}^{K(j)}\omega_k\prod_{i = j}^n L_i(k) \geq 1 /\alpha \right\}\right] - \nu \\
      & \quad \quad \quad ~~~~\text{(since $N_j \indep \mathcal{F}_{\nu}, \forall j > \nu + m $)}\\ 
      & = \min_{j -\nu>  m}  \mathbb{E}_{0, Q} \left[ \inf \left\{n \geq j - \nu : \sum_{k=1}^{K(j)} \omega_k\prod_{i = j-\nu}^n L_i(k) \geq 1 /\alpha \right\}\right]\\
      & = \min_{j'> m}  \mathbb{E}_{0, Q} \left[ \inf \left\{n \geq j' : \sum_{k=1}^{ K(j' + \nu)} \omega_k\prod_{i = j'}^n L_i(k) \geq 1 /\alpha \right\}\right] \\
      & \quad \quad \quad ~~\text{(by setting $j' := j - \nu)$} \\
      & = \min_{j'> m}  \mathbb{E}_{0, Q} \left[ \inf \left\{n \geq  1 : \sum_{k=1}^{ K(j' + \nu)} \omega_k\prod_{i = 1}^{n} L_i(k) \geq 1 /\alpha \right\}\right] + j' - 1 \\
      & \leq \min_{j'> m}  \mathbb{E}_{0, Q} \left[ \inf \left\{n \geq  1 : \sum_{k=1}^{ K(j')} \omega_k\prod_{i = 1}^{n} L_i(k) \geq 1 /\alpha \right\}\right] + j' - 1 \\
        & = \min_{j'> m}  \mathbb{E}_{0, Q} N_{1/\alpha}(j') + j' - 1 ,
 \end{align*}
 where the second equality comes from the fact that the distribution of $X_{j-m}, X_{j-m+1}, \dots$ under $\mathbb{P}_{P,\nu, Q}$ is equal to the one of $X_{j-\nu-m}, X_{j-\nu-m+1}, \dots$ under $\mathbb{P}_{0, Q}$ provided by $j-m > \nu$ and the forth equality is based on the strong stationarity of post-change observations. Since the last term does not depend on $P$, $\nu$ nor $\mathcal{F}_\nu$, we have the claimed result:
 \begin{align*}
 \mathcal{J}_P(\st_\ASR^*) \leq \mathcal{J}_L(\st_\ASR^*)= \sup_{P \in \mathcal{P}, \nu \geq 0} \esssup  \mathbb{E}_{P,\nu, Q} \left[[\st_\ASR^* - \nu]_+ \mid \mathcal{F}_\nu \right] \leq  \min_{j > m}  \left[\mathbb{E}_{0,Q} N_{1/\alpha}(j) + j - 1\right],
 \end{align*}
 as desired.

 To prove the adaptive e-CUSUM procedure case in \eqref{eq::delay_ACS}, first note that 
  \begin{align*}
     N_\ACS^* &= \inf \left\{n \geq 1 : \sum_{k=1}^{K(n)}\omega_k  \max_{K^{-1}(k) \leq j \leq n} \gamma_j\prod_{i = j}^n L_i(k)   \geq c_\alpha \right\}\\
     &\leq \inf \left\{n \geq 1 :  \max_{j \leq n} \sum_{k=1}^{K(j)}\omega_k  \prod_{i = j}^n L_i(k)  \geq c_\alpha \right\}\\
     & \leq \min_{j \geq 1}\inf \left\{n \geq j : \sum_{k=1}^{K(j)} \omega_k\prod_{i = j}^n L_i(k) \geq c_\alpha \right\} \\
     & := \min_{j \geq 1} N_j,
 \end{align*}
 where the first inequality comes from $\gamma_j \geq 1$ with the fact  $K^{-1}(k) \leq j$ for any $k \leq K(j)$. The remaining part of the proof of \eqref{eq::delay_ACS}  follows the same argument used to obtain \eqref{eq::delay_ASR}.

 \end{proof}

\section{Remaining Proofs}


\subsection{Proofs for statements in Section~\ref{sec::mixture_simple_exponential}} \label{appen::proofs_for_mixture_simple_exponential}

\begin{proof}[Proof of \cref{prop::nonneg_div}]
To simplify notation, we drop the subscripts $\{0,Q\}$ and $1$. The claim reduces to proving that
\begin{equation}
      \mathbb{E} \log L^{(\lambda^\op)} 
       = \psi^*\left(\Delta^\op\right) \sigma^2,
\end{equation}
where $\log L^{(\lambda^\op)} :=   \lambda^\op s(X) - \psi(\lambda^\op) v(X)$ and $\Delta^\op = \nabla \psi(\lambda^\op)$. To prove the equality, first note that from the definition of the convex conjugate $\psi^*$ of $\psi$, we have  
\begin{equation}
    \psi^*(\Delta^\op) := \sup_{\lambda \in \Pi} \left\{\lambda \Delta^\op - \psi(\lambda)\right\} = \lambda^\op \Delta^\op - \psi(\lambda^\op), 
\end{equation}
which implies
\begin{align*}
    \mathbb{E} \log L^{(\lambda^\op)} = \lambda^\op \mu - \psi(\lambda^\op) \sigma^2
    =  \sigma^2 \left[\lambda^\op \Delta^\op - \psi(\lambda^\op) \right] 
    = \sigma^2 \psi^*(\Delta^\op),
\end{align*}
as desired. 
\end{proof}

\begin{proof}[Proof of \cref{thm::upper_bound_well_sep_general}]
We first recall the definition of the stopping time in \eqref{eq::glr_st_main}:
\begin{equation} 
        \bar{N}_g:= \inf\left\{n \geq 1: \sup_{\lambda \in (\lambda_L, \lambda_U)} \sum_{i = 1}^n \log L_i^{(\lambda)} \geq g \right\}, \quad g > 0.
    \end{equation}
   Then, the same argument used in the proof  \cref{prop::upperbound_simple_form} immediately implies that, i the post-change observations are i.i.d.\ from $Q$, then, for any $g>0$,  
    \begin{equation}
        \mathbb{E}_{0,Q} \bar{N}_g 
        \leq   \frac{g}{D(Q||\mathcal{P})} +  \frac{\mathbb{V}_{0,Q} \left[\log L_1^{(\lambda^\op)}\right]}{\left[D(Q||\mathcal{P})\right]^2} + 1 .
    \end{equation}

The claim of the theorem follows from Lemma~\ref{lemma::alg_1_result}, whose statement and proof are given below.
\end{proof}

 \begin{lemma} \label{lemma::alg_1_result}
Let $N_{1/\alpha}$ and $N_{c_\alpha}$ be stopping times where the underlying mixing weights $\{\omega_k\}$ and parameters of baseline increments $\{\lambda_k\}$ are chosen via~\cref{alg::compute_base}. Let $\bar{N}_{g_\alpha}$ be the stopping time defined in \eqref{eq::glr_st_main} with the threshold given by~\cref{alg::compute_base}. Then, for any stream of observations $X_1, X_2,\dots$, 
\begin{equation}
    N_{c_\alpha} \leq N_{1/\alpha} \leq \bar{N}_{g_\alpha},
\end{equation}
deterministically, provided that $1 < c_\alpha < 1/\alpha $. 
 \end{lemma}

\begin{proof}[Proof of \cref{lemma::alg_1_result} and \cref{alg::compute_base}]
Throughout this proof, we set $D_L := \psi^*\left(\Delta_L\right) < \psi^*\left(\Delta_U\right) =: D_U$. 
The first inequality $N_{c_\alpha} \leq N_{1/\alpha}$ follows directly from the definition of the stopping time in \eqref{eq::N_sr_sub_psi} along with the condition that $c_\alpha \leq 1/\alpha$. To prove the second inequality $N_{1/\alpha} \leq \bar{N}_{g_\alpha}$, we will exploit on general geometric construction introduced in \citep{shin2021nonparametric} to analyze the performance of sequential generalized likelihood ratio tests. To that effect,  set $\hat{\mu}_n := S_n / V_n$. Then, for each fixed $\lambda > 0$ such that  $\Delta = \nabla\psi(\lambda)$, \cref{prop::nonneg_div} and the identity $\lambda = \nabla \psi^*(\Delta)$ imply that the function $\hat{\mu}_n \mapsto V_n^{-1}\sum_{i=1}^n\log L_i^{(\lambda)} = \lambda \hat{\mu}_n - \psi(\lambda) = \lambda\left(\hat{\mu}_n - \Delta\right) + \psi^*(\Delta)$  is a mapping of $\hat{\mu}_n$ into the tangent line of the function $z\mapsto \psi^*(z)$ at $z = \Delta$. 
Next, set $V_U := g_\alpha / D_U $ and $V_L := g_\alpha / D_L$, and define the set 
\begin{equation}
    R := \left\{(z,y) \in [0,\infty)^2 : y \leq \psi^*(z)\right\}.
\end{equation}
Then, the stopping event of $\bar{N}_{g_\alpha}$ can be expressed as \begin{align*}
         &\left\{\exists n \geq 1: \sup_{\lambda \in (\lambda_L, \lambda_U)} \sum_{i = 1}^n \log L_i^{(\lambda)} \geq g_\alpha \right\}  \\
         & =\left\{\exists n \geq 1: V_n < V_U, \left(\hat{\mu}_n, \frac{g_\alpha}{V_n}\right) \in H(\Delta_U) \right\} \cup \left\{ \exists n \geq 1: V_n \geq V_U, \left(\hat{\mu}_n, \frac{g_\alpha}{V_n}\right) \in R \setminus H(\Delta_L) \right\} \\
         &= \left\{\exists n \geq 1 :V_n < V_U,  \sum_{i = 1}^n \log L_i^{(\lambda_U)} \geq g_\alpha \right\}  \cup \left\{\exists n \geq 1: V_n \geq V_U, \sup_{\lambda > \lambda_L} \sum_{i = 1}^n \log L_i^{(\lambda)} \geq g_\alpha \right\},
   \end{align*}
 $H(\Delta_U)$ and $H(\Delta_L)$ are half spaces contained in and tangent to $R$ at $\left(\Delta_U, g_\alpha / V_U\right)$ and $\left(\Delta_L, g_\alpha / V_L\right)$, respectively. See \cref{fig::GLR-CP} for an illustration of the stopping event of $\bar{N}_{g_\alpha}$.
   
   \begin{figure}
    \begin{center}
    \includegraphics[scale =  0.75]{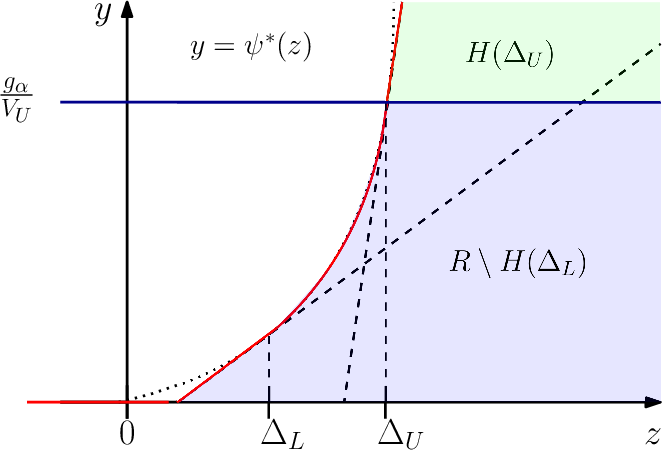}
    \end{center}
    \caption{Illustration of the stopping event of $\bar{N}_{g_\alpha}$ defined in \eqref{eq::glr_st_main}, and related regions $H(\Delta_U), H(\Delta_L)$ and  $R$. The stopping time $\bar{N}_{g_\alpha}$ is the first time when $(\hat{\mu}_n, g_\alpha / V_n)$ is located in one of the colored areas. } 
    \label{fig::GLR-CP}
\end{figure}
   
Note that the first decomposition part $\left\{\exists n \geq 1 : V_n < V_U,\sum_{i = 1}^n \log L_i^{(\lambda_U)} \geq g_\alpha \right\}$ is nonempty only if $V_U > \min_{x}v(x) := v_{\min}$, which is equivalent to $g_\alpha > v_{\min} D_U$. For the second part, a straightforward extension of Lemma 1 in the appendix of \citep{shin2021nonparametric} implies that, for any fixed $\eta > 1$, the second part can be further decomposed by sets of simple events as follows:
\begin{equation}
\begin{aligned}
    \left\{\exists n \geq 1: V_n \geq V_U, \sup_{\lambda > \lambda_L} \sum_{i = 1}^n \log L_i^{(\lambda)} \geq g_\alpha \right\}  
    & \subset \bigcup_{k=1}^{K(\eta) - 1}\left\{\exists n \geq 1: V_n \in \left[V_U \eta^{k-1}, V_U \eta^k \right), \sum_{i = 1}^n \log L_i^{(\lambda_k)} \geq g_\alpha / \eta \right\} \\
    &\qquad\cup \left\{\exists n \geq 1 : V_n \geq V_U \eta^{K(\eta) - 1},\sum_{i = 1}^n \log L_i^{(\lambda_{K(\eta)})} \geq g_\alpha / \eta  \right\},
\end{aligned}
\end{equation}
where $K(\eta)$ is a positive integer defined by
\begin{equation}
    K(\eta) := \left\lceil \log_{\eta} \left( \frac{D_U}{ D_L}\right)\right\rceil,
\end{equation}
 and, for $k = 1, \dots, K(\eta)-1$,  $\lambda_k$ is given by $\lambda_k := \psi^*\left(\Delta_k\right)$, with  $\Delta_k$ the solution with respect to $z > 0$ of the  equation 
\begin{equation} \label{eq::def_of_mu_k_appen}
    \psi^*(z)= \frac{D_U}{\eta^k},
\end{equation}
 while  $\Delta_{K(\eta)}:= \Delta_L$.
It can be checked that $ \lambda_U := \lambda_0 >\lambda_1 > \lambda_2 > \cdots > \lambda_{K(\eta)} = \lambda_L$. Decomposing the stopping event of $\bar{N}_{g_\alpha}$, we can lower bound the stopping time $\bar{N}_{g_\alpha}$ for any $\eta > 1$ as:
\begin{align*}
    \bar{N}_{g_\alpha}&= \inf\left\{n \geq 1: \sup_{\lambda \in (\lambda_L, \lambda_U)} \sum_{i = 1}^n \log L_i^{(\lambda)} \geq g_\alpha \right\} \\
    & \geq \inf \left\{n \geq 1: e^{-g_\alpha}  \prod_{i = 1}^n  L_i^{(\lambda_0)}\mathbbm{1}\left(g_\alpha > v_{\min}D_U\right)+\sum_{k=1}^{K(\eta)} e^{-g_\alpha/\eta}\prod_{i=1}^n L_{i}^{(\lambda_k)} \geq 1 \right\} \\
    & = \inf \Bigg\{n \geq 1: \alpha^{-1}e^{-g_\alpha} \mathbbm{1}\left(g_\alpha > v_{\min}D_U\right)\exp\left\{\lambda_0 S_n - \psi(\lambda_0) V_n\right\} \\ 
    &\left.\qquad\qquad\qquad\qquad+\sum_{k=1}^{K(\eta)} \alpha^{-1}e^{-g_\alpha/\eta} \exp\left\{\lambda_k S_n - \psi(\lambda_k) V_n\right\} \geq 1/\alpha \right\} \\ 
    &:= \inf \left\{n\geq 1 : \sum_{k = 0}^{K(\eta)}  \omega_k(\eta) \exp\left\{\lambda_k S_n - \psi(\lambda_k) V_n\right\}  \geq 1/\alpha \right\} ~ ~ ~ ~ := N(\eta),
\end{align*}
where $\omega_0(\eta) := \alpha^{-1}e^{-g_\alpha} \mathbbm{1}\left(g_\alpha > v_{\min}D_U\right)$ and $\omega_k(\eta):=\alpha^{-1}e^{-g_\alpha/\eta}$ for each $\eta > 1$ and $k = 1, \dots, K(\eta)$. As a quick remark, note that  $\omega_0(\eta)$ does not depend on $\eta$ and each $\omega_k(\eta)$ in fact does not depend on the index $k$ but we use this notation just for consistency. 

\begin{algorithm}
\DontPrintSemicolon

\KwInput{ARL parameter $\alpha \in (0,1)$, Boundary values $0<\Delta_L < \Delta_U$,\newline Maximum number of baselines $K_{\max} \in \mathbb{N}$.}
\KwOutput{Parameters of baseline increments $\lambda_U = \lambda_0 >  \lambda_1 > \cdots > \lambda_{K_{\alpha}} = \lambda_L$, \newline Mixing weights $\omega_0, \omega_1, \dots, \omega_{K_\alpha} \in [0,1]$, \newline Auxiliary values used to compute baseline increments $\left\{g_\alpha, K_\alpha, \eta_\alpha, W\right\}$.}

Compute parameters for boundary values by $\lambda_L := \nabla \psi^* (\Delta_L)$ and $\lambda_U := \nabla \psi^* (\Delta_U)$. 

\tcc{If the separation is large enough, use a single baseline.}
\If{$\log(1/\alpha) \leq v_{\min}\psi^*(\Delta_L)$}{
Set $K_\alpha := 1$, $\lambda_1 := \lambda_L$ and $\omega_1 := 1$ \;
\Return Parameter $\lambda_1$ and mixing weight $\omega_1$ \;
}

Compute the threshold $g_\alpha > \log(1/\alpha)$ given by
\begin{equation} \label{eq::g_alpha_explicit_main}
    g_\alpha:= \inf \left\{g > \log(1/\alpha): e^{-g} \mathbbm{1}(g > v_{\min}D_U) + \min_{k \in [K_{\max}]} k \exp\left\{-g\left(\frac{D_U}{D_L}\right)^{-1/k}\right\} \leq \alpha\right\}.
\end{equation}
(See \cref{alg::compute_threshold} for an explicit way to compute it) \;
Compute the number of baselines $K_\alpha \in \mathbb{N}$ by
\begin{equation} \label{eq::k_alpha}
    K_\alpha = \argmin_{k \in [K_{\max}]} k \exp\left\{-g_\alpha \left(\frac{D_U}{ D_L}\right)^{-1/k}\right\},
\end{equation}
where $D_L := \psi^*(\Delta_L) < \psi^*(\Delta_U) =: D_U$.\; \label{line::def_K_alpha_in_alg}
Compute the spacing parameter $\eta_\alpha :=  \left(\frac{D_U}{ D_L}\right)^{1/K_{\alpha}} $.\;

\tcc{Compute parameters of baseline increments and mixing weights}
Set $\lambda_0 := \lambda_U$ and $\lambda_{K_\alpha} := \lambda_L$.\;
\If{$K_\alpha \geq 2$}{
\For{$k = 1, \dots, K_\alpha - 1$}{
Compute $\Delta_k$ as the solution of the equation $\psi^*\left(z\right)= D_U \eta^{-k}$
with respect to $z > 0$.\;
Compute the $k$-th parameter as $\lambda_k := \nabla \psi^*\left(\Delta_k\right)$. 
}
}
Set $W:=e^{-g_\alpha}\mathbbm{1}\left(g_\alpha > v_{\min}D_U\right) + K_\alpha e^{-g_\alpha/\eta_\alpha}$  and compute mixing weights by 
\begin{equation}
\omega_0 = W^{-1} e^{-g_\alpha}\mathbbm{1}\left(g_\alpha > v_{\min}D_U\right)~~\text{and}~~\omega_k = W^{-1} e^{-g_\alpha/\eta_\alpha},~~ \forall k \in [K_\alpha].
\end{equation}
 \;

\Return $\left\{ \lambda_0, \lambda_1, \dots, \lambda_{K_{\alpha}}\right\}$, $\left\{\omega_0,\omega_1, \dots, \omega_{K_\alpha}\right\}$, $\left\{ g_\alpha, K_\alpha, \eta_\alpha, W\right\}$ \;
 \caption{ Pseudo-code of \texttt{computeBaseline} function}
 \label{alg::compute_base}
\end{algorithm}

Finally, set $\eta_\alpha := \left(\frac{D_U}{D_L}\right)^{1/K_\alpha}$ where $K_\alpha$ is the integer defined in \eqref{eq::k_alpha} of \cref{alg::compute_base}. It can be easily checked that $K(\eta_\alpha) = K_\alpha$. Therefore, once we choose $\eta = \eta_\alpha$ then, from the definition of the mixing weights in \cref{alg::compute_base}, we conclude that  $N(\eta_\alpha) \geq N_{1/\alpha}$ provided that
\begin{equation} \label{eq::validity_of_alg1}
    W := e^{-g_\alpha}\mathbbm{1}\left(g_\alpha > v_{\min}D_U\right) + K_\alpha e^{-g_\alpha / \eta_\alpha} \leq \alpha,
\end{equation}
where the constant $g_\alpha$ is the constant used in \cref{alg::compute_base} and given by 
\begin{equation} \label{eq::g_alpha_explicit}
    g_\alpha:= \inf \left\{g > \log(1/\alpha): e^{-g} \mathbbm{1}(g > v_{\min}D_U) + \min_{k \in [K_{\max}]} k \exp\left\{-g\left(\frac{D_U}{D_L}\right)^{-1/k}\right\} \leq \alpha\right\}.
\end{equation}
By the definition of  $ K_\alpha$, we can immediately check  that the inequality~\eqref{eq::validity_of_alg1} holds, which proves the claimed inequality $N_{1/\alpha} \leq N(\eta_\alpha) \leq \bar{N}_{g_\alpha}$, as desired. 

\end{proof}

We conclude this section with a formal proof of the validity of the upper bound in \eqref{eq::g_upper_explicit_main_sec3}.

 \begin{proposition}\label{prop::upper_bound_on_g}
      Let  $K_{\max}$ be a large enough integer such that
      
       \begin{equation}\label{eq:Kmax.large.enough}
            K_\alpha =  \argmin_{k \in \{1,\ldots, K_{\max}\}} k \exp\left\{-g_\alpha \left(\frac{D_U}{D_L}\right)^{-1/k}\right\} = \argmin_{k \in \mathbb{N}} k \exp\left\{-g_\alpha \left(\frac{D_U}{D_L}\right)^{-1/k}\right\}.
    \end{equation}
      Then, the quantity $g_\alpha$ \cref{alg::mixture_SR} specified in \eqref{eq::g_alpha_explicit_main} is such that
        \begin{equation} \label{eq::g_upper_explicit_appen}
            g_\alpha < \inf_{\eta > 1} \eta \left[\log(1/\alpha) + \log\left(1 + \left\lceil \log_\eta \frac{\psi^*(\Delta_U)}{\psi^*(\Delta_L)}\right\rceil \right)\right] .
        \end{equation}
    \end{proposition}

\begin{proof}[Proof of \cref{prop::upper_bound_on_g}]

Once $K_{\max}$ is large enough to satisfy \eqref{eq:Kmax.large.enough} then the constant $g_\alpha$ can be written as
        \begin{equation} \label{eq::g_alpha_for_large_K}
            g_\alpha = \inf \left\{g > \log(1/\alpha): e^{-g} \mathbbm{1}(g > v_{\min}D_U) + \inf_{\eta >1} \left\lceil\log_\eta \left(\frac{D_U}{D_L}\right)\right\rceil e^{-g / \eta} \leq \alpha\right\},
        \end{equation}
because the following equality holds for each $g > 0$ and $D_U > D_L$. 
\begin{equation}
    \min_{k \in \mathbb{N}} k \exp\left\{-g\left(\frac{D_U}{D_L}\right)^{-1/k}\right\} = \inf_{\eta >1} \left\lceil\log_\eta \left(\frac{D_U}{D_L}\right)\right\rceil e^{-g / \eta}. 
\end{equation}
Finally, to prove the claimed bound in \eqref{eq::g_upper_explicit_appen}, first note that from \eqref{eq::g_alpha_for_large_K}, we have $g_\alpha \leq g(\eta)$ where $g(\eta) > 0$ is given by
\begin{align*}
    g(\eta) &:= \inf\left\{g>\log(1/\alpha): e^{-g} + \left\lceil\log_\eta \left(\frac{D_U}{D_L}\right)\right\rceil e^{-g / \eta} \leq \alpha \right\} \\    
    & \leq \eta \left[\log(1/\alpha) + \log\left(1 + \left\lceil \log_\eta \frac{D_U}{D_L}\right\rceil \right)\right],
\end{align*}
for each $\eta > 1$. By taking infimum over $\eta > 1$, we get the upper bound in \eqref{eq::g_upper_explicit_appen}, as desired.
\end{proof}

\begin{proof}[Proof of \cref{lem::condition_for_SR_to_glrt_general}]
The proof of \cref{lem::condition_for_SR_to_glrt_general} is similar to the one of \cref{lemma::alg_1_result} except the previous threshold $g_\alpha$ being replaced with $g\left(V_0 \eta^{K(j)} \right)$. Also note that, in this proof, the terms $\Delta_0$ and $V_0$ play a similar role of  $\Delta_U$ and $V_U$ in the previous proof of \cref{lemma::alg_1_result}.

As same as the \cref{lemma::alg_1_result} case,  set $S_n := \sum_{i=1}^n s(X_i)$, $V_n := \sum_{i=1}^n v(X_i)$, and $\hat{\mu}_n := S_n / V_n$. Then, for each fixed $\lambda > 0$ with $\Delta = \nabla\psi(\lambda)$, the function $\hat{\mu}_n \mapsto V_n^{-1}\sum_{i=1}^n\log L_i^{(\lambda)} = \lambda \hat{\mu}_n - \psi(\lambda) = \lambda\left(\hat{\mu}_n - \Delta\right) + \psi^*(\Delta)$  is a mapping of $\hat{\mu}_n$ to the tangent line of the function $z\mapsto \psi^*(z)$ at $z = \Delta$. Now, since the boundary function $g$ is non-decreasing, the stopping event of $\bar{N}_g(j)$ can be bounded as
\begin{align*}
         &\left\{\exists n \geq 1: \sup_{\lambda \in (\lambda_{K(j)}, \lambda_0)} \sum_{i = 1}^n \log L_i^{(\lambda)} \geq g\left(V_0 \eta^{K(j)} \right) \right\}  \\
        & \subset\left\{\exists n \geq 1: V_n < V_0,  \sum_{i = 1}^n \log L_i^{(\lambda_0)} \geq g\left(V_0\right)  \right\} \\
         &\qquad\cup \left\{ \exists n \geq 1: V_n \in \left[V_0, V_0 \eta^{K(j)}\right), \sup_{\lambda \in (\lambda_{K(j)}, \lambda_0)} \sum_{i = 1}^n \log L_i^{(\lambda)} \geq g\left(V_n \right)  \right\} \\
         &\qquad\qquad\cup \left\{ \exists n \geq 1: V_n \geq V_0 \eta^{K(j)}, \sum_{i = 1}^n \log L_i^{(\lambda_{K(j)})} \geq g\left(V_0 \eta^{K(j)} \right)  \right\} \\
         & =\left\{\exists n \geq 1: V_n < V_0, \left(\hat{\mu}_n, \frac{g\left(V_0\right)}{V_n}\right) \in H(\Delta_0) \right\} \\
         &\qquad\cup \left\{ \exists n \geq 1: V_n \in \left[V_0, V_0 \eta^{K(j)}\right), \left(\hat{\mu}_n, \frac{g(V_n)}{V_n}\right) \in R \right\} \\
         &\qquad\qquad\cup \left\{ \exists n \geq 1: V_n \geq V_0 \eta^{K(j)}, \left(\hat{\mu}_n, \frac{g\left(V_0 \eta^{K(j)}\right)}{V_n}\right) \in H(\Delta_{K(j)}) \right\} \\
         &= \left\{\exists n \geq 1 :V_n < V_0,  \sum_{i = 1}^n \log L_i^{(\lambda_0)} \geq g\left(V_0\right) \right\} \\
         &\qquad\cup \left\{\exists n \geq 1:  V_n \in \left[V_0, V_0 \eta^{K(j)}\right), \sup_{\lambda \in (\lambda_{K(j)}, \lambda_0)} \sum_{i = 1}^n \log L_i^{(\lambda)} \geq g(V_n) \right\}\\
         &\qquad\qquad\cup \left\{\exists n \geq 1:  V_n \geq  V_0 \eta^{K(j)},  \sum_{i = 1}^n \log L_i^{(\lambda_{K(j)})} \geq g\left(V_0 \eta^{K(j)}\right) \right\},
   \end{align*}
where  $V_0 := \inf\left\{t \geq 1: D_0 \geq g(t) / t\right\}$ and the set  $R$ is defined by
\begin{equation}
    R := \left\{(z,y) \in [0,\infty)^2 : y \leq \psi^*(z)\right\},
\end{equation}
and $H(\Delta_0)$ and $H(\Delta_{K(j)})$ are half spaces contained in and tangent to $R$ at $\left(\Delta_0, \frac{g(V_0)}{V_0}\right)$ and $\left(\Delta_L, \frac{g\left(V_0 \eta^{K(j)}\right)}{V_0 \eta^{K(j)}}\right)$, respectively. See \cref{fig::GLR-CP-no-sep} for an illustration of the upper bound of the stopping event of $\bar{N}_g(j)$.
   
   \begin{figure}
    \begin{center}
    \includegraphics[scale =  0.75]{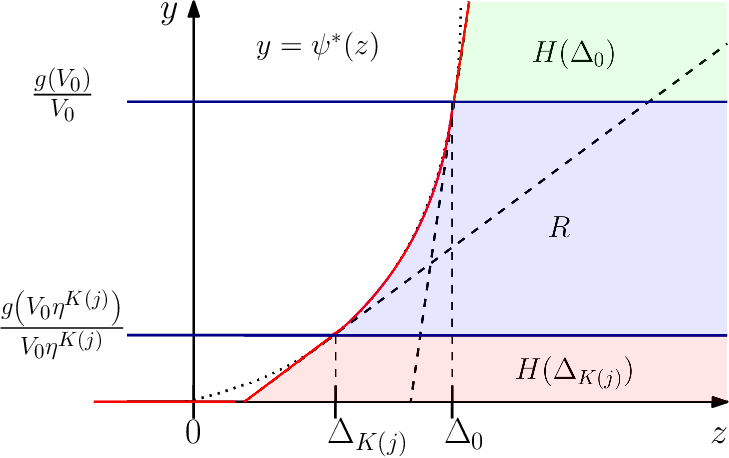}
    \end{center}
    \caption{Illustration of the stopping event of $\bar{N}_g(j)$ defined in \eqref{eq::N_G_j}, and related regions $H(\Delta_0), H(\Delta_{K(j)})$ and  $R$. The stopping time $\bar{N}_g(j)$ is the first time when $\left(\hat{\mu}_n, g\left(\overline{V}_n^{(j)} \right)/ V_n\right)$ is located in one of the colored areas, where $\overline{V}_n^{(j)} := \max\left\{V_0, \min\left\{V_n, V_0 \eta^{K(j)}\right\}\right\}$ } .
    \label{fig::GLR-CP-no-sep}
\end{figure}
   
Note that the first decomposition part $\left\{\exists n \geq 1 : V_n < V_0,\sum_{i = 1}^n \log L_i^{(\lambda_0)} \geq g(V_0) \right\}$ is nonempty only if $V_0 > v_{\min}$, which is equivalent to $g(V_0)> v_{\min} D_0$. For the second part, a straightforward extension of Lemma 1 in the appendix of \citep{shin2021nonparametric} implies that, for any fixed $\eta > 1$, the second part can be further decomposed by sets of simple events as follows:
\begin{equation}
\begin{aligned}
    & \left\{\exists n \geq 1:  V_n \in \left[V_0, V_0 \eta^{K(j)}\right), \sup_{\lambda \in (\lambda_{K(j)}, \lambda_0)} \sum_{i = 1}^n \log L_i^{(\lambda)} \geq g(V_n) \right\} \\
    & \subset \bigcup_{k=1}^{K(j)}\left\{\exists n \geq 1: V_n \in \left[V_0 \eta^{k-1}, V_0 \eta^k \right), \sum_{i = 1}^n \log L_i^{(\lambda_k)} \geq g\left(V_0 \eta^{k}\right)/ \eta \right\} ,
\end{aligned}
\end{equation}
where each $\lambda_k$ is given by $\lambda_k := \psi^*\left(\Delta_k\right)$ and  each $\Delta_k$ is the solution of the equation~\eqref{eq::const_of_SR_no_sep}
for $k = 1, \dots, K(j)$.

From this decomposition of the stopping event of $\bar{N}_g(j)$, for any fixed $\eta > 1$, we can lower bound the stopping time $\bar{N}_g(j)$ as follows:
\begin{align*}
    \bar{N}_g(j)&= \inf\left\{ n \geq 1: \sup_{\lambda \in (\lambda_{K(j)}, \lambda_0)} \sum_{i = 1}^n \log L_i^{(\lambda)} \geq g\left(V_0 \eta^{K(j)} \right) \right\} \\
    & \geq \inf \left\{n \geq 1: e^{-g(V_0)}  \prod_{i = 1}^n  L_i^{(\lambda_0)}\mathbbm{1}\left(g(V_0) > v_{\min}D_0\right)+\sum_{k=1}^{K(j)} e^{-g\left(V_0 \eta^k\right)/\eta}\prod_{i=1}^n L_{i}^{(\lambda_k)} \geq 1 \right\} \\
    & = \inf \left\{n \geq 1:  \sum_{k=0}^{K(j)} \omega_k\prod_{i=1}^n L_{i}^{(\lambda_k)} \geq 1/\alpha \right\} \\
    &= N_{1/\alpha}(j),
\end{align*}
which proves the claimed inequality. 
\end{proof}

\begin{proof}[Proof of \cref{cor::explicit_upper_no_sep}]
From  \cref{thm::upperbound_general_form} and \cref{lem::condition_for_SR_to_glrt_general}, worst average delays of adaptive e-SR and e-CUSUM procedures are upper bounded by $\max_{\nu \geq 0} \min_{j \geq 1} \left[ \mathbb{E}_{0,Q} N_{G} (j+\nu) + j - 1\right]$ where $\bar{N}_g(j)$ is a stopping time defined
        by
        \begin{equation}
            \bar{N}_g(j) := \inf\left\{n \geq 1 : \sup_{\lambda \in (\lambda_{K(j)}, \lambda_0)} \sum_{i = 1}^n \log L_i^{(\lambda)} \geq g\left(V_0 \eta^{K(j)} \right)\right\},
        \end{equation}
for each $j \geq 1$. Now, let us first consider the case $\lambda^\op \geq \lambda_0$. In this case, we use the following simple upper bound:
\begin{align*}
     \min_{j \geq 1} \left[ \mathbb{E}_{0,Q} N_{G} (j) + j - 1\right]
    & \leq  \mathbb{E}_{0,Q} N_{G} (1) \\
    & = \mathbb{E}_{0,Q}\inf\left\{n \geq 1 : \sup_{\lambda \in (\lambda_{K(1)}, \lambda_0)} \sum_{i = 1}^n \log L_i^{(\lambda)} \geq g\left(V_0 \eta^{K(1)} \right)\right\} \\
    & \leq \mathbb{E}_{0,Q}\inf\left\{n \geq 1 :  \sum_{i = 1}^n \log L_i^{(\lambda_0)} \geq g_{r\alpha}\right\}.
\end{align*}
Since $\mathbb{E}_{0,Q}\log L_1^{(\lambda_0)} = \sigma^2\left(\lambda_0 \Delta^\op - \psi(\lambda_0)\right) \geq \sigma^2\left(\lambda_0 \Delta_0 - \psi(\lambda_0)\right) = \sigma^2 \psi^*\left(\Delta_0\right)$, by the same argument of \cref{prop::upperbound_simple_form}, the last term above can be further upper bounded by
\begin{equation}
     \frac{g_{r\alpha}}{D(Q||\mathcal{P})}\frac{\psi^*\left(\Delta^\op\right)}{\psi^*\left(\Delta_0\right)} +  \frac{\mathbb{V}_{0,Q} \left[\log L_1^{(\lambda_0)}\right]}{\left[D(Q||\mathcal{P})\right]^2}\left[\frac{\psi^*\left(\Delta^\op\right)}{\psi^*\left(\Delta_0\right)}\right]^2 + 1.
\end{equation}
Since we are in the case $\lambda^\op \geq \lambda_0$, we have $\frac{\psi^*\left(\Delta^\op\right)}{\psi^*\left(\Delta_0\right)} \geq 1$, which can be understood as a measure of inefficiency due to the misspecified upper bound of the oracle $\lambda^\op$.

Now, consider the case $\lambda^\op < \lambda_0$ where we correctly specified the upper bound. In this case, let $j^\op$ be the smallest integer satisfying $\lambda_{K(j^\op)} := \lambda_{K^\op} < \lambda^\op < \lambda_0$. Then, we can further upper bound the worst average delays by 
\begin{align*}
     \min_{j \geq 1} \left[ \mathbb{E}_{0,Q} N_{G} (j) + j - 1\right] 
    & \leq \mathbb{E}_{0,Q}\inf\left\{n \geq 1 :  \sum_{i = 1}^n \log L_i^{(\lambda^\op)} \geq g\left(V_0 \eta^{K^\op} \right)\right\} + j^\op - 1 \\
    &:= \mathbb{E}_{0,Q}N_\op+ j^\op - 1.
\end{align*}
By \cref{eq::glr_bound_general}, we have the following intermediate upper bound on the worst average delays,
    \begin{equation}
        \mathbb{E}_{0,Q}N_\op+ j^\op - 1 \leq   \frac{g\left(V_0 \eta^{K^\op} \right)}{D(Q||P)} +  \frac{\mathbb{V}_{0,Q} \left[\log L_1^{(\lambda^\op)}\right]}{\left[D(Q||P)\right]^2} + j^\op.    
    \end{equation}
Note that if $j^\op = 1$ then $\lambda_{K(1)} = \lambda_L < \lambda^\op$. Thus, in this case, we also correctly specified the lower bound, and the above bound is reduced to the same upper bound on the worst average delays in \cref{thm::upper_bound_well_sep_general} of the well-separation case except the ARL parameter $\alpha$ being replaced by $r\alpha$. 

Finally, to get an explicit upper bound on $j^\op$ for the case $j^\op > 1$, fist note that, from the definition of  $\lambda_{K(j^\op - 1)}$ with the fact $ \lambda^\op < \lambda_1^{(j^\op - 1)} \Leftrightarrow \Delta^\op < \Delta_{K(j^\op -1)}$, we have 
    \begin{equation}
       \frac{g\left(V_0 \eta^{K(j^\op - 1)}\right)}{V_0 \eta^{K(j^\op - 1)}} = \psi^*\left(\Delta_{K(j^\op -1)}\right)  > \psi^*\left(\Delta^\op\right).
    \end{equation}
    Also, the condition $K(j) \geq K_L + m \log_\eta j$ implies
    \begin{equation}
        j \leq \left[\frac{\eta^{-K_L}}{V_0} V_0 \eta^{K(j)}\right]^{1/m},
    \end{equation}
    for each $j \geq 1$. By combining two inequalities above, we have
    \begin{align*}
        j^\op - 1 &\leq \left[\frac{\eta^{-K_L}}{V_0} V_0 \eta^{K(j^\op - 1)}\right]^{1/m} ~
         < ~ \left[\frac{1}{V_0\eta^{K_L}} \frac{g\left(V_0 \eta^{K(j^\op - 1)}\right)}{\psi^*\left(\Delta^\op\right)}\right]^{1/m} 
         \leq  \left[\frac{\psi^*\left(\Delta_L\right)}{\psi^*\left(\Delta^\op\right)} \frac{g\left(V_0 \eta^{K^\op}\right)}{g_{r\alpha}}\right]^{1/m}.
    \end{align*}
In sum, by combining all bounds above, we have
\begin{align*} 
&\min_{j \geq 1} \left[ \mathbb{E}_{0,Q} N_{G} (j) + j - 1\right]\\
&\leq
    \begin{cases} 
     \frac{g_{r\alpha}}{D(Q||\mathcal{P})}\frac{\psi^*\left(\Delta^\op\right)}{\psi^*\left(\Delta_0\right)} +  \frac{\mathbb{V}_{0,Q} \left[\log L_1^{(\lambda_0)}\right]}{\left[D(Q||\mathcal{P})\right]^2}\left[\frac{\psi^*\left(\Delta^\op\right)}{\psi^*\left(\Delta_0\right)}\right]^2 + 1 &\mbox{if } \lambda^\op\geq \lambda_0 \\
   \frac{g_{r\alpha}}{D(Q||P)} +  \frac{\mathbb{V}_{0,Q} \left[\log L_1^{(\lambda^\op)}\right]}{\left[D(Q||P)\right]^2} + 1    & \mbox{if } \lambda^\op \in (\lambda_L, \lambda_0) \\
    \frac{g_{r\alpha} + s\eta\log\left(1 + K^\op - K_L\right)}{D(Q||P)} +  \frac{\mathbb{V}_{0,Q} \left[\log L_1^{(\lambda^\op)}\right]}{\left[D(Q||P)\right]^2} +  \left[\frac{\psi^*\left(\Delta_L\right)}{\psi^*\left(\Delta^\op\right)} \frac{g_{r\alpha} + s\eta\log\left(1 + K^\op - K_L\right)}{g_{r\alpha}}\right]^{1/m}   &\mbox{if } \lambda^\op\leq \lambda_L
    \end{cases},
\end{align*}
as desired. 
\end{proof}

\section{An explicit way to compute the threshold in Algorithm~\ref{alg::compute_base}}

The following pseudo-code describes how to compute the threshold $g_\alpha$ defined by
\begin{equation} 
    g_\alpha:= \inf \left\{g > \log(1/\alpha): e^{-g} \mathbbm{1}(g > v_{\min}D_U) + \min_{k \in [K_{\max}]} k \exp\left\{-g\left(\frac{D_U}{D_L}\right)^{-1/k}\right\} \leq \alpha\right\}.
\end{equation}
\begin{algorithm}
\DontPrintSemicolon

\KwInput{ARL parameter $\alpha \in (0,1)$, Boundary values $0<\Delta_L < \Delta_U$,\newline Maximum number of baselines $K_{\max} \in \mathbb{N}$, Tolerance  $\epsilon > 0$.}
\KwOutput{Threshold $g_\alpha > 0$ that is defined and used in \cref{alg::compute_base}.}
Set $D_L := \psi^*(\Delta_L) < \psi^*(\Delta_U) =: D_U$ and define a function $f$ on $\mathbb{R}_+$ as 
\begin{equation}
f(g):= \min_{k \in [K_{\max}]} k \exp\left\{-g \left(\frac{D_U}{D_L}\right)^{-1/k}\right\}.
\end{equation}

\eIf{$f(v_{\min}D_U) \leq \alpha$}{
    Compute $g_\alpha := \inf\left\{g \in (\log(1/\alpha), v_{\min}D_U]: f(g) \leq \alpha\right\}$ by using the bisection method to the function $g \mapsto f(g) - \alpha$ with endpoints $\left\{\log(1/\alpha), v_{\min}D_U\right\}$ and tolerance $\epsilon$.\;
 }{
 Compute $g_\alpha := \inf\left\{g \in \left(v_{\min}D_U,\frac{D_U}{D_L}\log(2/\alpha)\right): e^{-g} + f(g) \leq \alpha\right\}$ by using the bisection method to the function $g \mapsto e^{-g} + f(g) - \alpha$ with endpoints $\left\{v_{\min}D_U,\frac{D_U}{D_L}\log(2/\alpha)\right\}$.
}
\Return $g_\alpha > 0 $ \;
 \caption{ Pseudo-code of \texttt{computeThreshold} function}
 \label{alg::compute_threshold}
\end{algorithm}

\end{document}